\documentclass{LMCS}

\def\doi{8 (1:24) 2012}
\lmcsheading%
{\doi}
{1--45}
{}
{}
{Jul.~\phantom04, 2011}
{Mar.~15, 2012}
{}

\newif\ifnocomment
\nocommenttrue
\newif\iflongversion
\longversionfalse
\longversiontrue

\usepackage{enumerate}
\usepackage{hyperref}
\usepackage[T1]{fontenc} 
\usepackage[utf8x]{inputenc}
\usepackage{alltt}
\usepackage{graphicx}
\usepackage{mathpartir}
\usepackage{url}
\usepackage{amsfonts,amssymb,amsmath,amsthm}
\usepackage{mathbbol}
\usepackage{stmaryrd}
\usepackage{empheq}
\usepackage{upgreek}
\usepackage{calrsfs}
\usepackage{color}
\usepackage{ifpdf}
\usepackage{tikz}
\usetikzlibrary{shapes,topaths,automata}
\newcommand{\eg}{\emph{e.g.}}
\newcommand{\ie}{\emph{i.e.}}
\newcommand{\recvar}{X}
\newcommand{\varY}{Y}
\newcommand{\role}{\roleP}
\newcommand{\varrole}{\roleQ}
\newcommand{\ba}{\ensuremath{\mathtt{buyer1}}}
\newcommand{\bb}{\ensuremath{\mathtt{buyer2}}}
\newcommand{\sel}{\ensuremath{\mathtt{seller}}}
\newcommand{\roleP}{\mathtt{p}}
\newcommand{\roleQ}{\mathtt{q}}
\newcommand{\roleR}{\mathtt{r}}
\newcommand{\roleS}{\mathtt{s}}
\newcommand{\buy}{\ensuremath{\mathtt{buyer}}}
\newcommand{\sell}{\ensuremath{\mathtt{seller}}}

\newcommand{\roleSet}{\mathtt{\Pi}}

\newcommand{\roles}{\pi}

\newcommand{\gtype}{\mathcal{G}}

\newcommand{\gact}{\alpha}

\newcommand{\val}{\valA}
\newcommand{\varval}{\valB}
\newcommand{\valA}{a}
\newcommand{\valB}{b}
\newcommand{\valC}{c}
\newcommand{\valD}{d}
\newcommand{\valE}{e}
\newcommand{\valF}{f}
\newcommand{\valG}{g}
\newcommand{\valH}{h}
\newcommand{\valI}{i}

\newcommand{\valSet}{\mathcal{A}}

\newcommand{\buffer}{\mathbb{B}}
\newcommand{\emptybuffer}{\varepsilon}

\newcommand{\sesstT}{T}
\newcommand{\sesstS}{S}
\newcommand{\sesstR}{R}
\newcommand{\sesst}{\sesstT}
\newcommand{\varsesst}{\sesstS}

\newcommand{\lang}{L}

\newcommand{\str}{\varphi}
\newcommand{\varstr}{\psi}

\newcommand{\permutation}{\sigma}

\newcommand{\alp}{\Sigma}

\newcommand{\act}[3]{\texttt{#1}\lred{\textit{#2}}\texttt{#3}}
\newcommand{\xact}[3]{\texttt{#1}\xrightarrow{\makebox[8mm]{$\scriptstyle#2$}}\texttt{#3}}
\newcommand{\gaction}[3]{#1\lred{#3}#2}
\newcommand{\gseq}{;}
\newcommand{\gor}{\vee}
\newcommand{\gand}{\wedge}

\newcommand{\gend}{\mathsf{\color{lucagreen}skip}}
\newcommand{\gstar}{\mathclose{{}^*}}
\newcommand{\kstar}[1]{\mathbin{{}^{#1*}}}

\newcommand{\End}{\mysf{\color{lucagreen}end}}
\newcommand{\Output}{\bigoplus}
\newcommand{\Input}{\sum}

\newcommand{\Out}[2]{#1{!}#2}
\newcommand{\In}[2]{#1{?}#2}

\newcommand{\scup}{\uplus}
\newcommand{\bsep}{\fatsemi}

\newcommand{\eqdef}{\mathrel{\stackrel{\text{def}}{=}}}
\newcommand{\gless}[2]{#2\leqslant#1}
\newcommand{\meq}{\simeq}

\newcommand{\traces}{\mathrm{tr}}
\newcommand{\rec}[2]{\mathtt{rec}~#1.#2}%
\newcommand{\shuffle}{\mathbin{\rotatebox[origin=c]{270}{\ensuremath{\exists}}}}
\newcommand{\closure}[1]{#1^{\#}}
\newcommand{\perm}[1]{#1^{\circ}}

\newcommand{\asup}{\merge}

\newcommand{\substitution}[1]{\substitu(#1)}
\newcommand{\substitu}{\rho}
\newcommand{\substitutionZ}[1]{\substitu_0(#1)}
\newcommand{\substitutionL}[1]{\substitu_\ell(#1)}
\newcommand{\substitutionLP}[1]{\substitu_{\ell+1}(#1)}
\newcommand{\substitutionU}[1]{\substitu_1(#1)}
\newcommand{\quotesymbol}[1]{\guillemotleft\,$#1$\,\guillemotright}
\newcommand{\project}[3]{#1 \vdash #2 ~\triangleright~ #3}
\newcommand{\projecta}[3]{#1 \vdash_{\color{lucared}\mathsf{a}} #2 ~\triangleright~ #3}

\newcommand{\rulename}[1]{\text{\textsc{(#1)}}}

\newcommand{\cont}{\Updelta}

\newcommand{\type}{\mysf{t}}

\newcommand{\true}{\mysf{true}}
\newcommand{\false}{\mysf{false}}

\newcommand{\extsum}{+}
\newcommand{\intsum}{\oplus}

\newcommand{\lred}[1]{\stackrel{#1}{\longrightarrow}}
\newcommand{\xlred}[1]{\xrightarrow{#1}}

\newcommand{\wlred}[1]{\stackrel{#1}{\Longrightarrow}}

\newcommand{\xwlred}[1]{\xRightarrow{#1}}

\newcommand{\mysf}[1]{\ensuremath{{\sf {#1}}}}

\definecolor{beppeblue}{rgb}{0,0,0.4}
\definecolor{lucared}{rgb}{0.5,0,0}
\definecolor{lucagreen}{rgb}{0,0.3,0}

\newcommand{\set}[1]{\ensuremath{\{#1\}}}

\newcommand{\marginnote}[2]{
\ifnocomment
\else
    \marginpar{\parbox{\marginparwidth}{\flushleft \footnotesize \textbf{#1}: #2}}
\fi
}

\newcommand{\Mariangiola}[1]{\marginnote{M}{\textcolor{magenta}{#1}}}
\newcommand{\bailout}{\textit{bailout}}
\newcommand{\handover}{\textit{handover}}
\newcommand{\bailoutT}{\textit{bailout}}
\newcommand{\handoverT}{\textit{handover}}

\newcommand{\lh}[1]{\textit{#1}}

\theoremstyle{plain}

\newtheorem{lemma}{Lemma}[section]
\newtheorem{theorem}{Theorem}[section]
\newtheorem{corollary}{Corollary}[section]

\theoremstyle{definition}
\newtheorem{definition}{Definition}[section]

\theoremstyle{remark}
\newtheorem{remark}{Remark}[section]

\begin{document}

\title[On Global Types and Multi-Party Sessions]{On Global Types and Multi-Party Sessions\rsuper*}
%\subtitle{Survey and new results}
\titlecomment{{\lsuper*}This work was
    presented as invited talk at FMOODS \& FORTE 2011, joint 13th IFIP
    International Conference on Formal Methods for Open Object-based
    Distributed Systems and 31th IFIP International Conference on
    FORmal TEchniques for Networked and Distributed Systems. A short 
    version was included in the proceedings thereof.}
%, volume 6722 of Lecture
%    Notes in Computer Science, Springer, 2011.}

\author[G.~Castagna]{Giuseppe Castagna\rsuper a}
\address{{\lsuper a}CNRS,  PPS, Univ Paris Diderot, Sorbonne Paris Cit\'e, Paris, France}
\email{Giuseppe.Castagna\@pps.jussieu.fr}  %optional

\author[M.~Dezani-Ciancaglini]{Mariangiola Dezani-Ciancaglini\rsuper b}
\address{{\lsuper{b,c}}Dipartimento di Informatica, Universit\`a degli Studi di Torino, Torino, Italy}
\email{\{dezani,padovani\}\@di.unito.it}  %optional

\author[L.~Padovani]{Luca Padovani\rsuper c}	%optional
\address{\vskip-6 pt}


%% the abstract has to PRECEED the command \maketitle:
%% be sure not to issue the \maketitle command twice!

\begin{abstract}
  Global types are formal specifications that describe communication protocols in terms of their global interactions. We present a new, streamlined language of global types equipped with
  a trace-based semantics and whose features and restrictions are
  semantically justified.
  The multi-party sessions obtained projecting our global types enjoy
  a liveness property in addition to the traditional progress and are
  shown to be sound and complete with respect to the set of traces of
  the originating global type.
  Our notion of completeness is less demanding than the classical
  ones, allowing a multi-party session to leave out redundant traces
  from an underspecified global type.
  In addition to the technical content, we discuss some limitations of
  our language of global types and provide an extensive comparison
  with related specification languages adopted in different
  communities.
\end{abstract}

\maketitle

\section{Introduction}
\label{sec:intro}

Relating the global specification of a system of communicating
entities with an implementation (or description) of the single
entities is a standard problem in many different areas of computer
science.  The recent development of \emph{session-oriented
  interactions} has renewed the interest in this problem. In this work
we attack it from the behavioral type and process algebra perspectives
and briefly compare the approaches used in other areas.

A (multi-party) session is a place of interaction for a restricted
number of participants that communicate messages. The interaction may
involve the exchange of arbitrary sequences of messages of possibly
different types. Sessions are restricted to a (usually fixed) number
of participants, which makes them suitable as a structuring construct
for systems of communicating entities.
In this work we define a language to describe the interactions that
may take place among the participants implementing a given session. In
particular, we aim at a definition based on few ``essential''
assumptions that should not depend on the way each single participant
is implemented. To give an example, a bargaining protocol that
includes two participants, ``seller'' and ``buyer'', can be informally
described as follows:
\begin{quote}
  \emph{Seller sends buyer a price \underline{and} a description of
    the product; \underline{then} buyer sends seller acceptance
    \underline{or} it quits the conversation}.
\end{quote}

If we abstract from the value of the price and the content of the
description sent by the seller, this simple protocol describes just
two possible executions, according to whether the buyer accepts or
quits. If we consider that the price and the description are in
distinct messages then the possible executions become four, according
to which communication happens first. While the protocol above
describes a finite set of possible interactions, it can be easily
modified to accommodate infinitely many possible executions, as well
as additional conversations: for instance the protocol may allow
``buyer'' to answer ``seller'' with a counteroffer, or it may
interleave this bargaining with an independent bargaining with a
second seller.

All essential features of protocols are in the example above, which
connects some basic communication actions by the flow control points
we underlined in the text. More generally, we interpret a protocol as a possibly
infinite set of finite sequences of interactions between a fixed set of
participants. We argue that the %set of
sequences that characterize %s 
 a
protocol---and thus the protocol itself---can be described by a
language with one form of atomic action and three composition
operators.

\begin{desCription}
\item\noindent{\hskip-12 pt\bf Atomic actions:}\
The only atomic action is the interaction, which consists of one (or
more) sender(s) (\eg, ``seller sends''), the content of the
communication (\eg, ``a price'', ``a description'', ``acceptance''),
and one (or more) receiver(s) (\eg, ``buyer'').

\item\noindent{\hskip-12 pt\bf Compound actions:}\
Actions and, more generally, protocols can be composed in three
different ways. First, two protocols can be composed sequentially
(\eg, ``Seller sends buyer a price\dots; \emph{\textbf{then}} buyer
sends\dots'') thus imposing a precise order between the actions of the
composed protocols. Alternatively, two protocols can be composed
without specifying any constraint (\eg, ``Seller sends a
price \emph{\textbf{and}} (sends) a description'') thus indicating
that any order between the actions of the composed protocols is
acceptable. Finally, protocols can be composed in alternative (\eg,
``buyer sends acceptance \emph{\textbf{or}} it quits''), thus offering
a choice between two or more protocols only one of which may be
chosen.
\end{desCription}

More formally, we use $\gaction{\role}{\varrole}{\val}$ to state that
participant $\role$ sends participant $\varrole$ a message whose
content is described by $\val$, and we use \quotesymbol\gseq,
\quotesymbol\gand, and \quotesymbol\gor{} to denote sequential,
unconstrained, and alternative composition, respectively. Our initial
example can thus be rewritten %described
as follows:
\begin{equation}\label{spec1}
\begin{array}{l}
(\act{\sell}{\textit{descr}}{\buy}\gand\act{\sell}{\textit{price}}{\buy})\gseq\\
(\act{\buy}{\textit{accept}}{\sell}\gor\act{\buy}{\textit{quit}}{\sell})
\end{array}
\end{equation}
The first two actions are composed without constraints, and they are to be
followed by one (and only one) action of the alternative before
ending. Interactions of unlimited length can be defined by resorting
to a Kleene star notation. For example to extend the previous protocol
so that the buyer may send a counter-offer and wait for a new price,
it suffices to add a Kleene-starred line:
\newcommand{\bro}{\texttt{broker}}
\begin{equation}\label{spec2}
\begin{array}{l}
(\act{\sell}{\textit{descr}}{\buy}\gand\act{\sell}{\textit{price}}{\buy})\gseq\\
(\act{\buy}{\textit{offer}}{\sell}\gseq\act{\sell}{\textit{price}}{\buy})
\texttt{*}\gseq\\
(\act{\buy}{\textit{accept}}{\sell}\gor\act{\buy}{\textit{quit}}{\sell})
%\gseq\\\gend
\end{array}
\end{equation}
% \begin{equation}\label{spec3}
% \begin{array}{l}
% (\gaction{\sell}{\bro}{\textit{agency}}\gand\gaction{\sell}{\buy}{\textit{price}})\gseq\\
% (\gaction{\buy}{\bro}{\textit{offer}};\gaction{\bro}{\buy}{\textit{price}})%^
% {\pmb{*}}\gseq\\
% (\gaction{\buy}{\bro}{\textit{accept}}\gor\gaction{\buy}{\bro}{\textit{quit}})\gseq\\
% \gaction{\bro}{\sell}{\textit{result}}
% %\gseq\\\gend
% \end{array}
% \end{equation}

The description above states that, after having received (in no
particular order) the price and the description from the seller, the
buyer can initiate a loop of zero or more interactions and then decide
whether to accept or quit.

Whenever there is an alternative there must be a participant that
decides which path to take. In both examples it is {\buy} that makes
the choice by deciding whether to send \textit{accept} or
\textit{quit}. The presence of a participant that decides holds true
in loops too, since it is again {\buy} that decides whether to enter
or repeat the iteration (by sending \textit{offer}) or to exit it (by
sending \textit{accept} or \textit{quit}).
We will later show that absence of such decision-makers makes
protocols impossible to implement. This last point critically depends
on the main hypothesis we assume about the systems we are going to the
describe, that is the absence of \emph{covert channels}. On the one
hand, we try to develop a protocol description language that is as
generic as possible; on the other hand, we limit the power of the
system and require all communications between different participants
to be explicitly stated. In doing so we rule out protocols whose
implementation essentially relies on the presence of secret/invisible
communications between participants: a protocol description must
contain all and only the interactions used to implement it.

%% The use of a Kleene star rather than general recursion is not fortuitous
%% since a fundamental property we want our protocols to satisfy is that they
%% eventually terminate. This is a design decision that distinguish our approach from
%% those ...

%% \textbf{Speak of liveness property and fair behaviour ... }

Protocol specifications such as the ones presented above are usually
called \emph{global types} to emphasize the fact that they describe
the acceptable behaviors of a system from a global point of view. In
an actual implementation of the system, though, each participant
autonomously implements a different part of the protocol.
To understand whether an implementation satisfies a specification, one
has to consider the set of all possible sequences of synchronizations
performed by the implementation and check whether this set satisfies
five basic properties:%\Luca{Siccome ne bastano due non sono ``basic''.}\Mariangiola{a me ed a Beppe sembra che vada bene, basic non vuol dire independent}
\begin{enumerate}[(1)]
\item \emph{Sequentiality:} if the specification states that two
  interactions must occur in a given order (by separating them by a
  \quotesymbol;), then this order must be respected by all possible
  executions. So an implementation in which {\buy} may send
  \textit{accept} before receiving \textit{price} violates the
  specification~\eqref{spec1} (and~\eqref{spec2}).

\item \emph{Alternativeness:} if the specification states that two
  interactions are alternative, then every execution must exhibit one
  and only one of these two actions. So an implementation in which
  {\buy} emits both \textit{accept} and \textit{quit} (or none of
  them) in the same execution violates the specification~\eqref{spec1}.

\item \emph{Shuffling:} if the specification composes two sequences of
  interactions in an unconstrained way, then all executions must
  exhibit some shuffling (in the sense used in combinatorics and
  algebra) of these sequences. So an implementation in which {\sell}
  emits \textit{price} without emitting \textit{descr} violates the
  specification~\eqref{spec1}.

\item \emph{Fitness:} if the implementation exhibits a sequence of
  interactions, then this sequence is expected by (\ie, it
  fits) the specification. So any implementation in which {\sell}
  sends {\buy} any message other than \textit{price} and
  \textit{descr} violates the specification~\eqref{spec1}.

\item \emph{Exhaustivity:} if some sequence of interactions is
  described by the specification, then there must exist at least an
  execution of the implementation that exhibits these actions
  (possibly in a different order). So an implementation in which
  no execution of {\buy} emits \textit{accept} violates the specification~\eqref{spec1}.
\end{enumerate}

% Of course these definitions are vague as long as we do not formally
% define what the sequences of actions described by a specification
% and emitted by an implementation are. Though they convey the idea about
% how we want to relate the global specification of a protocol with
% the local specifications of each single participants.

\noindent Checking whether an implemented system satisfies a specification by
comparing the actual and the expected sequences of interactions is
non-trivial, for systems are usually infinite-state.  Therefore, on
the lines of \cite{%carbone.honda.yoshida:esop07,
CHY08}, we proceed the
other way round: we extract from a global type the local specification
(usually dubbed \emph{local type} or \emph{session type}
\cite{THK,honda.vasconcelos.kubo:language-primitives}) of each
participant in the system and we type-check the implementation of each
participant against the corresponding session
type. %local specification.
If the projection operation is done properly and the global
specification satisfies some well-formedness conditions, then we are
guaranteed that the implementation satisfies the specification.
%
% If we can
% do so, then the problem of checking whether the implementation
% satisfies the specification can be done by checking if each
% implementation satisfies the extracted behaviour. But the interest is
% more than that since the extracted behavious should be a guide to the
% implementation of single participants ...
%
As an example, the global type~(\ref{spec1}) can be projected to the
following behaviors for {\buy} and {\sell}:
\[
\begin{array}{rcl}
\sell &\mapsto & \Out{\buy}{\textit{descr}}.\Out{\buy}{\textit{price}}.(\In{\buy}{\textit{accept}}\extsum\In{\buy}{\textit{quit}})\\
\buy  &\mapsto & \In{\sell}{\textit{descr}}.\In{\sell}{\textit{price}}.(\Out{\sell}{\textit{accept}}\intsum\Out{\sell}{\textit{quit}}) 
\end{array}
\]
or to
\[
\begin{array}{rcl}
\sell &\mapsto   &\Out{\buy}{\textit{price}}.\Out{\buy}{\textit{descr}}.(\In{\buy}{\textit{accept}}\extsum\In{\buy}{\textit{quit}})\\
\buy  &\mapsto   &\In{\sell}{\textit{price}}.\In{\sell}{\textit{descr}}.(\Out{\sell}{\textit{accept}}\intsum\Out{\sell}{\textit{quit}}) 
\end{array}
\]
where $\Out\role\val$ denotes the output of a message $\val$ to
participant $\role$, $\In\role\val$ the input of a message $\val$ from
participant $\role$,
$\In{\role}{\val}.\sesst\extsum\In{\varrole}{\varval}.\varsesst$ the
(external) choice to continue as $\sesst$ or $\varsesst$ according to
whether $\val$ is received from $\role$ or $\varval$ is received from
$\varrole$ and, finally,
$\Out{\role}{\val}.\sesst\intsum\Out{\varrole}{\varval}.\varsesst$
denotes the (internal) choice between sending $\val$ to $\role$ and
continue as $\sesst$ or sending $\varval$ to $\varrole$ and continue
as $\varsesst$. We will call {\em session environments} the mappings from participants to their session types.
It is easy to see that any two processes implementing  {\buy} and {\sell} will
satisfy the global type~(\ref{spec1}) if \emph{and only  if} %\Luca{L'``only if'' \`e falso, ci sono altre (in effetti  infinite) implementazioni che soddisfano la specifica.}\Mariangiola{tutte le note di questa pagina sono state tenute in considerazione nelle modifiche di Beppe} 
 their visible
behavior matches one of the two session environments %sets of behaviors 
above (these session environments %the above mapping
thus represent some sort of minimal typings of processes implementing {\buy} and {\sell}). In
particular, both the above session environments %mappings 
are fitting and
exhaustive with respect to the specification since they precisely
describe what the single participants are expected and bound to do.

In this work we will discuss how to characterize a set of session
environments (if any) from participants to session types that is sound
and complete, with respect to a given global type.
We will also show an algorithm that, in several practical cases, can
effectively perform the extraction of the session environment from a
global type.
Observe that there are global types
that are intrinsically flawed, in the sense that they do not admit any
implementation (without covert channels) satisfying them. We classify
flawed global types in three categories, according to the seriousness
of their flaws.
\begin{desCription}
\item\noindent{\hskip-12 pt\bf No sequentiality:}\
  The mildest flaws are those in which the global type specifies some
  sequentiality constraint between independent interactions, such as
  in $(\gaction{\role}{\varrole}{\valA} \gseq
  \gaction{\roleR}{\roleS}{\valB})$, since it is impossible to
  implement $\roleR$ so that it sends $\valB$ only after that $\roleQ$
  has received $\valA$ (unless this reception is notified on a covert
  channel, of course). Therefore, it is possible to find exhaustive
  (but not fitting) implementations that include some unexpected
  sequences which differ from the expected ones only by a permutation
  of interactions done by independent participants. The specification
  at issue can be easily patched by replacing some \quotesymbol\gseq's
  by \quotesymbol\gand's.

\item\noindent{\hskip-12 pt\bf No knowledge for choice:}\
  A more severe kind of flaw occurs when the global type requires some
  participant to behave in different ways in accordance with some
  choice it is unaware of.  For instance, in the global type
\[
\begin{array}{l}
  (\gaction{\role}{\varrole}{\valA} \gseq 
   \gaction{\varrole}{\roleR}{\valA} \gseq 
   \gaction{\roleR}{\role}{\valA})
  \quad\gor\quad
  (\gaction{\role}{\varrole}{\valB} \gseq 
   \gaction{\varrole}{\roleR}{\valA} \gseq
   \gaction{\roleR}{\role}{\valB})
\end{array}
\]
participant $\role$ chooses the branch to execute, but after having
received $\valA$ from $\varrole$ participant $\roleR$ has no way to
know whether it has to send $\valA$ or $\valB$.
Also in this case it is possible to find exhaustive (but
  not fitting)
implementations of the global type where the participant $\roleR$
chooses to send $\valA$ or $\valB$ independently of what $\roleP$
decided to do.

\item\noindent{\hskip-12 pt\bf No knowledge, no choice:}\
  In the worst case it is not possible to find an exhaustive
  implementation of the global type, for it specifies some combination
  of incompatible behaviors, such as performing an input or an output
  in mutual exclusion.
  This typically is the case of the absence of a decision-maker in the
  alternatives such as in
\[
\gaction{\role}{\varrole}{\valA} \gor 
\gaction{\varrole}{\role}{\valB}
\]
where each participant is required to choose between sending or
receiving. There seems to be no obvious way to patch these global
types without reconsidering also the intended semantics.%\vspace{-3mm}
\end{desCription}
% {\color{blue}One of the contribution of this work is the formal
%   definition of these three classes of errors, as well as the static
%   detection of flawed specifications.}\Luca{Posticipare}
We conclude this introduction by stressing that in this work we focus
on single sessions. The participants of a system can concurrently
implement and bring forward different sessions but we suppose the
management of different sessions (\eg, the exchange of sessions
channels) to belong to the meta-level. The internalization of such a
level (\ie, the use of delegation) is left for future work (see
Section~\ref{mpst}).

\subsubsection*{Outline and contributions.}

We introduce a streamlined language of global specifications---that we
dub \emph{global types} (Section~\ref{sec:gtypes})---and relate it
with \emph{session environments} (Section~\ref{sec:sessions}), that
is, with sets of independent, sequential, asynchronous {\em session
  types} to be type-checked against implementations.  Global types are
just regular expressions augmented with a shuffling operator and their
semantics is defined in terms of finite sequences of interactions.
The semantics chosen for global types ensures that every
implementation of a global type preserves the possibility to reach a
state where \emph{every} participant has successfully terminated.
This implies that no participant of a multi-party session starves
waiting for messages that are never sent or sends messages that no
other participant will ever receive.  This property is stronger than
the progress enforced by other theories of multi-party sessions, where
it is enough that two participants synchronize to be able to say that
the session has progress.  Technically, we make a \emph{strong
  fairness assumption} on sessions by considering only fair
computations, those where infinitely often enabled transitions occur
infinitely often.
%
% The main achievement of this work, however, resides in the fact that
% it ensures strong fairness conditions on every implementation of a
% global type. In particular, in our framework, global types are only
% implemented by systems where not only every component has, at any
% moment, a chance to sucessfully terminate but it is also assured not
% to starvate or to be left behind, and thus participate to the global
% progress of the session (technically, we ensure \emph{strong infinite
% fairness} or, in Lamport terminology, we produce hyperfair
% sessions). 

In Section~\ref{sec:projection} we study the relationship between
global types and sessions. We do so by defining a projection operation
that extracts from a global type \emph{all} the (sets of) possible session types %behaviors 
of its participants. This projection is useful not only to
check the implementability of a global description (and, incidentally,
to formally define the notions of errors informally described so far)
but, above all, to relate in a compositional and modular way a global
type with the sets of distributed processes %\Mariangiola{anche qui per me sono session types} 
that implement it. We also
identify a class of well-formed global types whose projections
need no covert channels. Interestingly, we are able to effectively
characterize well-formed global types solely in terms of their
semantics.

% Covert channels are needed to implement a global type in just two
% cases: either to synchronize two actions (the ``no sequentiality''
% problem of the introduction) or to transmit knowledge to a participant
% (the two ``no knowledge'' problems).

% We show that if a global type can be implemented, then these two cases
% are avoided. 

% That a global type does not need covert synchronization channels can
% be directly verified on the global type itself (by a finite test on
% the automaton associated to it): we call 
% that enjoy this property. In order to check that a global type does
% not need covert channels for knowledge transmission either, we give a
% sound characterization of the ``implementable'' well-formed global
% types. In other words, we associate to each implementable global type
% the set of all its possible implementations. 

In Section~\ref{sec:sal} we present a projection algorithm for global
types. The effective generation of all possible projections is
impossible.  The reason is that the projectability of a global type
may rely on some global knowledge that is no longer available when
working at the level of single session types: while in a global
approach we can, say, add to some participant new synchronization
offers that, thanks to our global knowledge, we know will never be
used, this cannot be done when working at the level of single
participant. Therefore in order to work at the projected level we will
use stronger assumptions that ensure a sound implementation in all
possible contexts.

In Section~\ref{sec:multistar} we show some limitations deriving from
the use of the Kleene star operator in our language of global types,
and we present one possible way to circumvent them.
Section~\ref{sec:related} contains an extended survey of related work,
with samples of the literature of session types and session
choreography expressed in our syntax and an in-depth comparison with
our work. Few final considerations conclude the work in
Section~\ref{sec:conclusion}. The Appendix contains proofs and some
technical discussions.

We summarize the contributions of our work below:
\begin{iteMize}{$\bullet$}
\item With respect to (multi-party) session type
  theories~\cite{%carbone.honda.yoshida:esop07,
  CHY08}, we adopt a more
  abstract and ---we claim--- natural language of global types
  (Section~\ref{sec:gtypes}) that is closely related to the language
  of Web service choreographies in~\cite{BZ07}.
  We define a notion of session correctness that depends on a strong
  fairness assumption.
  On the one hand, this is more demanding than in other multi-party
  session theories because we insist on the property that a correct
  session must preserve the ability to reach a terminated state;
  on the other hand, we claim that eventual termination is indeed a
  desirable property of sessions, and we provide a number of examples
  showing that, if the hypothesis of an eventual termination of every session is assumed, our formalism allows for a range of projectable global
  specifications that is strictly larger than that %MD0212 aggiunto that
  other formalisms, under the same assumption, have.

\item With respect to Web service choreography
  languages~\cite{BZ07,LGMZ08,BZ08,BLZ08}, where projection is defined
  by an homomorphism between the global and the local specifications,
  we define a significantly more sophisticated projection procedure
  (Sections~\ref{sec:projection} and~\ref{sec:sal}) with two main
  upshots.
  First, we handle the projection of unconstrained composition of
  global specifications in a more flexible way, by permitting
  (partial) serialization of independent activities whenever this
  is either convenient or necessary.
  Second, we widen the range of projectable choreographies by imposing
  fewer constraints on the way alternative specifications can be
  composed together.
  We also point out some shortcomings of the Kleene star operator and
  propose a solution based on $k$-exit iterations that circumvents
  them (Section~\ref{sec:multistar}).

\item In order to account for the possible serializations of
  independent activities, we identify an original notion of
  completeness (Definition~\ref{pog}) of projections with respect to
  global specifications that is weaker (and consequently more
  flexible) than the corresponding notions in other theories.

\item Section~\ref{sec:related} provides a rather detailed survey of a
  wide range of related formalisms and techniques.
\end{iteMize}

\section{Global Types}
\label{sec:gtypes}

In this section we define the syntax and semantics of global types.
We assume a set $\valSet$ of \emph{message types}, ranged over by
$\val$, $\varval$, \dots, and a set $\roleSet$ of \emph{roles}, ranged
over by $\role$, $\varrole$, \dots, which we use to uniquely identify
the participants of a session; we let $\roles$, \dots range over
non-empty, finite sets of roles.

\begin{table}[t]
\caption{\label{tab:gtypes}\strut Syntax of global types.}
\framebox[\textwidth]{
\begin{math}
\displaystyle
\begin{array}{rcl@{\qquad}l}
 \gtype & ~~::=~~ & & \text{\textbf{Global Type}} \\
 &      & \gend & \text{(skip)} \\
 & \mid & \gaction{\roles}{\role}{\val} & \text{(interaction)} \\
 & \mid & \gtype\gseq\gtype & \text{(sequence)} \\
 & \mid & \gtype\gand\gtype & \text{(both)} \\
 & \mid & \gtype \gor \gtype & \text{(either)} \\
 & \mid & \gtype \gstar & \text{(star)} \\
\end{array}
\end{math}
}
\end{table}

Global types, ranged over by $\gtype$,  are the terms generated
by the grammar in Table~\ref{tab:gtypes}. 
Their syntax was already explained in Section~\ref{sec:intro} except
for two novelties.
First, we include a $\gend$ atom which denotes the unit of sequential
composition (it plays the same role as the empty word in regular
expressions). This is useful, for instance, to express optional
interactions. Thus, if in our example we want the buyer to do at most
one counteroffer instead of several ones, we just replace the starred
line in~(\ref{spec2}) by
\[
(\act{\buy}{\textit{offer}}{\sell}\gseq\act{\sell}{\textit{price}}{\buy})\gor\gend
\]
which, using  syntactic sugar of regular expressions, might be
rendered as
\[
(\act{\buy}{\textit{offer}}{\sell}\gseq\act{\sell}{\textit{price}}{\buy})
\verb'?'
\,
\]

Second, we generalize interactions by allowing a finite set of roles
on the l.h.s. of interactions. Therefore, $\gaction\roles\role\val$
denotes the fact that (the participant identified by) $\role$ waits
for an $\val$ message from all of the participants whose tags are in
$\roles$. We will write $\gaction\role\varrole\val$ as a shorthand for
$\gaction{\{\role\}}{\varrole}{\val}$.
An example showing the usefulness of multiple roles on the left-hand side of
actions is the following one
\begin{equation*}
\begin{array}{l}
(\act{seller}{price}{buyer1}\gand\act{bank}{mortgage}{buyer2})\gseq\\
 (\act{$\{$buyer1,buyer2$\}$}{accept}{seller}\gand\act{$\{$buyer1,buyer2$\}$}{accept}{bank})\end{array}
\end{equation*}
which represents two buyers waiting for both the price from a seller
and the mortgage from a bank before deciding the purchase.  Notice
that without this generalization the communication of
$\mathit{accept}$ to, say, the \texttt{seller} would be performed by
two distinct communications from $\mathtt{buyer1}$ and
$\mathtt{buyer2}$. But in that case, how could $\mathtt{buyer1}$ be
sure that $\mathtt{buyer2}$ had received $\mathit{mortgage}$ before
sending $\mathit{accept}$ to $\mathtt{seller}$? And symmetrically, how
could $\mathtt{buyer2}$ be sure that $\mathtt{buyer1}$ had received
$\mathit{price}$ before sending $\mathit{accept}$ to
$\mathtt{seller}$? Actions with multiple-senders allow us to express
the join of independent activities (in this case, the receival of
$\mathit{price}$ and $\mathit{mortgage}$).

To be as general as possible, one could also consider interactions of
the form $\gaction\roles{\roles'}\val$, which could be used to specify
broadcast communications between participants.
We avoided this generalization since it cannot be implemented without
covert channels. In fact in a sound execution
of \[\act{seller}{price}{$\{$buyer1,buyer2$\}$},\] the reception of
\textit{price} by \texttt{buyer1} should wait also for the reception
of \textit{price} by \texttt{buyer2} and vice versa, and this requires
a synchronization between \texttt{buyer1} and \texttt{buyer2}.

% \Mariangiola{Ho eliminato: It turns out that this further
%   generalization is superfluous in our setting since the interaction
%   $\gaction\roles{\set{\role_i}_{i\in I}}\val$ can be encoded as
%   $\gAnd_{i\in I}(\gaction\roles{\role_i}\val)$.
% %
% The encoding is made possible by the fact that communication is
% asynchronous and output actions are not blocking (see
% Section~\ref{sec:sessions}), therefore the order in which the
% participants in $\roles$ send $\val$ to the $\role_i$'s is irrelevant.
% %
% Vice versa, we will see \textcolor{blue}{in
%   Section~\ref{sec:projection}} that allowing sets of multiple senders
% enriches the expressiveness of the calculus, because
% $\gaction{\roles}{\role}{\val}$ can be used to \emph{join} different
% activities involving the participants in $\roles\cup\set\role$, while
% $\gAnd_{i\in I}(\gaction{\role_i}\varrole\val)$ cannot
% %
% {\color{blue} (we must defer a more detailed discussion until we have
%   formalized the sequentiality property of global types).  }}

In general we will assume $\role\not\in\roles$ for every interaction
$\gaction\roles\role\val$ occurring in a global type, that is, we forbid
participants to send messages to themselves.

For the sake of readability we adopt the following precedence of
global type operators $\longrightarrow$~~ $\gstar$~~ $\gseq$~~
$\gand$~~ $\gor$. 

Global types denote languages of legal interactions that can occur in
a multi-party session. These languages are defined over the alphabet
of interactions
\[
\alp=\set{\gaction{\roles}{\role}{\val}\mid
  \roles\subset_{\mathit{fin}}\roleSet, \role\in\roleSet, \role\not\in\roles,
  \val\in\valSet}
\]
and we use $\gact$ as short for $\gaction{\roles}{\role}{\val}$ when
possible; we use $\str$, $\varstr$, \dots to range over strings in
$\alp^*$ and $\varepsilon$ to denote the empty string, as usual. To
improve readability we will occasionally use \quotesymbol\gseq{} to
denote string concatenation.

In order to express the language of a global type having the shape $\gtype_1 \gand
\gtype_2$ we need a standard shuffling operator over languages, which
can be defined as follows:

\begin{definition}[shuffling]
  The \emph{shuffle} of $\lang_1$ and $\lang_2$, denoted by $ \lang_1
  \shuffle \lang_2$, is the language defined by:
$
  \lang_1 \shuffle \lang_2
  \eqdef
  \{ \str_1\varstr_1\cdots\str_n\varstr_n \mid \str_1\cdots\str_n \in \lang_1
  \wedge
  \varstr_1\cdots\varstr_n \in \lang_2 \}.
$
\end{definition}

Observe that, in $\lang_1 \shuffle \lang_2$, the order of interactions
coming from one language is preserved, but these interactions can be
interspersed with other interactions coming from the other language.

\begin{definition}[traces of global types]
\label{def:traces}
The set of \emph{traces} of a global type is inductively defined by
the following equations:
\[
\begin{array}{r@{~}c@{~}l}
  \traces(\gend) & = & \{ \varepsilon \} \\
  \traces(\gaction{\roles}{\role}{\val}) & = & \{ \gaction\roles\role\val \} 
\end{array} 
\quad
\begin{array}{r@{~}c@{~}l}
  \traces(\gtype_1\gseq\gtype_2) & = & \traces(\gtype_1)\traces(\gtype_2) \\
  \traces(\gtype\gstar) & = & (\traces(\gtype))^\star 
\end{array}
\quad
\begin{array}{r@{~}c@{~}l}
  \traces(\gtype_1\gor\gtype_2) & = & \traces(\gtype_1)\cup\traces(\gtype_2) \\
  \traces(\gtype_1\gand\gtype_2) & = & \traces(\gtype_1)\shuffle\traces(\gtype_2) 
\end{array}
\]
where juxtaposition denotes concatenation and $(\cdot)^\star$ is the usual
Kleene closure of regular languages.
\end{definition}

Before we move on, it is worth noting that $\traces(\gtype)$ is a
regular language (recall that regular languages are closed under
shuffling). Since a regular language is made of \emph{finite strings},
we are implicitly making the assumption that a global type specifies
interactions of finite length. This means that we are considering
interactions of arbitrary length, but such that the termination of all
the involved participants is always within reach. This is a subtle,
yet radical change from other multi-party session theories, where
infinite interactions are considered legal.

\iflongversion
By way of example, consider the global type
\[
\gtype =
(\gaction\role\varrole\valA \gand \gaction\role\varrole\valB);
(\gaction\varrole\role\valC; \gaction\role\varrole\valB)\gstar;
(\gaction\varrole\role\valD \gor \gaction\varrole\role\valE)
\]
which represents the bargain protocol described in the
introduction. Every long enough string in $\traces(\gtype)$ has either the
form
\[
 \varstr;
  \gaction\varrole\role\valC;\gaction\role\varrole\valB;
  \cdots;
  \gaction\varrole\role\valD
  \text{\qquad or\qquad}
  \varstr;
  \gaction\varrole\role\valC;\gaction\role\varrole\valB;
  \cdots;
  \gaction\varrole\role\valE
\]
for some appropriate $\varstr$, meaning that the phase in which the
buyer makes new offers can be arbitrarily long, although it must
eventually terminate with the decision to either quit or accept.
\fi
%%% Local Variables: 
%%% mode: latex
%%% TeX-master: "main"
%%% End: 

\newpage
\section{Multi-Party Sessions}
\label{sec:sessions}

We devote this section to the formal definition of the behavior of the
participants of a multi-party session.

\subsection{Session Types}

We need an infinite set of recursion variables ranged over by
$\recvar$, \dots. Pre-session types, ranged over by $\sesst$,
$\varsesst$, \dots, are the terms generated by the grammar in
Table~\ref{tab:stypes} such that all recursion variables are guarded
by at least one input or output prefix.
We consider pre-session types modulo associativity, commutativity, and
idempotence of internal and external choices, fold/unfold of
recursions and the equalities
\[
\Out{\role}{\val}.\sesst\intsum\Out{\role}{\val}.\varsesst
=
\Out{\role}{\val}.(\sesst\intsum\varsesst)
\qquad\qquad
\In{\roles}{\val}.\sesst\extsum\In{\roles}{\val}.\varsesst
=
\In{\roles}{\val}.(\sesst\extsum\varsesst)
\]

Pre-session types are behavioral descriptions of the participants of a
multi-party session.
Informally, $\End$ describes a successfully terminated party that no
longer participates to a session.
The pre-session type $\Out\role\val.\sesst$
describes a participant that sends an $\val$ message to participant
$\role$ and afterwards behaves according to $\sesst$; 
the pre-session type $\In\roles\val.\sesst$ describes a participant
that waits for an $\val$ message from all the participants in $\roles$
and, upon arrival of the message, behaves according to $\sesst$; we
will usually abbreviate $\In{\{\role\}}\val.\sesst$ with
$\In\role\val.\sesst$.
Behaviors can be combined by means of behavioral choices $\intsum$ and
$\extsum$: $\sesst \intsum \varsesst$ describes a participant that
internally decides whether to behave according to $\sesst$ or
$\varsesst$; $\sesst \extsum \varsesst$ describes a participant that
offers to the other participants two possible behaviors, $\sesst$ and
$\varsesst$. The choice as to which behavior is taken depends on the
messages sent by the other participants.
In the following, we sometimes use $n$-ary versions of internal and
external choices and write, for example, $\Output_{i=1}^n \Out{\role_i}{\val_i}.\sesst_i$ for
$\Out{\role_1}{\val_1}.\sesst_1 \intsum \cdots \intsum
\Out{\role_n}{\val_n}.\sesst_n$ and
$\Input_{i=1}^n
\In{\roles_i}{\val_i}.\sesst_i$ for $\In{\roles_1}{\val_1}.\sesst_1
\extsum \cdots \extsum \In{\roles_n}{\val_n}.\sesst_n$.
As usual, terms $\recvar$ and $\rec\recvar\sesst$ are used for
describing recursive behaviors. For example,
$\rec\recvar(\Out\role\val.\recvar \intsum \Out\role\valB.\End)$
describes a participant that sends an arbitrary number of $\val$
messages to $\role$ and terminates by sending a $\valB$ message; dually, $\rec\recvar(\In\role\val.\recvar \extsum
\In\role\valB.\End)$ describes a participant that is capable of
receiving an arbitrary number of $\val$ messages from $\role$ and
terminates as soon a $\valB$ message is received.

\begin{table}[t]
\caption{\label{tab:stypes}\strut Syntax of pre-session types.}
\framebox[\textwidth]{
\begin{math}
\displaystyle
\begin{array}{rcl@{\qquad}l}
 \sesst
 & ::= & & \text{\textbf{Pre-Session Type}}  \\ 
 &      & \End & \text{(termination)} \\
 &  |   & \recvar & \text{(variable)} \\
 &  |  &  \Out{\role}{\val}.\sesst & \text{(output)} \\
 &  |  & \In{\roles}{\val}.\sesst& \text{(input)} \\
 &  |  &  \sesst\intsum\sesst & \text{(internal choice)} \\
 &  |  &  \sesst\extsum\sesst & \text{(external choice)} \\
 &  |  &  \rec\recvar\sesst & \text{(recursion)} \\
\end{array}
\end{math}
}
\end{table}

Session types are the pre-session types where internal choices are
used to combine outputs, external choices are used to combine inputs,
and the continuation after every prefix is uniquely determined by the
prefix. Formally:

\begin{definition}[session types]
  A pre-session type $\sesst$ is a \emph{session type} if either:
\begin{iteMize}{$\bullet$}
\item $\sesst = \End$, or

\item $\sesst=\Output_{i\in I} \Out{\role_i}{\val_i}.\sesst_i$ and
  $\forall i,j\in I$ we have that $\Out{\role_i}{\val_i} =
  \Out{\role_j}{\val_j}$ implies $i = j$ and each $\sesst_i$ is a
  session type, or

\item $\sesst=\Input_{i\in I} \In{\roles_i}{\val_i}.\sesst_i$ and
  $\forall i,j\in I$ we have that $\roles_i \subseteq \roles_j$ and
  $\val_i = \val_j$ imply $i= j$ and each $\sesst_i$ is a session
  type.
\end{iteMize}
\end{definition}

\subsection{Session Environments}

A session environment is defined as the set of the session types of its
participants, where each participant is uniquely identified by a
role. Formally:

\begin{definition}[session environment]
  A \emph{session environment} (briefly, \emph{session}) is a finite
  map $\{ \role_i : \sesst_i \}_{i\in I}$.
\end{definition}

In what follows we use $\cont$ to range over sessions and we write
$\cont \scup \cont'$ to denote the union of sessions, when their
domains are disjoint.

To describe the operational semantics of a session we
model an asynchronous form of communication where the messages sent by
the participants of the session are stored within a \emph{buffer}
associated with the session. Each message has the form
$\gaction\role\varrole\val$ describing the sender $\role$, the
receiver $\varrole$, and the type $\val$ of message.
Buffers, ranged over by $\buffer$, \dots, are finite sequences
$\gaction{\role_1}{\varrole_1}{\val_1}::\cdots::\gaction{\role_n}{\varrole_n}{\val_n}$
of messages considered modulo the least congruence $\meq$ over buffers
such that
\[
  \gaction\role\varrole\val::\gaction{\role'}{\varrole'}{\varval}
  \meq
  \gaction{\role'}{\varrole'}{\varval}::\gaction\role\varrole\val
\]
when $ \role\ne\role'$ or $\varrole\ne\varrole'$, that is, we care
about the order of messages in the buffer only when these have both
the same sender and the same receiver. In practice, this corresponds
to a form of communication where each pair of participants of a
multi-party session is connected by a distinct FIFO buffer.

There are two possible reductions of a session:
\[\begin{array}{@{}rclr@{}}
 \buffer\bsep
    \{ \role : \Output_{i\in I} \Out{\role_i}{\val_i}.\sesst_i \} \scup \cont
    &\lred{}&
    (\gaction{\role}{\role_k}{\val_k}) {::} \buffer\bsep
    \{ \role : \sesst_k \} \scup \cont \hspace{-2em}&\qquad\scriptstyle (k\in I)~\\\\
\buffer{::}(\gaction{\role_i\!}{\!\role}{\val})_{i\in I}\bsep
    \{ \role : \Input_{j\in J} \In{\roles_j}{\val_j}.\sesst_j \} \scup \cont
    &\xlred{\gaction{\roles_k}{\role}{\val}}&
    \buffer\bsep
    \{ \role : \sesst_k \} \scup \cont &{\scriptstyle\left(
        \begin{array}{c}\scriptstyle
          k\in J\quad \val_k = \val
          \\\scriptstyle
          \roles_k = \{ \role_i \mid i \in I \}
        \end{array}\right)}   
  
\end{array}
\]
The first rule describes the effect of an output operation performed
by participant $\role$, which stores the message
$\gaction\role{\role_k}{\val_k}$ in the buffer and leaves participant
$\role$ with a residual session type $\sesst_k$ corresponding to the
message that has been sent.
The second rule describes the effect of an input operation performed
by participant $\role$. If the buffer contains enough messages of type
$\val$ coming from all the participants in $\roles_k$, those messages
are removed from the buffer and the receiver continues as described in
$\sesst_k$.
In this rule we decorate the reduction relation with
$\gaction{\roles_k}{\role}{\val}$ that describes the occurred
interaction (as we have already remarked, we take the point of view
that an interaction is completed when messages are received).
This decoration will allow us to relate the behavior of an implemented
session with the traces of a global type (see
Definition~\ref{def:traces}).
According to this semantics, the input prefixes
$\In{\{\role_1,\dots,\role_n\}}\val$ resemble \emph{join patterns}
$\In{\role_1}\val \mathbin{\&} \cdots \mathbin{\&} \In{\role_n}\val$
in the join calculus~\cite{FournetGonthier96}, except that we impose
that all the messages coming from $\role_1,\dots,\role_n$ have the
same type.

We adopt some conventional notation: we write $\wlred{}$ for the
reflexive, transitive closure of $\lred{}$; we write $\wlred\gact$ for
the composition $\wlred{}\lred\gact\wlred{}$ and
$\xwlred{\gact_1\cdots\gact_n}$ for the composition
$\wlred{\gact_1}\cdots\wlred{\gact_n}$.

We can now formally characterize the ``correct sessions'' as those in
which, no matter how they reduce, it is always possible to reach a
state where all of the participants are successfully terminated and
the buffer has been emptied.

\begin{definition}[live session]
\label{def:liveness}
We say that $\cont$ is a \emph{live session} if $\emptybuffer\bsep
\cont \wlred{\str} \buffer\bsep \cont'$ implies $\buffer\bsep \cont'
\wlred{\varstr} \emptybuffer\bsep \{ \role_i : \End \}_{i\in I}$ for
some $\varstr$.
\end{definition}

We adopt the term ``live session'' to emphasize the fact that
Definition~\ref{def:liveness} states a \emph{liveness property}: every
finite computation $\emptybuffer\bsep \cont \wlred{\str} \buffer\bsep
\cont'$ can always be extended to a successful computation
$\emptybuffer\bsep \cont \wlred{\str} \buffer\bsep \cont'
\wlred{\varstr} \emptybuffer\bsep \{ \role_i : \End \}_{i\in I}$.
This is stronger than the progress property enforced by other
multi-party session type theories, where it is only required that a
session must never get stuck (but it is possible that some
participants starve for messages that are never sent).
As an example, the session
\[
\cont_1 =
  \{ \role:\rec\recvar(\Out\varrole\val.\recvar \intsum \Out\varrole\valB.\End)\;,\;
   \varrole:\rec\varY(\In\role\val.\varY \extsum \In\role\valB.\End)
\}
\]
is alive because, no matter how many $\val$ messages $\role$ sends,
$\varrole$ can receive all of them and, if $\role$ chooses to send a
$\valB$ message, the interaction terminates successfully for both
$\role$ and $\varrole$.
This example also shows that, despite the fact that session types
describe finite-state processes, the session $\cont_1$ is not
finite-state, in the sense that the set of configurations $\{
(\buffer\bsep \cont') \mid \exists \str, \buffer, \cont':
\varepsilon\bsep \cont_1 \wlred\str \buffer\bsep \cont' \}$ is
infinite. This happens because there is no bound on the size of the
buffer and an arbitrary number of $\val$ messages sent by $\role$ can
accumulate in $\buffer$ before $\varrole$ receives them. As a
consequence, the fact that a session is alive cannot be established in
general by means of a brute force algorithm that checks every
reachable configuration.
By contrast, the session
\[
\cont_2 = \{ \role:\rec\recvar\Out\varrole\val.\recvar\;,\;
\varrole:\rec\varY\In\role\val.\varY \}
\]
which is normally regarded correct in other session type theories, is
not alive because there is no way for $\role$ and $\varrole$ to reach
a successfully terminated state. The point is that hitherto
correctness of session was associated to progress (\ie, the system is
never stuck). This is a weak notion of correctness since, for instance
the session $\cont_2\scup\{\roleR:\In\role c.\End\}$ satisfies
progress even though role $\roleR$ starves waiting for its
input. While in this example starvation is clear since no $c$ message
is ever sent, determining starvation is in general less obvious, as for
\[
\cont_3 = \{ \role:\rec\recvar\Out\varrole\val.\Out\varrole\varval.\recvar\;,\;
\varrole:\rec\varY(\In\role\val.\In\role\varval.\varY+\In\role\varval.\Out\roleR c.\End)\;,\;\roleR:\In\varrole c.\End\}
\]
which satisfies progress, where every input corresponds to a compatible output, and viceversa, but which is not alive.

We remark once again that our work focuses on a single session. In
particular, our definition of live session does not preclude the
existence of a perpetual server that opens an unbounded number of
sessions, each of them having a finite but unbounded length.

We can now define the traces of a session as the set of sequences of
interactions that can occur in every possible reduction. It is
convenient to define the traces of an incorrect (\ie, non-live) session as the empty
set (observe that $\traces(\gtype) \ne \emptyset$ for every $\gtype$).

\begin{definition}[session traces]\label{st}
\[
\traces(\cont) \eqdef
\begin{cases}
  \{ \str \mid \emptybuffer \bsep \cont \wlred{\str}
  \emptybuffer \bsep \{ \role_i : \End \}_{i\in I} \} & \text{if $\cont$ is a live session} \\
  \emptyset & \text{otherwise}
\end{cases}
\]
\end{definition}

It is easy to verify that $\traces(\cont_1) =
\traces((\gaction\role\varrole\valA)\gstar;\gaction\role\varrole\valB)$
while $\traces(\cont_2) = \traces(\cont_3) = \emptyset$ since neither $\cont_2$ nor $\cont_3$ is a live
session.

%%% Local Variables: 
%%% mode: latex
%%% TeX-master: "main"
%%% End: 

\section{Semantic projection}
\label{sec:projection}

In this section we show how to project a global type to the session
types of its participants ---\ie, to a session--- in such a way
that the projection is correct with respect to the global type.
Before we move on, we must be more precise about what we mean by
correctness of a session $\cont$ with respect to a global type
$\gtype$.
In our setting, correctness refers to some relationship between the
traces of $\cont$ and those of $\gtype$.
In general, however, we cannot require that $\gtype$ and $\cont$ have
exactly the same traces: when projecting $\gtype_1\gand\gtype_2$ we
might need to impose a particular order in which the interactions
specified by $\gtype_1$ and $\gtype_2$ must occur (shuffling
condition).
At the same time, asking only $\traces(\cont)\subseteq\traces(\gtype)$
would lead us to immediately lose the exhaustivity property, since
for instance $\set{\role : \Out{\varrole}{\val}.\End\;,\; \varrole :
  \In{\role}{\val}.\End}$ would implement
$\gaction{\role}{\varrole}{\val}\gor\gaction{\role}{\varrole}{\varval}$
even though the implementation systematically exhibits only one
($\gaction{\role}{\varrole}{\val}$) of the specified alternative
behaviors.
In the end, we say that $\cont$ is a correct implementation of
$\gtype$ if: first, every trace of $\cont$ is a trace of $\gtype$
(\emph{soundness}); second, every trace of $\gtype$ is the
permutation of a trace of $\cont$ (\emph{completeness}). Formally:
\[
\traces(\cont)\subseteq\traces(\gtype)\subseteq\perm{\traces(\cont)}
\]
where $\perm\lang$ is the closure of $\lang$ under arbitrary
permutations of the strings in $\lang$:
\[
\perm\lang \eqdef \{
  \gact_1\cdots\gact_n
  \mid
  \text{there exists a permutation $\permutation$ such that $\gact_{\permutation(1)}\cdots\gact_{\permutation(n)}\in\lang$}\}
\]

Since these relations between languages (of traces) play a crucial
role, it is convenient to define a suitable pre-order relation:

\begin{definition}[implementation pre-order]
\label{pog}\label{pos}
We let $\gless{\lang_2}{\lang_1}$ if $\lang_1 \subseteq \lang_2
\subseteq \perm{\lang_1}$ and extend it to global types and sessions
in the natural way, by considering the corresponding sets of
traces. Therefore, we write $\gless\gtype\cont$ if
$\gless{\traces(\gtype)}{\traces(\cont)}$ and similarly for $\gless{\gtype'}\gtype$ and $\gless{\cont'}\cont$.
\end{definition}

It is easy to see that soundness and completeness respectively
formalize the notions of fitness and exhaustivity that we have
outlined in the introduction. As for the remaining three properties
listed in the introduction (\ie, sequentiality, alternativeness, and shuffling), they are entailed by the formalization of the semantics of
a global type in terms of its traces (Definition~\ref{def:traces}). In particular, we
have that soundness implies sequentiality and alternativeness, while
completeness implies shuffling. Therefore, in the formal treatment
that follows we will focus on soundness and completeness as the
only characterizing properties connecting sessions and global types.
The relation $\gless\gtype\cont$ summarizes the fact that $\cont$ is
both sound and complete with respect to $\gtype$, namely that $\cont$
is a correct implementation of the specification $\gtype$.

\begin{table*}[t]
\caption{\label{tab:semantic_projection}\strut Rules for semantic projection.}
\framebox[\textwidth]{
\begin{math}
\displaystyle
\begin{array}{c}
  \inferrule[\rulename{SP-Skip}]{}{
    \project{\cont}{\gend}{\cont}
  }
  \\
  \inferrule[\rulename{SP-Action}]{}{
    \project{
      \{ \role_i : \sesst_i \}_{i\in I}
      \scup
      \{ \role : \sesst \}
      \scup
      \cont
    }{
      \gaction{\{ \role_i \}_{i \in I}}{\role}{\val}
    }{
      \{ \role_i : \Out\role\val.\sesst_i \}_{i\in I}
      \scup
      \{ \role : \In{\{ \role_i \}_{i \in I}}\val.\sesst \}
      \scup
      \cont
    }
  }
  \\\\
  \inferrule[\rulename{SP-Sequence}]{
    \project{\cont}{\gtype_2}{\cont'}
    \\
    \project{\cont'}{\gtype_1}{\cont''}
  }{
    \project{\cont}{\gtype_1\gseq\gtype_2}{\cont''}
  }
  \qquad
  \inferrule[\rulename{SP-Alternative}]{
    \project{\cont}{\gtype_1}{
      \{ \role : \sesst_1 \} \scup \cont'
    }
    \\
    \project{\cont}{\gtype_2}{
      \{ \role : \sesst_2 \} \scup \cont'
    }
  }{
    \project{\cont}{\gtype_1 \gor \gtype_2}{
      \{ \role : \sesst_1 \intsum \sesst_2 \}
      \scup
      \cont'
    }
  }
  \\\\
  \inferrule[\rulename{SP-Iteration}]{
    \project{
      \{ \role : \sesst_1 \intsum \sesst_2 \} \scup \cont
    }{\gtype}{
      \{ \role : \sesst_1 \} \scup \cont
    }
  }{
    \project{
      \{ \role : \sesst_2 \} \scup \cont
    }{\gtype\gstar}{
      \{ \role : \sesst_1 \intsum \sesst_2 \} \scup \cont
    }
  }
  \qquad
  \inferrule[\rulename{SP-Subsumption}]{
    \project{\cont}{\gtype'}{\cont'}
    \\
   \gless {\gtype}{\gtype'}
    \\
    \gless {\cont' }{ \cont''}
  }{
    \project{\cont}{\gtype}{\cont''}
  }
\end{array}
\end{math}
}
\end{table*}

Table \ref{tab:semantic_projection} presents our rules for building the projections of global types. Projecting a global type basically means compiling it to an ``equivalent'' set %composition 
of session types. Since the source language (global types) is equipped with sequential composition while the target language (session types)
is not, it is convenient to parameterize projection on a continuation,
\ie, we consider judgments of the shape:
\[
\project{\cont}{\gtype}{\cont'}
\]
meaning that if $\cont$ is the projection of some $\gtype'$, then
$\cont'$ is the projection of $\gtype\gseq\gtype'$. We say that
$\cont'$ is a \emph{projection} of $\gtype$ with \emph{continuation}
$\cont$. This shape of judgments 
immediately gives us the rule \rulename{SP-Sequence}. 

The projection of an \emph{interaction}
$\gaction{\roles}{\role}{\val}$ adds $\Out\role\val$ in front of the
session type of all the roles in $\roles$, and $\In{\roles}\val$ in
front of the session type of $\role$ (rule \rulename{SP-Action}). For
example we have:
\[\project{\set{\role : \End, \varrole  : \End}}{\gaction{\role}{\varrole}{\val}}{\set{\role : \Out{\varrole}{\val}.\End,~~ \varrole  : \In{\role}{\val}.\End}}\]

An \emph{alternative} $\gtype_1 \gor \gtype_2$ (rule
\rulename{SP-Alternative}) can be projected only if there is a
participant $\role$ that actively chooses among different behaviors by
sending different messages, while all the other participants must
exhibit the same behavior in both branches.  The subsumption rule
\rulename{SP-Subsumption} can be used to fulfill this requirement in
many cases. For example we have
$\project{\cont_0}{\gaction{\role}{\varrole}{\val}}{\set{\role :
    \Out{\varrole}{\val}.\End, \varrole : \In{\role}{\val}.\End}}$ and
$\project{\cont_0}{\gaction{\role}{\varrole}{\varval}}{\set{\role :
    \Out{\varrole}{\varval}.\End, \varrole :
    \In{\role}{\varval}.\End}}$, where $\cont_0=\set{\role : \End,
  \varrole : \End}$. In order to project
$\gaction{\role}{\varrole}{\val}\gor\gaction{\role}{\varrole}{\varval}$
with continuation $\cont_0$ we derive first by subsumption
$\project{\cont_0}{\gaction{\role}{\varrole}{\val}}{\set{\role :
    \Out{\varrole}{\val}.\End\;,\; \varrole : \sesst}}$ and
$\project{\cont_0}{\gaction{\role}{\varrole}{\varval}}{\set{\role :
    \Out{\varrole}{\varval}.\End\;,\; \varrole : \sesst}}$ where
$\sesst=\In{\role}{\val}.\End\extsum\In{\role}{\varval}.\End$. Then we
obtain \[\project{\cont_0}{\gaction{\role}{\varrole}{\val}\gor\gaction{\role}{\varrole}{\varval}}{\set{\role
    : \Out{\varrole}{\val}.\End\intsum\Out{\varrole}{\varval}.\End\;,~~
    \varrole : \sesst}}
\]
Notice that rule \rulename{SP-Alternative} imposes that in alternative branches there must be one \emph{and only one} participant that takes the decision. For instance, the global type
\[
\gaction{\{\role,\varrole\}}\roleR\val
~\gor~
\gaction{\{\role,\varrole\}}\roleR\varval
\]
cannot be projected since we would need a covert channel for $\role$ to agree with $\varrole$ about whether to send to $\roleR$ the message $\val$ or $\varval$.

Rule \rulename{SP-Subsumption} can be easily understood by recalling
that we require %MD0212 eliminato that
a projection $\cont$ of a global type $\gtype$ to
satisfy $\gless {\gtype}\cont$. Therefore if $\cont'$ is a projection
of $\gtype'$ with continuation $\cont$ and $\gless {\gtype}{\gtype'}$,
then $\cont'$ is also a projection of $\gtype$ with continuation
$\cont$. Similarly if $\cont'$ is a projection of $\gtype'$ with
continuation $\cont$ and $\gless {\cont'}{\cont''}$, then also
$\cont''$ is a projection of $\gtype'$ with continuation $\cont$.

To project a \emph{starred} global type we also require that one participant $\role$ chooses between repeating the loop or exiting by sending messages, while the session types of all other participants are unchanged. If $\sesst_1$ and $\sesst_2$ are the session types describing the behavior of $\role$ when it has respectively decided to perform one more iteration or to terminate the iteration, then $\sesst_1 \intsum \sesst_2$ describes the behavior of $\role$ before it takes the decision. The projection rule requires that one execution of $\gtype$ followed by the choice between $\sesst_1$ and $\sesst_2$ projects in a session with type $\sesst_1$ for $\role$. This judgment is possible only if $\sesst_1$ is a recursive type, as expected, and it is the premise of rule \rulename{SP-Iteration}.
For example if $\sesst_1=\Out\varrole\val.\rec\recvar(\Out\varrole\val.\recvar \intsum \Out\varrole\valB.\End)$, $\sesst_2=\Out\varrole\valB.\End$, and $\varsesst=\rec\varY(\In\role\val.\varY \extsum \In\role\valB.\End)$ we can derive $\project{\set{\role:\sesst_1     \intsum \sesst_2,     \varrole:\varsesst}}{\gaction{\role}{\varrole}{\val}}{\set{\role:\sesst_1     , \varrole:\varsesst}}$ and then
\[
\project{\set{\role: \sesst_2,     \varrole:\varsesst}}{(\gaction{\role}{\varrole}{\val})\gstar}{\set{\role:\sesst_1     \intsum \sesst_2,~~ \varrole:\varsesst}}
\]

Notably there is no rule for \quotesymbol\gand, the \emph{both}
constructor. We deal with this constructor by observing that all
interleavings of the actions in $\gtype_1$ and $\gtype_2$ give global
types $\gtype$ such that $\gless {\gtype_1\gand\gtype_2}{\gtype}$, and
therefore we can use the subsumption rule to eliminate every
occurrence of \quotesymbol\gand. For example, to project the global type
$\gaction{\role}{\varrole}{\val}\gand\gaction{\roleR}{\roleS}{\varval}$
we can use
$\gaction{\role}{\varrole}{\val}\gseq\gaction{\roleR}{\roleS}{\varval}$:
since the two actions that compose both global types have disjoint
participants, then the projections of these global types (as well as
that of
$\gaction{\roleR}{\roleS}{\varval}\gseq\gaction{\role}{\varrole}{\val}$)
will have exactly the same set of traces.

Other interesting examples of subsumptions useful  for projecting are  
\begin{eqnarray}
 {\gaction{\roleR}{\role}{\varval}\gseq\gaction{\role}{\varrole}{\val}}&\gless{}{} & (\gaction{\role}{\varrole}{\val}\gseq\gaction{\roleR}{\role}{\varval})\gor(\gaction{\roleR}{\role}{\varval}\gseq\gaction{\role}{\varrole}{\val})\label{uno}\\
{\gaction{\roleR}{\role}{\varval}\gseq(\gaction{\role}{\varrole}{\val}\gor\gaction{\role}{\varrole}{\varval})} &\gless{}{} &  {(\gaction{\roleR}{\role}{\varval}\gseq\gaction{\role}{\varrole}{\val})\gor(\gaction{\roleR}{\role}{\varval}\gseq\gaction{\role}{\varrole}{\varval})}\label{due}
\end{eqnarray}
In~(\ref{uno}) the $\gless{}{}$-larger global type describes the
shuffling of two interactions, therefore we can project one particular
ordering still preserving completeness. In~(\ref{due}) we exploit the
left-distributivity law of regular expressions to push the
{\quotesymbol\gor} construct where the choice is actually being made
(this is possible thanks to the trace semantics we adopt for global
types).

We are interested in projections without continuations, that is, in
judgments of the shape $\project{\set{\role:\End\mid
    \role\in\gtype}}{\gtype}{\cont}$ (where $\role\in\gtype$ means
that $\role$ occurs in $\gtype$) which we shortly write as
\[\project{}{\gtype}{\cont}\]

The mere existence of a projection does not mean that the projection
behaves as specified in the global type. For example, we have
\[
\project{}{
  \gaction{\roleP}{\roleQ}{\val};
  \gaction{\roleR}{\roleS}{\val}
}{
  \{ \roleP : \Out{\roleQ}{\val}.\End,~~
     \roleQ : \In{\roleP}{\val}.\End,~~
     \roleR : \Out{\roleS}{\val}.\End,~~ 
     \roleS : \In{\roleR}{\val}.\End
  \}
}
\]
but the projection admits the trace
$\gaction{\roleR}{\roleS}{\val};\gaction{\roleP}{\roleQ}{\val}$
which is not allowed by the global type. Clearly the problem resides
in the global type, which tries to impose a temporal ordering between
interactions involving disjoint participants.
What we want, in accordance with the traces of a global type
$\gtype_1\gseq\gtype_2$, is that no interaction in $\gtype_2$ can be
completed before all the interactions in $\gtype_1$ are completed.
In more detail:
\begin{iteMize}{$\bullet$}
\item an action $\gaction{\roles}{\role}{\val}$ is completed when the
  participant $\role$ has received the message $\val$ from all the
  participants in $\roles$;
\item if
  $\str\gseq\gaction{\roles}{\role}{\val}\gseq\gaction{\roles'}{\role'}{\varval}\gseq\varstr$
  is a trace of a global type, then either the action
  \mbox{$\gaction{\roles'}{\role'}{\varval}$} cannot be completed before the
  action $\gaction{\roles}{\role}{\val}$ is completed, or they can be
  executed in any order.  The first case requires $\role$ to be either
  $\role'$ or a member of $\roles'$, in the second case the set of
  traces must also contain the trace
  $\str\gseq\gaction{\roles'}{\role'}{\varval}\gseq\gaction{\roles}{\role}{\val}\gseq\varstr$.
\end{iteMize}
This leads us to the following definition of well-formed global type.

\begin{definition}[well-formed global type]\label{def:wf}
We say that a set of traces $\lang$ is \emph{well formed} if
$\str\gseq\gaction{\roles}{\role}{\val}\gseq\gaction{\roles'}{\role'}{\varval}\gseq\varstr
\in \lang$ implies either $\role \in \roles' \cup \{\role'\}$ or
$\str\gseq\gaction{\roles'}{\role'}{\varval}\gseq\gaction{\roles}{\role}{\val}\gseq\varstr
\in \lang$.
We say that a global type $\gtype$ is \emph{well formed} if so is
$\traces(\gtype)$.
\end{definition}

It is easy to decide well-formedness of an arbitrary global type
$\gtype$ by looking at the automaton that recognizes the language of
traces generated by $\gtype$.

Projectability and well-formedness must be kept separate because it is
sometimes necessary to project ill-formed global types. For example,
the ill-formed global type
$\gaction{\roleP}{\roleQ}{\val}\gseq\gaction{\roleR}{\roleS}{\val}$
above is useful to project
$\gaction{\roleP}{\roleQ}{\val}\gand\gaction{\roleR}{\roleS}{\val}$
which is well formed.

Clearly, if a global type is projectable (\ie,
$\project{}{\gtype}{\cont}$ is derivable) then well-formedness of
$\gtype$ is a necessary condition for the soundness and completeness
of its projection (\ie, for $\gless\gtype \cont$). It turns out that
well-formedness is also a sufficient condition for having soundness
and completeness of projections, as stated in the following theorem,
whose proof is the content of Appendix \ref{ptscT}.

\begin{theorem}\label{scT}
If $\gtype$ is well formed and $\project{}{\gtype}{\cont}$, then $\gless{\gtype}{\cont}$. 
\end{theorem}
In summary, if a well-formed global type $\gtype$ is projectable, then
its projection $\cont$ is a \emph{live} session (it cannot generate
the empty set of traces since
$\traces(\gtype)\subseteq\perm{\traces(\cont)}$) which is sound and
complete wrt $\gtype$ and, therefore, satisfies the sequentiality,
alternativeness, and shuffling properties outlined in the
introduction.

\begin{remark}
\label{rem:join}
We now have all the ingredients for showing that actions involving
multiple senders are not redundant, in the sense that they cannot be
encoded in terms of more primitive actions with single senders.
In particular, we show that the global type
\[
\gtype_1 = \gaction{\{\varrole_1,\varrole_2\}}\varrole\valB
\]
is not always equivalent to the expansion
\[
\gtype_2 = \gaction{\varrole_1}\varrole\valB \gand \gaction{\varrole_2}\varrole\valB
\]
despite the fact that, in $\gtype_1$ and $\gtype_2$, the same number
of messages is exchanged between the very same participants.

If we consider the global types $\gtype_1'$ and $\gtype_2'$ defined by:
\[
\gtype_i' =
(\gaction\role{\varrole_1}\valA \gand \gaction\role{\varrole_2}\valA)
\gseq
\gtype_i
\]
we see that $\gtype_1'$ is well formed while $\gtype_2'$ is not. The
reason is because of the trace $\gaction\role{\varrole_1}\valA \gseq
\gaction\role{\varrole_2}\valA \gseq
\gaction{\varrole_1}\varrole\valB\gseq
\gaction{\varrole_2}\varrole\valB \in \traces(\gtype_2')$ where
$\varrole_2 \not\in \{ \varrole, \varrole_1 \}$ and
$\gaction\role{\varrole_1}\valA \gseq
\gaction{\varrole_1}\varrole\valB\gseq \gaction\role{\varrole_2}\valA
\gseq \gaction{\varrole_2}\varrole\valB \not\in \traces(\gtype_2')$.
Basically, both $\gtype_1'$ and $\gtype_2'$ specify the constraint
that no $\valB$ message is received by $\varrole$ \emph{before} both
$\valA$ messages have been received by $\varrole_1$ and $\varrole_2$.
However, in $\gtype_2$ the $\valB$ messages are received by means of
independent actions, and therefore it can happen that the $\valB$
message from $\varrole_1$ is received by $\varrole$ before the $\valA$
message from $\role$ is received by $\varrole_2$, which is exactly the
scenario described by the trace above that is \emph{not} in
$\traces(\gtype_2')$.
The global type $\gtype_1$, on the other hand, specifies that the
receive operation performed by $\varrole$ is considered completed only
when \emph{both} $\valB$ messages from $\varrole_1$ and $\varrole_2$
are available.
The interested reader may compare the projections of $\gtype_1'$ and
$\gtype_2'$ and verify that the one for $\gtype_2'$ does indeed
exhibit an undesired trace.
\end{remark}

We conclude this section by formally characterizing the three kinds of
problematic global types we have %informally
described earlier. We start from the least severe problem and move
towards the more serious ones. Let $\closure\lang$ denote the smallest
well-formed set such that $\lang\subseteq\closure\lang$.

%The case of the theorem above is the the best of the possible cases for a $\gtype$: if $\exists\cont:\gless{\gtype}{\cont}$, then $\gtype$ is implementable. 

\subsection*{No sequentiality.}
Assuming that there is no $\cont$ that is both sound and complete for
$\gtype$, it might be the case that we can find a session whose traces
are complete for $\gtype$ and sound for the global type $\gtype'$
obtained from $\gtype$ by turning some \quotesymbol\gseq's into
\quotesymbol\gand's.  This means that the original global type
$\gtype$ is ill formed, namely that it specifies some sequentiality
constraints that are impossible to implement.
For instance, $\set{\role : \Out{\varrole}{\val}.\End,~ \varrole :
  \In{\role}{\val}.\End,~ \roleR: \Out{\roleS}{\varval}.\End,~ \roleS :
  \In{\roleR}{\varval}.\End}$ is a complete but not sound session for the ill-formed global type
$\gaction{\role}{\varrole}{\textit{a}}\gseq\gaction{\roleR}{\roleS}{\textit{b}}$
(while it is a sound and complete session for
$\gaction{\role}{\varrole}{\textit{a}}\gand\gaction{\roleR}{\roleS}{\textit{b}}$).
We characterize the global types $\gtype$ that present this error as:
\[
\nexists\cont:\gless\gtype\cont
\text{ and }
\exists\cont:\traces(\gtype)\subseteq\traces(\cont)\subseteq\closure{\traces(\gtype)}
\,.
\]

\subsection*{No knowledge for choice.}
In this case every session $\cont$ that is complete for $\gtype$
invariably exhibits some interactions that are not allowed by $\gtype$
despite the fact that $\gtype$ is well formed.
This happens when the global type specifies alternative behaviors, but
some participants do not have enough information to behave
consistently. For example, the global type
\[
(\gaction{\role}{\varrole}{\valA}\gseq
 \gaction{\varrole}{\roleR}{\valA}\gseq
 \gaction{\roleR}{\role}{\valA})
\gor
(\gaction{\role}{\varrole}{\valB}\gseq
 \gaction{\varrole}{\roleR}{\valA}\gseq
 \gaction{\roleR}{\role}{\valB})
\]
mandates that $\roleR$ should send either $\valA$ or $\valB$ in
accordance with the message that $\roleP$ sends to
$\roleQ$. Unfortunately, $\roleR$ has no information as to which
message $\roleQ$ has received, because $\roleQ$ notifies $\roleR$ with
an $\valA$ message in both branches. A complete implementation of this
global type is
\[
\begin{array}{r@{}l}
  \{ & \role : \Out{\varrole}{\val}.(\In{\roleR}{\val}.\End\extsum\In{\roleR}{\varval}.\End)\intsum\Out{\varrole}{\varval}.(\In{\roleR}{\val}.\End\extsum\In{\roleR}{\varval}.\End), \\
& \varrole : \In{\role}{\val}.\Out{\roleR}{\val}.\End\extsum\In{\role}{\varval}.\Out{\roleR}{\val}.\End,
  \roleR: \In{\varrole}{\val}.(\Out{\varrole}{\val}.\End\intsum\Out{\varrole}{\varval}.\End) \}
\end{array}
\]
which also produces the traces
$\gaction{\role}{\varrole}{\valA}\gseq\gaction{\varrole}{\roleR}{\valA}\gseq\gaction{\roleR}{\role}{\valB}$
and
$\gaction{\role}{\varrole}{\valB}\gseq\gaction{\varrole}{\roleR}{\valA}\gseq\gaction{\roleR}{\role}{\valA}$.
We characterize this error as:
\[
\nexists\cont:\traces(\gtype)\subseteq\traces(\cont)\subseteq\closure{\traces(\gtype)}
\text{ and }
\exists\cont:\traces(\gtype)\subseteq\traces(\cont)
\,.
\]
  
\subsection*{No knowledge, no choice.}
In this case we cannot find a complete session $\cont$ for
$\gtype$. This typically means that $\gtype$ specifies some
combination of incompatible behaviors. For example, the global type
$\gaction\role\varrole\valA \gor \gaction\varrole\role\valA$
implies an agreement between $\role$ and $\varrole$ for establishing
who is entitled to send the $\valA$ message. In a distributed
environment, however, there can be no agreement without a previous
message exchange. Therefore, we can either have a sound but not
complete session that implements just one of the two branches (for
example, $\{ \role : \Out\varrole\val.\End, \varrole :
\In\role\val.\End \}$) or a session like $\{ \role :
\Out\varrole\val.\In\varrole\val.\End, \varrole :
\In\role\val.\Out\role\val.\End \}$ where both $\role$ and $\varrole$
send their message but which is neither sound nor complete.
We characterize this error as:
\[
\nexists\cont:\traces(\gtype)\subseteq\traces(\cont)
\,.
\]

%%% Local Variables: 
%%% mode: latex
%%% TeX-master: "main"
%%% End: 

\section{Algorithmic projection}
\label{sec:sal}

We now attack the problem of \emph{computing} the projection of a
global type. We are looking for an algorithm that ``implements'' the
projection rules of Section \ref{sec:projection}, that is, that given
a session continuation $\cont$ and a global type $\gtype$, produces a
projection $\cont'$ such that $\cont\vdash\gtype:\cont'$. In other
terms this algorithm must be sound with respect to the semantic
projection (completeness, that is, returning a projection for every
global type that is semantically projectable, seems out of reach,
yet).

The deduction system in Table~\ref{tab:semantic_projection} is not
algorithmic because of two rules: the rule \rulename{SP-Iteration}
that does not satisfy the subformula property since the context
$\cont$ used in the premises is the result of the conclusion; the rule
\rulename{SP-Subsumption} since it is neither syntax-directed (it is
defined for a generic $\gtype$) nor does it satisfy the subformula
property (the $\gtype'$ and $\cont''$ in the premises are not uniquely
determined).\footnote{The rule \rulename{SP-Alternative} is
  algorithmic: in fact there is a finite number of participants in the
  two sessions of the premises and at most one of them can have
  different session types starting with outputs.\label{foot}}
The latter rule can be expressed as the composition of the two rules
\[\inferrule[\rulename{SP-SubsumptionG}]{
    \project{\cont}{\gtype'}{\cont'}
    \\
   \gless {\gtype}{\gtype'}
      }{
    \project{\cont}{\gtype}{\cont'}
  }\qquad\qquad
  \inferrule[\rulename{SP-SubsumptionS}]{
    \project{\cont}{\gtype}{\cont'}
    \\
     \gless {\cont' }{ \cont''}
  }{
    \project{\cont}{\gtype}{\cont''}
  }\]
  Splitting~\rulename{SP-Subsumption} into~\rulename{SP-SubsumptionG}
  and~\rulename{SP-SubsumptionS} is useful to explain the following
  problems we have to tackle to define an algorithm:
\begin{enumerate}[(1)]
\item How to eliminate \rulename{SP-SubsumptionS}, the subsumption
  rule for sessions.

\item How to define an algorithmic version of \rulename{SP-Iteration},
  the rule for Kleene star.

\item How to eliminate \rulename{SP-SubsumptionG}, the subsumption
  rule for global types.
\end{enumerate}
We address each problem in order and discuss the related rules in the
next sections.

\subsection{Session subsumption}

% We can also do a sound (but incomplete) checking of the
% projectability of a global type by producing a finite set of
% determinizations obtained by considering starred terms as
% unbreakables \textbf{... arrrrrgh}.

Rule~\rulename{SP-SubsumptionS} is needed to project alternative
branches and iterations (a loop is an unbound repetition of
alternatives, each one starting with the choice of whether to enter
the loop or to skip it): each participant different from the one that
actively chooses must behave according to the same session type in
both branches.
More precisely, to project $\gtype_1\gor\gtype_2$ the rule
\rulename{SP-Alternative} requires to deduce for $\gtype_1$ and
$\gtype_2$ the same projection: if different projections are deduced,
then they must be previously subsumed to a common lower bound.
The algorithmic projection of an alternative (see the corresponding
rule in Table~\ref{tab:algorithmic_projection}) allows premises with
two different sessions, but then \emph{merges} them. Of course not every pair
of projections is mergeable. Intuitively, two projections are
mergeable if so are the behaviors of each participant. This requires
participants to respect a precise behavior: as long as a participant
cannot determine in which branch (\ie, projection) it is, then
it must do the same actions in all branches (\ie, projections).
For example, to project $\gtype =
(\gaction{\role}{\varrole}{\valA}\gseq\gaction{\roleR}{\varrole}{\valC}\gseq\dots)\gor(\gaction{\role}{\varrole}{\valB}\gseq\gaction{\roleR}{\varrole}{\valC}\gseq\dots)$
we project each branch separately obtaining
$\cont_1=\{\role:\Out{\roleQ}{\valA}\dots,\varrole:\In{\role}{\valA}.\In{\roleR}{\valC}\dots,\roleR:\Out{\varrole}{\valC}\dots\}$
and
$\cont_2=\{\role:\Out{\roleQ}{\valB}\dots,\varrole:\In{\role}{\valB}.\In{\roleR}{\valC}\dots,\roleR:\Out{\varrole}{\valC}\dots\}$.
Since $\role$ performs the choice, in the projection of $\gtype$ we
obtain $\role:\Out{\roleQ}{\valA}\ldots\oplus\Out{\roleQ}{\valB}\dots$
and we must merge
$\{\varrole:\In{\role}{\valA}.\In{\roleR}{\valC}\dots,\roleR:\Out{\varrole}{\valC}\dots\}$
with
$\{\varrole:\In{\role}{\valB}.\In{\roleR}{\valC}\dots,\roleR:\Out{\varrole}{\valC}\dots\}$.
Regarding $\varrole$, observe that it is the receiver of the message
from $\role$, therefore it becomes aware of the choice and can behave
differently right after the first input operation. Merging its
behaviors yields
$\varrole:\In{\role}{\valA}.\In{\roleR}{\valC}\dots\extsum\In{\role}{\valB}.\In{\roleR}{\valC}\dots$.
Regarding $\roleR$, it has no information as to which choice has been
made by $\roleP$, therefore it must have the same behavior in both
branches, as is the case. Since merging is idempotent, we obtain
$\roleR:\Out{\roleQ}{\valC}\dots$.
%
% Had $\roleR$ sent different messages in the two branches, the
% projection would not have been possible, since $\roleR$ does not know
% yet the choice made by $\role$.
%
In summary, \emph{mergeability} of two branches of an
\quotesymbol\gor{} corresponds to the ``awareness'' of the choice made
when branching (see the discussion in Section~\ref{sec:projection}
about the ``No knowledge for choice'' error), and it is possible when,
roughly, each participant performs the same internal choices and
disjoint external choices in the two sessions.

Special care must be taken when merging external choices to avoid
unexpected interactions that may invalidate the correctness of the
projection. To illustrate the problem consider the session types
$\sesst = \In{\role}{\val}.\In{\varrole}{\varval}.\End$ and $\varsesst
= \In{\varrole}{\varval}.\End$ describing the behavior of a
participant $\roleR$. If we let $\roleR$ behave according to the merge
of $\sesst$ and $\varsesst$, which intuitively is the external choice
$\In{\role}{\val}.\In{\varrole}{\varval}.\End \extsum
\In{\varrole}{\varval}.\End$, it may be possible that the message
$\valB$ from $\roleQ$ is read \emph{before} the message $\valA$ from
$\roleP$ arrives. Therefore, $\roleR$ may mistakenly think that it
should no longer participate to the session, while there is still a
message targeted to $\roleR$ that will never be read. Therefore,
$\sesst$ and $\varsesst$ are \emph{incompatible} and it is not
possible to merge them safely.
On the contrary, $\In{\role}{\val}.\In{\role}{\varval}.\End$ and
$\In{\role}{\varval}.\End$ are compatible and can be merged to
$\In{\role}{\val}.\In{\role}{\varval}.\End \extsum
\In{\role}{\varval}.\End$.
In this case, since the order of messages coming from the same sender
is preserved, it is not possible for $\roleR$ to read the $\valB$
message coming from $\roleP$ before the $\valA$ message, assuming that
$\roleP$ sent both.
More formally:
  
\begin{definition}[compatibility]
  We say that an input $\In{\role}{\val}$ is \emph{compatible} with a
  session type $\sesst$ if either
  \begin{enumerate}[(i)] \item $\In{\role}{\val}$ does not
  occur in $\sesst$, or \item $\sesst = \Output_{i\in I}
  \Out{\role_i}{\val_i}.\sesst_i$ and $\In{\role}{\val}$ is compatible
  with $\sesst_i$ for all $i \in I$, or \item $\sesst = \Input_{i\in
    I} \In{\roles_i}{\val_i}.\sesst_i$ and for all $i \in I$ either
  $\role\in\roles_i$ and $\val\not=\val_i$ or $\role\not\in\roles_i$
  and $\In{\role}{\val}$ is compatible with $\sesst_i$.
  \end{enumerate}

  We say that an input $\In{\roles}{\val}$ is \emph{compatible} with a
  session type $\sesst$ if $\In{\role}{\val}$ is compatible with
  $\sesst$ for some $\role\in\roles$.

  Finally, $\sesst = \Input_{i\in I}
  \In{\roles_i}{\val_i}.\sesst_i\extsum\Input_{j\in J}
  \In{\roles_j}{\val_j}.\sesst_j\text{ and } \varsesst = \Input_{i\in
    I} \In{\roles_i}{\val_i}.\varsesst_i\extsum\Input_{h\in H}
  \In{\roles_h}{\val_h}.\varsesst_h$ are \emph{compatible} if
  $\In{\roles_j}{\val_j}$ is compatible with $\varsesst$ for all $j
  \in J$ and $\In{\roles_h}{\val_h}$ is compatible with $\sesst$ for
  all $h \in H$.
\end{definition}

The merge operator just connects sessions with the \emph{same} output
guards by internal choices and with \emph{compatible} input guards by
external choices:

\begin{definition}[merge]
  The \emph{merge} of $\sesst$ and $\varsesst$, written $\sesst \asup
  \varsesst$, is defined coinductively and by cases on the structure
  of $\sesst$ and $\varsesst$ thus:
\begin{iteMize}{$\bullet$}
\item if $\sesst = \varsesst = \End$, then $\sesst \asup \varsesst
  = \End$;

\item if $\sesst = \Output_{i\in I} \Out{\role_i}{\val_i}.\sesst_i$
  and $\varsesst = \Output_{i\in I}
  \Out{\role_i}{\val_i}.\varsesst_i$, then $\sesst \asup \varsesst =
  \Output_{i\in I} \Out{\role_i}{\val_i}.(\sesst_i \asup \varsesst_i)$;

\item if $\sesst = \Input_{i\in I}
  \In{\roles_i}{\val_i}.\sesst_i+\Input_{j\in J}
  \In{\roles_j}{\val_j}.\sesst_j$ and $\varsesst = \Input_{i\in I}
  \In{\roles_i}{\val_i}.\varsesst_i+\Input_{h\in H}
  \In{\roles_h}{\val_h}.\varsesst_h$ are compatible, then $\sesst
  \asup \varsesst = \Input_{i\in I} \In{\roles_i}{\val_i}.(\sesst_i
  \asup \varsesst_i)+\Input_{j\in J} \In{\roles_j}{\val_j}.\sesst_j
  \extsum \Input_{h\in H} \In{\roles_h}{\val_h}.\varsesst_h.$
\end{iteMize}
We extend merging to sessions so that
$\cont\asup\cont'=\set{\role:\sesst \asup \varsesst\mid
  \role:\sesst\in\cont~\&~\role:\varsesst\in\cont'}.$
\end{definition}

Rules \rulename{AP-Alternative} and \rulename{AP-Iteration} of
Table~\ref{tab:algorithmic_projection} are the algorithmic versions
of~\rulename{SP-Alternative} and~\rulename{SP-Iteration}, but instead
of relying on subsumption they use the merge operator to compute
common behaviors.

\begin{table*}[t]
\caption{\label{tab:algorithmic_projection}\strut Rules for algorithmic projection.}
\framebox[\textwidth]{
\begin{math}
\displaystyle
\begin{array}{c}
  \inferrule[\rulename{AP-Skip}]{}{
    \projecta{\cont}{\gend}{\cont}
  }
  \\
  \inferrule[\rulename{AP-Action}]{}{
    \projecta{
      \{ \role_i : \sesst_i \}_{i\in I}
      \scup
      \{ \role : \sesst \}
      \scup
      \cont
    }{
      \gaction{\{ \role_i \}_{i \in I}}{\role}{\val}
    }{
      \{ \role_i : \Out\role\val.\sesst_i \}_{i\in I}
      \scup
      \{ \role : \In{\{ \role_i \}_{i\in I}}\val.\sesst \}
      \scup
      \cont
    }
  }
  \\\\
  \inferrule[\rulename{AP-Sequence}]{
    \projecta{\cont}{\gtype_2}{\cont'}
    \\
    \projecta{\cont'}{\gtype_1}{\cont''}
  }{
    \projecta{\cont}{\gtype_1\gseq\gtype_2}{\cont''}
  }
  \quad
  \inferrule[\rulename{AP-Alternative}]{
    \projecta{\cont}{\gtype_1}{
      \{ \role : \sesst_1 \} \scup \cont_1
    }
    \\
    \projecta{\cont}{\gtype_2}{
      \{ \role : \sesst_2 \} \scup \cont_2
    }
  }{
    \projecta{\cont}{\gtype_1 \gor \gtype_2}{
      \{ \role : \sesst_1 \intsum \sesst_2 \}
      \scup
      (\cont_1\asup\cont_2)
    }
  }
  \\\\
  \inferrule[\rulename{AP-Iteration}]{
    \projecta{
      \{ \role : \recvar \}
      \scup \{ \role_i : \recvar_i \}_{i\in I}
      \scup \cont
    }{\gtype}{
      \{ \role : \varsesst \}
      \scup \{ \role_i : \varsesst_i \}_{i\in I}
      \scup \cont
    }
  }{
    \projecta{
      \{ \role : \sesst \}
      \scup \{ \role_i : \sesst_i \}_{i\in I}
      \scup \cont
    }{\gtype\gstar}{
      \{ \role : \rec{\recvar}{(\sesst \intsum \varsesst)} \}
      \scup \{ \role_i : \rec{\recvar_i}{(\sesst_i \asup \varsesst_i)} \}_{i\in I}
      \scup \cont
    }
  }
\end{array}
\end{math}
}
\end{table*}

The merge operation is a sound but incomplete approximation of session
subsumption insofar as the merge of two sessions can be undefined even
though the two sessions completed with the participant that makes the
decision have a common lower bound according to $\gless{}{}$. This
implies that there are global types which can be semantically but not
algorithmically projected.\\
Take for example $\gtype_1 \gor \gtype_2$
where $\gtype_1 = \gaction\roleP\roleR\valA \gseq
\gaction\roleR\roleP\valA \gseq \gaction\roleP\roleQ\valA \gseq
\gaction\roleQ\roleR\valB$ and $\gtype_2 = \gaction\roleP\roleQ\valB
\gseq \gaction\roleQ\roleR\valB$.
The behavior of $\roleR$ in $\gtype_1$ and $\gtype_2$ respectively is
$\sesst = \In\roleP\valA.\Out\roleP\valA.\In\roleQ\valB.\End$ and
$\varsesst = \In\roleQ\valB.\End$. Then we see that $\gtype_1 \gor
\gtype_2$ is semantically projectable, for instance by inferring the
behavior $\sesst \extsum \varsesst$ for $\roleR$. However, $\sesst$
and $\varsesst$ are incompatible and $\gtype_1 \gor \gtype_2$ is not
algorithmically projectable.
The point is that the $\gless{}{}$ relation on projections has a
comprehensive perspective of the \emph{whole} session and ``realizes''
that, if $\roleP$ initially chooses to send $\valA$, then $\roleR$
will not receive a $\valB$ message coming from $\roleQ$ until $\roleR$
has sent $\valA$ to $\roleP$. The merge operator, on the other hand,
is defined locally on pairs of session types and ignores that the
$\valA$ message that $\roleR$ sends to $\roleP$ is used to enforce the
arrival of the $\valB$ message from $\roleQ$ to $\roleR$ only
afterwards. For this reason it conservatively declares $\sesst$ and
$\varsesst$ incompatible, making $\gtype_1 \gor \gtype_2$ impossible
to project algorithmically. Appendix \ref{mmc} discusses further
examples illustrating merge and compatibility.

\subsection{Projection of Kleene star}

Since an iteration $\gtype\gstar$ is intuitively equivalent to $\gend
\gor \gtype\gseq\gtype\gstar$ it comes as no surprise that the
algorithmic rule \rulename{AP-Iteration} uses the merge operator.
The use of recursion variables for continuations is also natural: in
the premise we project $\gtype$ taking recursion variables as session
types in the continuation; the conclusion projects $\gtype\gstar$ as
the choice between exiting and entering the loop.
There is, however, a subtle point in this rule that may go unnoticed:
%although in the premises of \rulename{AP-Iteration} the only actions
%and roles taken into account are those occurring in $\gtype$, in its
%conclusion 
the projection of $\gtype\gstar$ may require a continuation
that includes actions and roles that precede $\gtype\gstar$.  The
point can be illustrated by the global type
\[
(\gaction{\role}{\varrole}{\valA} \gseq 
 (\gaction{\role}{\varrole}{\valB})\gstar)\gstar \gseq 
\gaction{\role}{\varrole}\valC
\]
where $\role$ initially decides whether to enter the outermost
iteration (by sending $\valA$) or not (by sending $\valC$). If it
enters the iteration, then it eventually decides whether to also enter
the innermost iteration (by sending $\valB$), whether to repeat the
outermost one (by sending $\valA$), or to exit both (by sending
$\valC$).
Therefore, when we project
$(\gaction{\role}{\varrole}{\varval})\gstar$, we must do it in a
context in which both $\gaction{\role}{\varrole}\valC$ and
$\gaction{\role}{\varrole}{\valA}$ are possible, that is a
continuation of the form $\{\role:\Out\varrole{\valA} \dots \intsum
\Out\varrole{\valC}.\End\}$ even though no $\valA$ is sent by an
action (syntactically) following
$(\gaction{\role}{\varrole}{\varval})\gstar$.
For the same reason, the projection of
$(\gaction{\role}{\varrole}{\valB})\gstar$ in
$(\gaction{\role}{\varrole}{\valA} \gseq
\gaction{\role}{\roleR}{\valA} \gseq
(\gaction{\role}{\varrole}{\valB})\gstar)\gstar \gseq
\gaction{\role}{\varrole}\valC \gseq \gaction{\varrole}{\roleR}\valC$
will need a recursive session type for $\roleR$ %$\roleR:X_{\roleR}$ 
in the continuation.

%\Mariangiola{Sposterei questo discorso alla fine di 6.1, c'e' un esempio piu' semplice?}
%
%
%\textcolor{blue}{(B: bisognera' parlare del fatto che il merge non e'
%  completo. Trovare un esempio di $\cont_,\cont_2$ che hanno un common
%  lower bound ma su cui il merge fallisce)}

\subsection{Global type subsumption}\label{gts}

Elimination of global type subsumption is the most difficult problem
when defining the projection algorithm. While in the case of sessions
the definition of the merge operator gives us a sound ---though not
complete--- tool that replaces session subsumption in very specific
places, we do not have such a tool for global type containment.
This is unfortunate since global type subsumption is necessary to
project several usage patterns (see for example the
inequations~(\ref{uno}) and~(\ref{due})), but most importantly it is
the only way to eliminate $\gand$-types (neither the semantic nor the
algorithmic deduction systems have projection rules for {\quotesymbol\gand}).
The minimal facility  that a projection
algorithm should provide is to feed the algorithmic rules with all the
variants of a global type obtained by replacing occurrences of
$\gtype_1\gand\gtype_2$ by either $\gtype_1\gseq\gtype_2$ or
$\gtype_2\gseq\gtype_1$.
Unfortunately, this is not enough to cover all the occurrences in
which rule~\rulename{SP-SubsumptionG} is necessary. Indeed, while
$\gtype_1\gseq\gtype_2$ and $\gtype_2\gseq\gtype_1$ are  in
many cases projectable (for instance, when $\gtype_1$ and $\gtype_2$ have distinct
roles and are both projectable), there exist $\gtype_1$ and $\gtype_2$
such that $\gtype_1 \gand \gtype_2$ is projectable only by considering
a clever interleaving of the actions occurring in them.
Consider for instance
$\gtype_1=(\gaction{\role}{\varrole}{\val}\gseq\gaction{\varrole}{\roleS}{\valC}\gseq\gaction{\roleS}{\varrole}{\valE})\gor(\gaction{\role}{\roleR}{\varval}\gseq\gaction{\roleR}{\roleS}{\valD}\gseq\gaction{\roleS}{\roleR}{\valF})$
and
$\gtype_2=\gaction{\roleR}{\roleS}{\valG}\gseq\gaction{\roleS}{\roleR}{\valH}\gseq\gaction{\roleS}{\varrole}{\valI}$. The
projection of $\gtype_1\gand\gtype_2$ from the environment
$\set{\varrole:\role!\val.\End, \roleR:\role!\varval.\End}$ can be
obtained only from the interleaving
\[\gaction{\roleR}{\roleS}{\valG}\gseq\gtype_1\gseq\gaction{\roleS}{\roleR}{\valH}\gseq\gaction{\roleS}{\varrole}{\valI}.\]
The reason is that $\varrole$ and $\roleR$ receive messages only in
one of the two branches of the \quotesymbol\gor, so we need to compute
the merge of their types in these branches with their types in the
continuations.
The example shows that to project $\gtype_1 \gand \gtype_2$ it may be
necessary to arbitrarily decompose one or both of $\gtype_1$ and
$\gtype_2$ to find the particular interleaving of actions that can be
projected. As long as $\gtype_1$ and $\gtype_2$ are finite (no
non-trivial iteration occurs in them), we can use a brute force
approach and try to project all the elements in their shuffle, since
there are only finitely many of them.
In general ---\ie, in presence of iteration--- this is not an
effective solution. However, we conjecture that even in the presence
of infinitely many traces one may always resort to the finite case by
considering only zero, one, and two unfoldings of starred global
types. To give a rough idea of the intuition supporting this
conjecture consider the global type $\gtype\gstar \gand \gtype'$: its
projectability requires the projectability of $\gtype'$ (since
$\gtype$ can be iterated zero times), of $\gtype\gand \gtype'$ (since
$\gtype$ can occur only once) and of $\gtype\gseq \gtype$ (since the
number of occurrences of $\gtype$ is unbounded). It is enough to
require also that either $\gtype\gseq(\gtype\gand \gtype')$ or
$(\gtype\gand \gtype')\gseq\gtype$ can be projected, since then the
projectability of either $\gtype^n\gseq(\gtype\gand \gtype')$ or
$(\gtype\gand \gtype')\gseq\gtype^n$ for an arbitrary $n$ follows (see
\iflongversion Appendix~\ref{moreee}).  \else
the appendix in the extended version). 
 \fi

 So we can ---or, conjecture we can--- get rid of all occurrences of
 {\quotesymbol\gand} operators automatically, without losing in
 projectability. However, examples (\ref{uno}) and (\ref{due}) in
 Section~\ref{sec:projection} show that
 rule~\rulename{SP-SubsumptionG} is useful to project also global
 types in which the $\gand$-constructor does not occur. A fully
 automated approach may consider (\ref{uno}) and (\ref{due}) as
 right-to-left rewriting rules that, in conjunction with some other
 rules, form a rewriting system generating a set of global types to be
 fed to the algorithm of Table~\ref{tab:algorithmic_projection}. The
 choice of such rewriting rules must rely on a more thorough study to
 formally characterize the sensible classes of approximations to be
 used in the algorithms.  An alternative approach is to consider a
 global type $\gtype$ as somewhat underspecified, in that it may allow
 for a large number of \emph{different} implementations (exhibiting
 \emph{different} sets of traces) that are sound and
 complete. Therefore, rule~\rulename{SP-SubsumptionG} may be
 interpreted as a human-assisted refinement process where the designer
 of a system proposes one particular implementation
 $\gless{\gtype}\gtype'$ of a system described by $\gtype'$.
In this respect it is interesting to observe that checking whether
$\gless{\lang_2}{\lang_1}$ when $\lang_1$ and $\lang_2$ are regular is
decidable, since this is a direct consequence of the decidability of
the Parikh equivalence on regular languages~\cite{Parikh66}.\iflongversion
\footnote{Whether two regular languages
  have the same Parikh image is decidable. The Parikh image of a word
  $w$ maps each letter of the alphabet to the number of times it
  appears in $w$, the Parikh image of a language is the set of Parikh
  images of all words in the language. By checking Parikh images one
  can check equivalence of languages modulo permutations.}
  \else
\fi
% Finally, whatever the rewriting strategy is used, in case of failure of projectability
% it is always possible to use the help of the programmer, as a last resort.
% If the system designer explicitly disambiguates the
% use global type subsumption by providing an alternative
% global type to be used, then 

\subsection{Properties of the algorithmic rules}

Every deduction of the algorithmic system given in
Table~\ref{tab:algorithmic_projection}, possibly preceded by the
elimination of {\quotesymbol\gand} and other potential sources of failures by
applying the rewritings/heuristics outlined in the previous
subsection, induces a similar deduction using the rules for semantic
projection (Table~\ref{tab:semantic_projection}). For the proof see
Appendix \ref{ptthm:ap}.

\begin{theorem}
  \label{thm:ap}
If $ \projecta{}{\gtype}{\cont}$, then  $\project{}{\gtype}{\cont}$.
\end{theorem}
As a corollary of Theorems~\ref{scT} and~\ref{thm:ap}, we immediately
obtain that the projection $\cont$ of a well-formed $\gtype$ returned by the
algorithm is sound and complete with respect to $\gtype$.

\begin{remark}
  Although every projection of a global type $\gtype$ produced by the
  algorithm is sound and complete with respect to $\gtype$, let us
  stress once more that the algorithm itself is sound but not complete
  with respect to the semantic projection system defined in
  Figure~\ref{tab:semantic_projection}: while every algorithmic
  projection is a semantic projection as well, there exist global
  types which are projectable semantically but not algorithmically.
\end{remark}

%%% Local Variables: 
%%% mode: latex
%%% TeX-master: "main"
%%% End: 

\section{$k$-Exit iterations} %${\color{lucared}\skull}$}
\label{sec:multistar}

The syntax of global types (Table~\ref{tab:gtypes}) includes that of
regular expressions and therefore is expressive enough for describing
any protocol that follows a regular pattern. Nonetheless, the simple
Kleene star prevents us from projecting some useful protocols.
To illustrate the point, suppose we want to describe an interaction
where two participants $\role$ and $\varrole$ alternate in a
negotiation in which each of them may decide to bail out. On $\role$'s
turn, $\role$ sends either a $\bailout$ message or a $\handover$
message to $\varrole$; if a $\bailout$ message is sent, the
negotiation ends, otherwise it continues with $\varrole$ that behaves
in a symmetric way. The global type
\[
  (\gaction\role\varrole\handoverT\gseq\gaction\varrole\role\handoverT)\gstar\gseq
  (\gaction\role\varrole\bailoutT \gor
   \gaction\role\varrole\handoverT\gseq\gaction\varrole\role\bailoutT)
\]
describes this protocol as an arbitrarily long negotiation that may
end in two possible ways, according to the participant that chooses to
bail out.
This global type cannot be projected because of the two occurrences of
the interaction $\gaction\role\varrole\handoverT$, which make it
ambiguous whether $\role$ actually chooses to bail out or to continue
the negotiation. In general, our projection
rules~\rulename{SP-Iteration} and~\rulename{AP-Iteration} make the
assumption that an iteration can be exited in one way only, while in
this case there are two possibilities according to which participant
bails out.
This lack of expressiveness of the simple Kleene star used in a
nondeterministic setting~\cite{Milner84} led researchers to seek for
alternative iterative constructs. One proposal is the \emph{$k$-exit
  iteration}~\cite{BergstraBethkePonse93}, which is a generalization
of the binary Kleene star and has the form
\[
  (\gtype_1,\dots,\gtype_k)\kstar{k}(\gtype_1',\dots,\gtype_k')
\]
indicating a loop consisting of $k$ subsequent phases
$\gtype_1,\dots,\gtype_k$. The loop can be exited just before each
phase through the corresponding $\gtype_i'$.
Formally, the traces of the $k$-exit iteration can be expressed thus:
\[
\begin{array}{@{}rcl@{}}
\traces((\gtype_1,\dots,\gtype_k)\kstar{k}(\gtype_1',\dots,\gtype_k'))
& \eqdef &
\traces((\gtype_1\gseq\dots\gseq\gtype_k)\gstar\gseq(\gtype_1' \gor \gtype_1\gseq\gtype_2' \gor \cdots \gor \gtype_1\gseq\dots\gseq\gtype_{k-1}\gseq\gtype_k'))
\end{array}
\]
and, for example, the negotiation above can be represented as the
global type
\begin{equation}
\label{eqn:negotiation}
  (\gaction\role\varrole\handoverT, \gaction\varrole\role\handoverT)
  \kstar{2}
  (\gaction\role\varrole\bailoutT, \gaction\varrole\role\bailoutT)
\end{equation}
while the unary Kleene star $\gtype\gstar$ can be encoded as
$(\gtype)\kstar{1}(\gend)$.

\begin{table}[t]
\caption{\label{tab:multiexit}\strut Semantic projection of $k$-exit iteration.}
\framebox[\textwidth]{
\begin{math}
\displaystyle
\begin{array}{c}
  \inferrule[\rulename{SP-$k$-Exit Iteration}]{
    \project{\cont}{\gtype_i'}{
      \{ \role_i : \sesstS_i \} \scup
      \{ \role_j : \sesstR_j \}_{j=1,\dots,i-1,i+1,\dots,k} \scup
      \cont'
    }~{}^{(i\in\{1,\dots,k\})}
    \\
    \project{
      \{ \role_2 : \sesstT_2 \intsum \sesstS_2 \} \scup
      \{ \role_i : \sesstR_i \}_{i=1,3,\dots,k} \scup
      \cont'
    }{\gtype_1}{
      \{ \role_1 : \sesstT_1 \} \scup
      \{ \role_i : \sesstR_i \}_{i=2,\dots,k} \scup
      \cont'
    }
    \\
    \project{
      \{ \role_3 : \sesstT_3 \intsum \sesstS_3 \} \scup
      \{ \role_i : \sesstR_i \}_{i=1,2,4,\dots,k} \scup
      \cont'
    }{\gtype_2}{
      \{ \role_2 : \sesstT_2 \} \scup
      \{ \role_i : \sesstR_i \}_{i=1,3,\dots,k} \scup
      \cont'
    }
    \\\\\vdots\\\\
    \project{
      \{ \role_1 : \sesstT_1 \intsum \sesstS_1 \} \scup
      \{ \role_i : \sesstR_i \}_{i=2,\dots,k} \scup
      \cont'
    }{\gtype_k}{
      \{ \role_k : \sesstT_k \} \scup
      \{ \role_i : \sesstR_i \}_{i=1,\dots,k-1} \scup
      \cont'
    }
%     \\
%     \project{
%       \{ \role_{(i\bmod n)+1} : \sesstT_{(i\bmod n)+1} \intsum \sesstS_{(i\bmod n)+1} \} \scup
%       \{ \role_j : \sesstR_j \}_{j\in\{1,\dots,n\}\setminus\{(i\bmod n)+1\}} \scup
%       \cont'
%     }{\gtype_i}{
%       \{ \role_i : \sesstT_i \} \scup
%       \{ \role_j : \sesstR_j \}_{j\in\{1,\dots,n\}\setminus\{i\}},
%       \cont'
%     }~{}^{(i\in\{1,\dots,n\})}
  }{
    \project{\cont}{(\gtype_1,\dots,\gtype_k)\kstar{k}(\gtype_1',\dots,\gtype_k')}{
      \{ \role_1 : \sesstT_1 \intsum \sesstS_1 \} \scup
      \{ \role_i : \sesstR_i \}_{i=2,\dots,k} \scup
      \cont'
    }
  }
\end{array}
\end{math}
}
\end{table}

In our setting, the advantage of the $k$-exit iteration over the
Kleene star is that it syntactically identifies the $k$ points in
which a decision is made by a participant of a multi-party session
and, in this way, it enables more sophisticated projection rules such
as those in Table~\ref{tab:multiexit}.  Albeit intimidating,
rule~\rulename{SP-$k$-Exit Iteration} is just a generalization of
rule~\rulename{SP-Iteration}.
For each phase $i$ a (distinct) participant $\role_i$ is identified:
%\Mariangiola{T ed S secondo me erano invertiti, li ho cambiati}\Luca{Secondo me andavano bene prima, li ho ri-invertiti.}
the participant may decide to exit the loop behaving as $\sesstS_i$ or
to continue the iteration behaving as $\sesstT_i$. While projecting
each phase $\gtype_i$, the participant $\role_{(i\bmod k)+1}$ that
will decide at the next turn is given the continuation
$\sesstT_{(i\bmod k)+1} \intsum \sesstS_{(i\bmod k)+1}$, while the
others must behave according to some $\sesstR_i$ that is the same for
every phase in which they play no active role. Once again,
rule~\rulename{SP-Subsumption} is required in order to synthesize
these behaviors.
For example, the global type~(\ref{eqn:negotiation}) is projected to
\[
\begin{array}{r@{}l}
  \{ & \role : \rec\recvar(\Out\varrole\handover.(\In\varrole\handover.X \extsum \In\varrole\bailout.\End) \intsum \Out\varrole\bailout.\End), \\
     & \varrole : \rec\varY(\In\role\handover.(\Out\role\handover.Y \intsum \Out\role\bailout.\End) \extsum \In\role\bailout.\End) \}
   \end{array}
\]
as one expects.

%%% Local Variables: 
%%% mode: latex
%%% TeX-master: "main"
%%% End: 

\section{Related work}\label{sec:related}

The formalization and analysis of the relation between a global
description of a distributed system and a more machine-oriented
description of a set of components that implements it is a problem
that has been studied in several contexts and by different
communities. In this setting, important properties that are considered
are the \emph{verification} that an implementation satisfies the
specification, the \emph{implementability} of the specification,
% by automatically producing an implementation from it,
and the study of
different properties of the specification that can then be transposed
to each (possibly automatically produced) implementation satisfying
it.  In this work we focused on the implementability problem, and we
tackled it from the ``Web service coordination'' perspective developed
by the community that works on behavioral types and process
algebrae. We are just the latest ones to attack this problem. So many
other communities have been considering it before us that even a
sketchy survey has no chance to be exhaustive.
\iflongversion%%%%%%%%%%%%%%%%%%%%%%%%%%%%%%%%%%LONGVERSION
In what follows we describe two alternative approaches studied by
important communities with a large amount of different and important
contributions, namely the ``automata'' and ``cryptographic protocols''
approaches, and then focus on surveying our ``behavioral types/process
algebra'' approach stressing the relations with the two other
approaches and its peculiarities.

\subsection{Automata approach}
Probably the most extensive research on this problem is pursued by the
``automata/model-checking'' (particularly, finite state automata)
community where special care is paid to software engineering
specification problems. In particular, a lot of research effort has
focused on two specification languages standardized in
telecommunications, the \emph{Message Sequence Charts} (MSCs, ITU
Z.120 standard) and the \emph{Specification and Description Language}
(SDL, ITU Z.100 standard). These respectively play the roles of our
global types and session types. MSCs have become popular in software
development thanks to their graphical representation that depicts
every process by a vertical line and each message as an arrow from the
sender to the receiver process fired according to their top-down
ordering. This standard, included in UML, can also represent other
features, such as timers, atomic events, local/global conditions, but
it can represent neither iterations nor branching. This is why it has
been extended to \emph{Message Sequence Graphs} (MSGs, a special case
of the \emph{High-Level Message Sequence Charts} included
in the Z.120 standard, with equivalent expressivity~\cite{MR97}) which consist of
finite transition systems whose states encapsulate a single MSC:
reaching a given state starts the execution of the embedded MSC whose
termination makes the control move to another state. MSGs play the
same role as our global types.
\begin{figure}
 \centering
  \begin{tikzpicture}[scale=.5] 
    \node[draw, rounded corners=2pt,thick] (A) at (0,0) {
      \begin{tikzpicture}[scale=.5]
        \node[draw, rounded corners=3pt, minimum height=6mm] (s) at (3,4) {\tt seller};
        \node[draw, rounded corners=3pt, minimum height=6mm] (b) at (0,4) {\tt buyer};
        \node[draw, fill=black, minimum width=6mm] (be) at (0,0) {};
        \node[draw, fill=black, minimum width=6mm] (se) at (3,0) {};
        \draw (s) -- (se);
        \draw[<-, >=latex] (0,2) -- (1.5,2) node[above]{\it descr} -- (3,2);
        \draw[<-, >=latex] (0,1) -- (1.5,1) node[above]{\it price} -- (3,1);
        \draw (b) -- (be);
      \end{tikzpicture}
    } ;
    \node[draw, rounded corners=2pt,thick] (B) at (8.5,0) {
      \begin{tikzpicture}[scale=.5]
        \node[draw, rounded corners=3pt, minimum height=6mm] (s) at (3,4) {\tt seller};
        \node[draw, rounded corners=3pt, minimum height=6mm] (b) at (0,4) {\tt buyer};
        \node[draw, fill=black, minimum width=6mm] (be) at (0,0) {};
        \node[draw, fill=black, minimum width=6mm] (se) at (3,0) {};
        \draw (s) -- (se);
        \draw[->, >=latex] (0,2) -- (1.5,2) node[above]{\it offer} -- (3,2);
        \draw[<-, >=latex] (0,1) -- (1.5,1) node[above]{\it price} -- (3,1);
        \draw (b) -- (be);
      \end{tikzpicture}
    } ;
    \node[draw, rounded corners=2pt,thick] (C) at (17,-2.5) {
      \begin{tikzpicture}[scale=.5]
        \node[draw, rounded corners=3pt, minimum height=6mm] (s) at (3,3) {\tt seller};
        \node[draw, rounded corners=3pt, minimum height=6mm] (b) at (0,3) {\tt buyer};
        \node[draw, fill=black, minimum width=6mm] (be) at (0,0) {};
        \node[draw, fill=black, minimum width=6mm] (se) at (3,0) {};
        \draw (s) -- (se);
        \draw[->, >=latex] (0,1) -- (1.5,1) node[above]{\it accept} -- (3,1);
        \draw (b) -- (be);
      \end{tikzpicture}
    } ;
    \node[draw, rounded corners=2pt,thick] (D) at (17,2.5) {
      \begin{tikzpicture}[scale=.5]
        \node[draw, rounded corners=3pt, minimum height=6mm] (s) at (3,3) {\tt seller};
        \node[draw, rounded corners=3pt, minimum height=6mm] (b) at (0,3) {\tt buyer};
        \node[draw, fill=black, minimum width=6mm] (be) at (0,0) {};
        \node[draw, fill=black, minimum width=6mm] (se) at (3,0) {};
        \draw (s) -- (se);
        \draw[->, >=latex] (0,1) -- (1.5,1) node[above]{\it quit} -- (3,1);
        \draw (b) -- (be);
      \end{tikzpicture}
    } ;
    \draw[->, >=latex, thick, rounded corners=6pt] (B) -- (12,-1) |- (C.west);
    \draw[->, >=latex, thick, rounded corners=6pt] (B) -- (12,1) |- (D.west);
    \draw[->, >=latex, thick] (A) -- (B);
%    \draw[->, >=latex, thick] (B.east) arc (90:-180:77pt);
    \path (B) edge[->, out=50, in=140,loop below, min distance=30mm, thick] (B);
\end{tikzpicture}
\caption{MSG of the seller-buyer protocol\label{fig:msg}}
\end{figure}
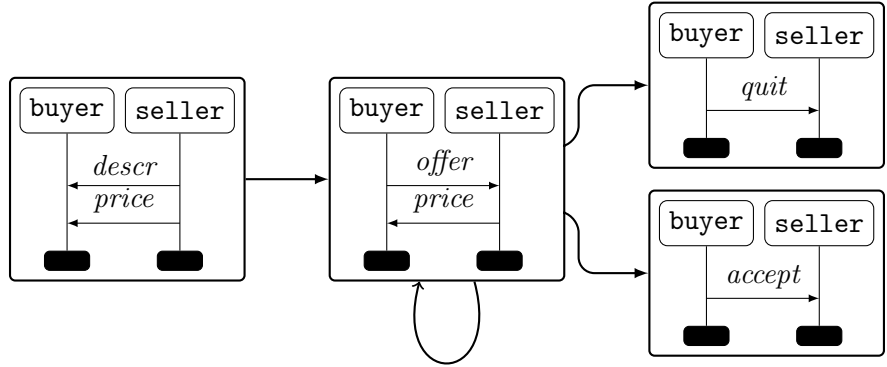

In particular the global type (\ref{spec2}) of the introduction
corresponds to the MSG in Figure~\ref{fig:msg}. The MSG is formed by
four
%three 
states that embed a MSC each. The middle state can loop on
itself or branch in one of the two possible final states.

While a MSG specifies the behavior of a distributed system in terms of
interactions, \emph{Communicating Finite-State Machines} (CFSMs)
---the core theoretical model of SDL--- describe it in terms of its
single components. They are systems of finite state automata that
communicate via asynchronous unbounded FIFO channels. The automata
transitions are labeled by communication primitives which specify the
message and the sender or receiver of it and their execution triggers
a read or write action on the corresponding buffer. A run is
successful if each automaton ends its execution in a final state and
all buffers are empty.
An example is depicted in Figure~\ref{fig:cfm} which implements the
protocol described by the MSG of Figure~\ref{fig:msg}. The automaton
on the top implements the seller while the one on the bottom the
buyer. They communicate by two directional buffers depicted in the
middle of the figure. It is clear that every run of these machines
places at most 2 messages in the buffers and that buffers of length 1
would suffice to implement this protocol without causing deadlocks.

\begin{figure}
\centering
\begin{tikzpicture}
\node[state,initial](q0){$q_0$};  
\node[state](q1)[right=25mm]{$q_1$};  
\node[state](q2)[right=60mm]{$q_2$};  
\node[state,accepting](q3)[above=5mm, right=95mm]{$q_3$};  
\node[state,accepting](q4)[below=5mm, right=95mm]{$q_4$}; 

\path[->] 
    (q0) edge node [above]{$\scriptstyle\Out{\tt buyer}{\it descr}$}   (q1)
    (q1) edge[bend left] node [above]{$\scriptstyle\Out{\tt buyer}{\it price}$}   (q2) 
    (q2) edge[bend left] node [below]{$\scriptstyle\In{\tt buyer}{\it offer}$}   (q1) 
    (q2) edge[bend left] node [above]{$\scriptstyle\In{\tt buyer}{\it accept}$}  (q3) 
    (q2) edge[bend right] node[below]{$\scriptstyle\In{\tt buyer}{\it quit}$}    (q4) 
;
\end{tikzpicture}

\medskip

\begin{tikzpicture}[scale=.3]
 \node(a) at (0,.5){\footnotesize$\gaction\buy\sell{}$};
 \draw (5,0) grid (10,1);
 \node(a) at (20,.5){\footnotesize$\gaction\sell\buy{}$};
 \draw (25,0) grid (30,1); 
\end{tikzpicture}

\medskip

\begin{tikzpicture}
\node[state,initial](q0){$q_0$};  
\node[state](q1)[right=25mm]{$q_1$};  
\node[state](q2)[right=60mm]{$q_2$};  
\node[state,accepting](q3)[above=5mm, right=95mm]{$q_3$};  
\node[state,accepting](q4)[below=5mm, right=95mm]{$q_4$}; 

\path[->] 
    (q0) edge node [above]{$\scriptstyle\In{\tt seller}{\it descr}$}   (q1)
    (q1) edge[bend left] node [above]{$\scriptstyle\In{\tt seller}{\it price}$}   (q2) 
    (q2) edge[bend left] node [below]{$\scriptstyle\Out{\tt seller}{\it offer}$}   (q1) 
    (q2) edge[bend left] node [above]{$\scriptstyle\Out{\tt seller}{\it accept}$}  (q3) 
    (q2) edge[bend right] node[below]{$\scriptstyle\Out{\tt seller}{\it quit}$}    (q4) 
;
\end{tikzpicture}

\caption{CFSMs implementing the seller-buyer protocol.\label{fig:cfm}}
  
\end{figure}
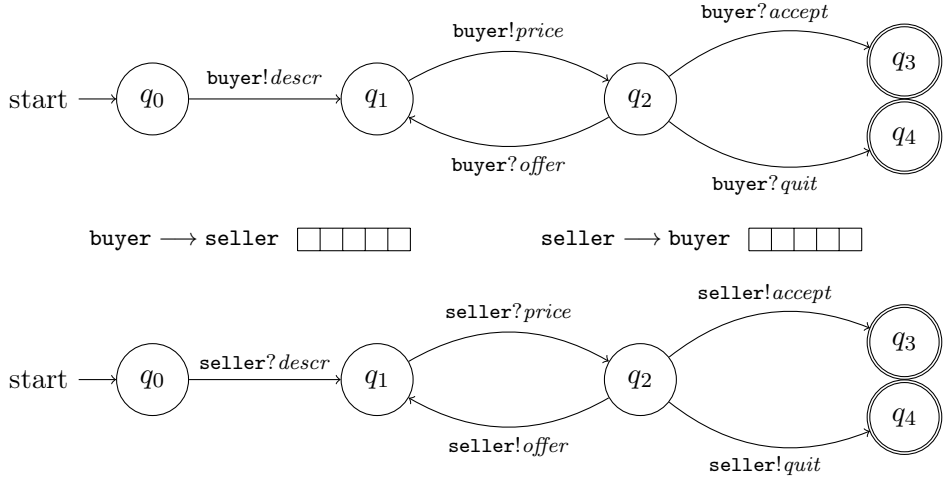

CFSMs essentially are our pre-session types: nothing prevents two
transitions respectively labeled by an input and an output operation
to spring from the same state.  As in our case the interest is in
relating MSGs with CFSMs so that the latter are implementations of the
former. It comes as no surprise that the two formalisms are in general
incomparable. As pointed out in~\cite{GMP03,GM05} this depends on two
fundamental parameters: \emph{control} and \emph{state}. In MSGs (as
well as in our global types) the control of branching is essentially
global since it affects all the roles that occur in future executions,
whereas in CFSMs (as well as in session types) it is inherently local,
since it corresponds to the local transition function. Consequently,
there are MSGs that are not implementable by CFSMs, insofar as the
latter cannot implement global choices (in this work we further
distinguished three degrees of ``non implementability'': no
sequentiality, no knowledge for choice and no knowledge no
choice).
Viceversa, the unbounded buffers of CFSMs provide them with infinite
states and this gives them a Turing equivalent
expressivity~\cite{BZ93}. MSGs, instead, are finitely generated, in
the sense that for every MSG $G$ there exists a finite set
$\mathcal{S}$ of finite MSCs such that any execution of $G$ can be
written as the juxtaposition of the execution of elements in
$\mathcal{S}$. It is then clear that MSGs cannot specify all CFSMs
systems (an example of this is the \emph{alternating bit protocol} in
which a sender resends a message to a receiver since the
acknowledgment arrived too late: to be specified, this protocol needs
MSCs of arbitrary length, see~\cite{GMP03}).  The relative expressive
powers of the two formalisms (finitely generated vs.\ Turing complete)
makes it apparent that the static verification of properties should be
much ``easier'' on MSGs than on CFSMs. Indeed, the expressivity of
CFSMs is used to justify the use of MSGs as an early specification
tool to then be implemented (\ie, projected) into CFSMs: since CFSMs
are Turing complete, all nontrivial behavioral properties --
termination, reachability (\ie, is a given control state reachable?),
deadlock-freedom, boundedness (\ie, is there some bound $n$ such that
every reachable configuration has buffers of size at most $n$?) -- are
undecidable.
Even if some of these properties can be made decidable by some
restrictions (\eg, reachability and safety properties become decidable
with lossy channels, even though liveness properties and boundedness
remain undecidable, see \cite{Sch04}) it is believed that a
satisfactory set of decidable properties can be obtained only with
trivial CFSMs (\eg, with only two processes or with bounded
buffers).
Half-duplex systems~\cite{CeceFinkel05} made of two CFSMs, where each
reachable configuration has at most one buffer non-empty, are closely
related to dyadic sessions and exhibit a number of decidable results
which, unfortunately, do not scale to systems made of an arbitrary
number of machines, even if the half-duplex restriction is maintained.
MSGs have potentially much better properties, since they are finitely
generated. For instance, it is possible to determine the maximum size of the buffers
that each MSC that composes an MSG has to use in order to execute it.
%
%MD0212 eliminati i canali dalle 2 frasi precedenti (sotto nella versione originale)
%
%Half-duplex systems~\cite{CeceFinkel05} made of two CFSMs, where each
%reachable configuration has at most one channel non-empty, are closely
%related to dyadic sessions and exhibit a number of decidable results
%which, unfortunately, do not scale to systems made of an arbitrary
%number of machines, even if the half-duplex restriction is maintained.
%%
%MSGs have potentially much better properties, since they are finitely
%generated. For instance, MSGs have existentially-bounded channels,
%that is, it is possible to determine the maximum size of the buffers
%that each MSC that composes an MSG has to use in order to execute it.
Such properties combined with the fact that the global semantics of
CFSMs/SDL specifications is much more difficult to understand than
that of MSGs, explain why it is very sensible to start with a MSG,
model-check its properties and then implement it as a set of CFSMs.
However, MSGs do not have robust closure properties as, say, regular
languages (the choice we made for our global types). As a consequence,
many variants of MSGs have been proposed in the literature to make
verification and projection effectively and efficiently implementable
(an extensive list of references can be found in~\cite{GM05} and a
more detailed comparison is given in~\cite{GMP03}). In particular if
one considers the restrictions we imposed on our global types, namely
that branching is controlled by one process (they are called
local-choice MSGs), then properties can be model-checked in polynomial
or tractable time (while in the general setting of MSGs many variants
of model-checking are undecidable~\cite{AY99,MPS98}). MSGs can also be
restricted to the class of \emph{regular} MSGs that have robust
properties and for which the implementability by deadlock-free CFSMs
is decidable. In this context however implementability means
generating \emph{the same} set of traces~\cite{Alur00,Alur01}. So we
are in the presence of quite a strict definition of implementability.
Other notions of implementability have been studied yielding different
decidability results (\eg,
see~\cite{Alur00,Alur01}): among these we can cite 
implementations allowed to produce messages not described by the MSG (\ie,
unfit implementations, in the terminology used in our introduction), or
the use of internal communications with messages on a distinct
alphabet to synchronize the system (we avoided this approach which
corresponds to using covert channels), or implementations
allowed to admit deadlocks. 
%MD0212 modificata la frase sopra, la versione precedente e' sotto commentata
%
%Other notions of implementability have been studied yielding different
%decidability results: among these we can cite the case in which the
%implementation may produce messages not described by the MSG (\ie,
%unfit implementation, in the terminology used in our introduction), or
%the use of internal communications with messages on a distinct
%alphabet to synchronize the system (we avoided this approach which
%corresponds to using covert channels), or in which the implementation
%is allowed to admit deadlocks, which improves decidability (\eg,
%see~\cite{Alur00,Alur01}). 
The reader can refer to~\cite{GMP03} for an
extensive survey. However we are not aware of weaker implementability
definitions such as the notions of soundness and completeness we
introduced here. These, besides being an original contribution of our
work, are also the main point that makes algorithmic projection
difficult.
There are some works, such as~\cite{BasuBultan11}, characterizing
classes of CFSMs for which it is possible to decide the conformance
with respect to a global specification (choreography).

\subsection{Cryptographic protocols}

Another domain in which much research on this topic has been done is
the verification of cryptographic protocols. In this context, protocol
narrations, which describe protocols in terms of conversations between
``roles'', must be matched against or implemented into a set of
specifications for the single roles. However the goals pursued in this
area are quite different from the one we outlined in the previous
section, which yields global specification languages with
characteristics different from the one considered by the automata
approach. A first important difference is the content of
messages. While in the automata based research the content of
communications is of lesser importance since it is usually drawn from
a finite set of messages, in the domain of cryptographic protocols
messages are defined by expressive languages that at least include
cryptographic primitives. Whereas message content is richer, the
communication pattern is somewhat simpler since security protocols are
always of finite length, which is why MSCs rather than MSGs are
used. However one has to be very precise about the way an agent
processes its messages (which parts of a message should be extracted
and checked by an agent and how an answer should be computed). This is
why MSCs are annotated or enriched with mechanisms that express the
internal actions to be performed by the agents. This gives raise to
different flavors of formalisms (Figure~\ref{kaochow} gives three
samples of such languages: for more examples and a list of references
see~\cite{CR10}).  These global specifications are then used to verify
security properties and, in some cases, to generate specifications for
the roles composing them. Local specifications are much finer-grained
and lower-level than those used in the automata approach. The details
of internal executions of each agent are exposed and precisely defined
since the overlook of small details may lead to dramatic flaws. This
explains why the palette of languages used to describe the local
behavior appears to be more variegated than in the previous area: the
pioneering work on compilation by Carlsen~\cite{Car94} compiles
protocol narrations into a modal logic of communication; the system
Casper produces CSP descriptions of protocols that are suitable to be
model-checked~\cite{casper} while CAPSL~\cite{MD02} and
CASRUL~\cite{casrl} translate global specifications of protocols, such
as those given in Figure~\ref{kaochow} (HLPSL is the protocol
specification language used by CASRUL), into rewriting systems;
in~\cite{Cal06} MSCs are interpreted into systems of pattern matching
spi-calculus processes~\cite{spi99,HJ04}.  Recent work has shown that
most of the annotation and extensions of MSCs aimed at describing
internal computations, can be computed automatically from the protocol
narration, and thus compile lightly annotated MSCs into an operational
semantics that describes the necessary internal actions~\cite{CR10}.
\begin{figure}
{\scriptsize
\centering
\begin{minipage}[t]{8.2cm}
\begin{alltt}
1 (spec ([c (a b s kas) (kab)]
2        [b (b s kbs) (kab)] [s (a b s kas kbs) ()])
3 [a -> s : a, b, na:nonce]
4 [s -> b : {|a, b, na, kab|} kas, {|a, b, na, kab|} kbs]
5 [b -> a : {|a, b, na, kab|} kas, {|na|} kab, nb:nonce]
6 [a -> b : {|nb|} kab] .)
\end{alltt}
\end{minipage}\bigskip

\begin{minipage}[t]{8.2cm}
\begin{alltt}
PROTOCOL KaoChow;
VARIABLES
	S : Server;
	A, B : Client;
	Na, Nb: Nonce;
	Kab: Skey, CRYPTO, FRESH;	
	F : Field;
	Kas,Kbs : Skey;
DENOTES
	Kas = csk(A): A;
	Kas = ssk(S,A): S;
	Kbs = csk(B): B;
	Kbs = ssk(S,B): S;
ASSUMPTIONS
	HOLDS A: B,S;
MESSAGES
	1. A -> S: A, B, Na;
	2. S -> B: S, {A, B, Na, Kab}Kas%F, {A, B, Na, Kab}Kbs;
	3. B -> A: B, F%{A, B, Na, Kab}Kas, {Na}Kab, Nb;
	4. A -> B: {Nb}Kab;		
GOALS
	SECRET Kab;
	PRECEDES A: B | Na; 
	PRECEDES B: A | Nb, Kab; 
END;
\end{alltt}  
\end{minipage}
\begin{minipage}[t]{5.3cm}
\begin{alltt}
Protocol KaoChow ;
Identifiers
  A,B,S : user ;
  Na,Nb : number;
  Kas,Kbs,Kab : symmetric_key;
knowledge
  A : S,B,Kas ;
  B : A, S, Kbs ;
  S : A, B, Kas, Kbs;
Messages
  1. A -> S : A,B,Na
  2. S -> B : {A,B,Na,Kab}Kas,{A,B,Na,Kab}Kbs
  3. B -> A : {A,B,Na,Kab}Kas,{Na}Kab,Nb
  4. A -> B : {Nb}Kab
Session_instances
[ A:a ; B:b ; S:se ; Kas:kas ; Kbs:kbs ];
Intruder divert , impersonate;
Intruder_knowledge a,b,se;
Goal Short_Term_secret Kab;
Goal B authenticate A on Nb;
\end{alltt}
\end{minipage}
}
\caption{Kao Chow protocol in WPPL, HLPSL and CAPSL (clockwise from top).\label{kaochow}
}
\end{figure}

The degree of detail about local behavior present both in global and
local specification languages is not the only difference with the
previous automata based approach. The other fundamental difference is
the dynamism of the scenarios that both compilation and analysis must
account for. Each role is not necessarily implemented by a single
agent or process but the concurrent presence of several agents that
interpret the same role must be allowed in the system. The system may
include intruder agents that are not described by the global
specification and that may interfere with it; in particular, they may
intercept, read, destroy and forge messages and, more generally,
change the topology of the communications. Furthermore different
executions of the protocol may be not independent as attackers can
store and detour information in one execution to use it in a later
one.

In this context the works closest to our approach are~\cite{McK08}
and~\cite{BCDFL09}.  McCarthy and Krishnamurthi~\cite{McK08} describe
WPPL, a global description language which besides the basic
communication action of MSCs provides actions for role definition and
trust management. WPPL specifications are then projected in local
behaviors defined in CPPL, a domain specific language that describes
cryptographic protocol roles with trust annotations. In their work
they give a nice comparison of their approach with the one used in Web
services that we describe next. In particular, cryptography introduces
information asymmetries (\eg, because of the presence of an intruder
the message received by a role may be different from the one that was
sent to it, or a encrypted message can be received only if the partner
has the corresponding key) that are not handled by existing end-point
projection systems. In a nutshell, in Web services global description
formalisms as well as in the automata approach the focus is on
communication patterns and the communication content is neglected,
while in the realm of cryptographic protocols it is the combination of
the two that really matters.

Bhargavan et al. describe in~\cite{BCDFL09} a compiler from
high-level multi-party session descriptions to custom cryptographic
protocols coded as ML modules. In the generated code each
participant has strong security guarantees for all her/his messages
against any adversary that may control both the network and some
participants to the session.

\subsection{Web services}

\fi%%%%%%%%%%%%%%%%%%%%%%%LONGVERSION%%%%%%%%%%%%%%%%%%%%%%%%

Our work springs from the research done to formally describe and
verify compositions of Web services. This research has mainly
centered on using process algebras to describe and verify  visible local
behavior of services and just recently (all the references date of the
last five years) has started to consider global
\emph{choreographic} descriptions of multiple services and the problem
of their projection. This yielded
the three layered structure depicted in Figure~\ref{globtypes} (courtesy of P.-M.\ Deni\'elou) where a global type describing the choreography is projected into a set of session types that are then used to type-check the processes that implement it (as well as guide their implementation).  
\begin{figure}[t]
\newcommand{\MakeCircle}[2][]{
  \tikz[baseline, remember picture]{
    \node (#1) [draw,circle,minimum size=4em,semithick] {$#2$};
  }
}

\newcommand{\MakeProcess}[2][]{
  \tikz[baseline, remember picture]{
    \node (#1) [draw,minimum width=4em,semithick] {$#2$};
  }
}

\[
\hspace*{12mm}\begin{array}{@{}ccc@{~~~}c@{~~~}rcl@{}}
  & \tikz[baseline,remember picture]{
    \node (gnode) [draw,minimum size=2em,semithick] {$\gtype$};
  } & & \textbf{Global Type} & \gtype &=&
  \begin{array}[t]{l}
  \act{alice}{\color{beppeblue}nat}{bob}
  \gseq\\
  \act{bob}{\color{beppeblue}nat}{carol}
  \end{array}
  \\\\\\
\MakeCircle[t1]{T_{\texttt{alice}}}
& 
\MakeCircle[t2]{T_{\texttt{bob}}}
&
\MakeCircle[t3]{T_{\texttt{carol}}}
&
\textbf{Session Types}
&  T_{\texttt{bob}}&=& 
   \begin{array}[t]{l}
     \In{\texttt{alice}}{\textit{\color{beppeblue}nat}}.\\
     \Out{\texttt{carol}}{\textit{\color{beppeblue}nat}}.\\
     \End
   \end{array}
\\\\\\
\MakeProcess[p1]{P_{\texttt{alice}}}
&
\MakeProcess[p2]{P_{\texttt{bob}}}
&
\MakeProcess[p3]{P_{\texttt{carol}}}
&
\textbf{Processes}&P_{\texttt{bob}}&=&
                          \texttt{receive {\color{lucared} x} from alice;}\\[-2mm]
                      &&&&&&\texttt{send {\color{lucared} x+42} to carol;}\\[-.6mm] 
                      &&&&&&\texttt{end}\\[-5mm] 
\end{array}
\]
\begin{tikzpicture}[remember picture,overlay]
  \draw[->,semithick] (gnode) to node[auto,swap] {Projection} (t1);
  \draw[->,semithick] (gnode) -- (t2);
  \draw[->,semithick] (gnode) -- (t3);
  \draw[<->,semithick] (t1) to node[auto,swap] {Type checking} (p1);
  \draw[<->,semithick] (t2) -- (p2);
  \draw[<->,semithick] (t3) -- (p3);
\end{tikzpicture}
\caption{Global types and multi-party sessions in a nutshell.\label{globtypes}}
\end{figure}
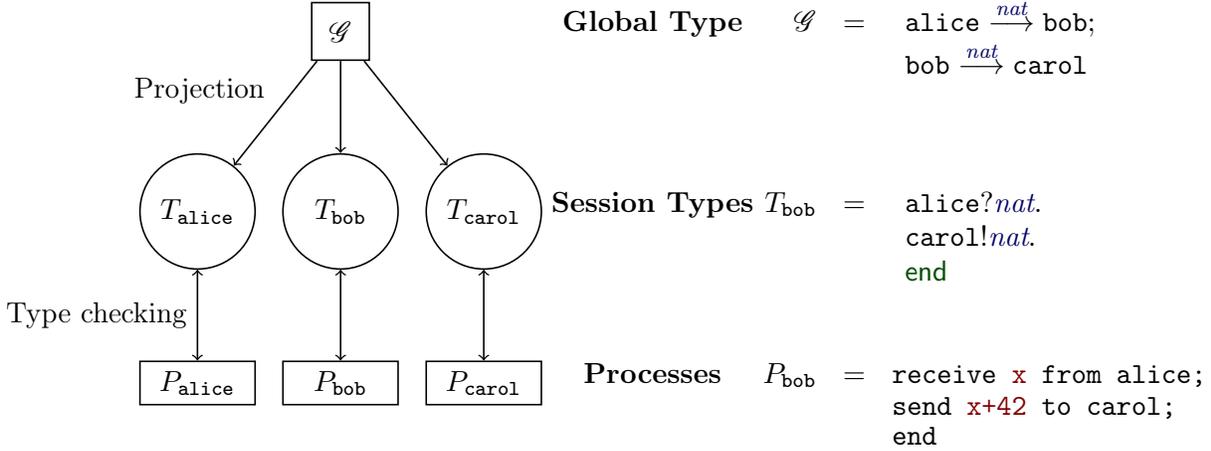
The study thus focuses on defining the relation between the different
layers. Implementability is the relation between the first and second
layer. Here the important properties are that projection
produces systems that are sound and complete with respect to the global
description (in the sense stated by Theorem~\ref{scT}) and
deadlock free (\eg, we rule out specifications such as
$\gaction{\role}{\varrole}{\val}\gor\gaction{\role}{\roleR}{\val}$ when it has no continuation,
since whatever the choice either $\varrole$ or $\roleR$ will be stuck). Typeability is the relation between the
second and third layer. Here the important properties are subject
reduction (well-typed processes reduce only to well-typed processes)
and progress (which in this context implies deadlock freedom). 

Although in this work we disregarded the lower layer of processes, it
is nevertheless an essential component of this research. In
particular, it explains the nature of the messages that characterize
this approach, %research, MD
which are \emph{types}. One of the principal aims of this research,
thus, is to find the right level of abstraction that must be expressed
by types and session types. Consider again Figure~\ref{globtypes}. The
process layer clearly shows the relation between the message received
by \texttt{bob} and the one it sends to \texttt{carol}, but this
relation (actually, any relation) is abstracted away both in the
session and the global type layers. The level of abstraction is
greater than that of cryptographic protocols since values are not
tracked by global descriptions. Although tracking of values could be
partially recovered by resorting to singleton types, there is a
particular class of values that deserves special care and whose
handling is one of the main future challenges of this research, that
is, \emph{channels}.  The goal is to include higher order types in
global specifications thus enabling the transmission of session
channels and therefore the reification of dynamic reconfiguration of
session topology. We thus aim at defining reconfiguration in the
specification itself, as opposed to the case of cryptographic
protocols where the reconfiguration of the communication topology is
considered at meta-level for verification purposes. As a matter of
fact, this feature has already been studied in the literature. For
instance, the extension of {\sc ws-cdl}~\cite{WSCDL} with channel
passing is studied in~\cite{CZ08} (as the automata approach has the
MSC as their reference standard, so the Web service community refers
to the {\sc ws-cdl} standard whose implementability has been studied
in~\cite{QZ07}); the paper that first introduced a global calculus
for session types~\cite{carbone.honda.yoshida:esop07} explicitly
mentions channels in messages that can be sent to other participants
to open new sessions on them.
% n the subsequent
% formalism introduced in~\cite{CHY08} such higher order features are
% used to ``delegate'' the current session (\ie, the communication on
% the current channel) to another participant transparently for the
% current interlocutor. Here we do not consider, yet, such higher-order
% channels (but we are currently studying them) since we preferred to
% follow a semantic approach.
In our opinion the existing works on session types are deeply
syntactic in nature, in the sense that the operators in global types
have been conceived as syntactic adaptations of the corresponding ones
in session types. As a consequence, these operators do not always have
a clear semantic justification. Here we preferred to take a step back
and to start by defining global descriptions whose restrictions are
semantically justified. So we favored a less rich language with few
semantically justified features and leave the addition of more
advanced features for a later time.

Coming back to the comparison of the three approaches, the Web
service-oriented approach shares several features in common with the
other two. As for the automata approach we (in the sense of the Web
service community) focus on the expressiveness of the control, the
possibility of branching and iteration, and the effective
implementability into deadlock-free local descriptions. However the
tendency for Web services is to impose syntactic restrictions from the
beginning rather than study the general case and then devise
appropriate restrictions with the sought properties (in this respect
our work and those of Bravetti, Zavattaro and
Lanese~\cite{BZ07,BZ08,BLZ08} are few exceptions in the panorama of
the Web service approach).  Commonalities with the cryptographic
protocol approach are more technical. In particular we share the
dynamism of the communication topology (with the caveat about whether
this dynamism is performed at the linguistic or meta-linguistic level)
and the robustness with respect to reconfiguration (the projected
session types should ensure that well-typed process will be deadlock
free even in the presence of multiple interleaved sessions and session
delegation, though few works actually enforce this
property~\cite{BCDDDY08,DLY07}). As for cryptographic protocols, this
dynamism is also accounted at level of participants since recent work
in session types studies global descriptions of roles that can then be
implemented by several different agents~\cite{DY11}.
Finally, we take into account the internal behavior of processes
(similarly to what happens for cryptographic protocols) without giving
a precise specification of it but using precise enough (session) types
to prevent any possible internal behavior to disrupt the properties of
systems.\Mariangiola{ho cancellato: and encodable
multi-receivers}
There are also some characteristics that are specific to our approach
such as the exploration of new linguistic features (for instance in
this work we introduced actions with multi-senders)  and a pervasive use of compositional deduction
systems that we inherit from type theory.  We conclude this section
with a more in-depth description of the main references in this
specific area so as to give a more detailed comparison with our work.

\subsubsection{Multi-party global types.}\label{mpst}

%Global types were introduced in \cite{carbone.honda.yoshida:esop07}
%for dyadic sessions and in 
%\cite{CHY08}% for multi-party sessions

Global types were introduced in
\cite{CHY08}  for multi-party sessions, while \cite{carbone.honda.yoshida:esop07} describes a global calculus for dyadic sessions. 
Channels are present in %the global types of 
both
\cite{carbone.honda.yoshida:esop07} and \cite{CHY08}.  However the
language of %global types of 
\cite{carbone.honda.yoshida:esop07}
includes control structures and messages of complex form, since it was
intended to be an executable language to describe Web-service
interactions and, as such, it is directly projected into a language of
processes. Thus it lacks the intermediate layer of
Figure~\ref{globtypes} which is bypassed by providing a more concrete
upper layer.  The three-layered structure of Figure~\ref{globtypes}
faithfully describes the work in~\cite{CHY08} which, nevertheless,
presents several differences with the work presented here. In the
syntax of our work, the global types of~\cite{CHY08} can essentially
be described by the following grammar:
\[
\begin{array}{rcl@{\qquad}l}
  \gtype & ~~::=~~ & \End & \text{(end)} \\
            & \mid &
            \gaction{\role}{\varrole}{k\langle\val\rangle}.\gtype &
            \text{(interaction)} \\
            & \mid & \gtype \gor \gtype & \text{(branching)} \\
            & \mid & \gtype\gand\gtype & \text{(parallel)} \\
            & \mid & X & \text{(variable)} \\
            & \mid & \mu X.\gtype & \text{(recursion)} \\
\end{array}
\]
In a nutshell, sequencing is replaced by prefix actions (terminated by
``\End''), labels are decorated by channels (ranged over by $k$), and
general $\mu$-recursive definitions replace the (less expressive)
Kleene star. Session types (called ``local types'' in~\cite{CHY08})
are even more similar to those presented here, the only difference
being that input/output actions, which have the form $k{?}a.T$ and
$k{!}a.T$, specify channel names rather than participant names.

While the syntactic differences are minimal, it is not so for 
semantic ones. A first important difference is that the global types
of~\cite{CHY08} must satisfy several restrictions:
\begin{enumerate}[(1)]
\item The set of participants of two global types composed in parallel
  must be disjoint. While this restriction clearly simplifies the
  algorithmic projection (the projection of of $\gtype_1\gand\gtype_2$
  reduces to the projection of $\gtype_1\gseq\gtype_2$, \emph{cf.}
  Section~\ref{gts}), it rules out simple protocols such as (\ref{spec1}),
  the very first we presented in this work.

\item The first actions of global types composed by branching must
  specify the same channel, the same sender, the same receiver, and
  distinct messages (actually, labels). Furthermore every participant
  that is neither the first sender nor the first receiver must behave
  the same in all branches.  The use of the same channel and, to a
  lesser extent, of the same senders and receivers for branching is a
  consequence of having adopted the original syntax of labeled
  branching used in the session types
  of~\cite{honda.vasconcelos.kubo:language-primitives}. This first
  restriction forces the adoption of the second one: since session
  type communication specifies channels rather than participants, and
  since the channel is the same in all branches, then the only way for
  the (unique) receiver to distinguish the branches is to receive
  distinct messages on each of them. These restrictions, of syntactic
  origin, are more constraining than ours which just require the
  presence of a single ``decision maker''. The restriction for
  ``passive'' participants to have the same behavior in all branches
  is a quite coarse condition to enforce what in our system is called
  ``mergeability'' (a similar notion of merge was already introduced
  in~\cite{YDBH10,DY11}).
\end{enumerate}

The syntax of global types in~\cite{CHY08} is more constraining than
ours and the semantics of sequential composition is weaker. For
example, two interactions like
$\gaction{\role}{\varrole}{k\langle\val\rangle}.\gaction{\roleR}{\roleS}{k'\langle
  b\rangle}.\End$ are required to happen in the same order as they
occur in the global types only if $k$ and $k'$ are the same
channel. Thus if $k\not=k'$ the participants $\roleP$, $\roleQ$,
$\roleR$, and $\roleS$ can be unrelated. The reason of such a choice
is, once more, due to the fact that global types are designed in
function of the session types as defined
by~\cite{honda.vasconcelos.kubo:language-primitives} where different
channels are typed independently and, thus, sequentiality constraints
can be enforced only between communications on a same channel. It is
interesting to notice that the situation is somehow dual to the one
presented here. While we demand the sequentiality of
\quotesymbol\gseq{} be strictly enforced, we accept any order on
actions composed in parallel by a \quotesymbol\gand. In~\cite{CHY08}
instead, while actions composed in parallel are forced to be
independent (by demanding disjoint participants), any order of the
``sequential'' composition is accepted as long as it happens on
distinct channels.

%MD0212 nel seguito ho sostituito i label etc. in sf con it ma si puo' tornare indietro con la macro \lt 
%
In order to appreciate the usage of global types of~\cite{CHY08} and
their projection let us revisit the paradigmatic example given
in~\cite{CHY08}, according to which two buyers, \ba\ and \bb, wish to
collaborate to buy an item from a seller \sel: \ba\ asks the item to
\sel, which sends a price to both buyers; \ba\ communicates to \bb\
its participation and \bb\ decides either to quit (by sending
\lh{quit} to \sel) or to accept the price by communicating
\lh {ok} and the delivery address to the seller, and expecting a
delivery date. This compound protocol is expressed as the following
global type:
\begin{equation}\label{popl}
\begin{array}{l}
\xact{buyer1}{h\langle\lh {string}\rangle}{seller}.\\
\xact{seller}{k\langle\lh {int}\rangle}{buyer1}.\\
\xact{seller}{k'\langle\lh {int}\rangle}{buyer2}.\\
\xact{buyer1}{l\langle\lh {int}\rangle}{buyer2}.\\
\textcolor{white}{\gor}(\xact{buyer2}{h\langle\lh {quit}\rangle}{seller}.\End)\\
\gor(\xact{buyer2}{h\langle\lh {ok}\rangle}{seller}.\xact{buyer2}{h\langle\lh {string}\rangle}{seller}.\xact{seller}{k'\langle\lh {date}\rangle}{buyer2}.\End)\hspace*{-3mm}
\end{array}
\end{equation}
Notice that in the final branching of the protocol each action
starting a branch is a communication from \bb\ to \sel\ on the same
channel $h$ of two different labels \lh {ok} and \lh {quit}
(strictly speaking, two singleton types whose only value is,
respectively, \lh {ok} and \lh {quit}).  As expected the above
global type is projected into
$$
\begin{array}{rcl}
\sel & \mapsto & h?\lh {string}.k!\lh {int}.k'!\lh {int}.(h?\lh {quit}.\End + h?\lh {ok}.h?\lh {string}.k'!\lh {date}.\End)\\
\ba  & \mapsto & h!\lh {string}.k?\lh {int}.l!\lh {int}.\End\\
\bb  & \mapsto & k'?\lh {int}.l?\lh {int}.(h!\lh {quit}.\End + h!\lh {ok}.h!\lh {string}.k'?\lh {date}.\End)
\end{array}
$$
Notice how participants are replaced by channels. In particular this
implies that \bb\ can distinguish the receptions from \sel\ and \ba\
because they happen on distinct channels.
% But if the same channel
% were used then there would be no way for \bb\ to ensure that the first
% communication was with \sel\ and the second with \ba.
%
Thus, in a sense, explicit channels play the same role of explicit
participants in session types, except that the presence of channels
makes global type analysis more difficult. This explains why such a
feature has been abandoned in~\cite{DY11} (the latest follow up of the
multi-party sessions work) where global types no longer specify
channels and session types use participants instead of channels (see
later on).

We said that~\cite{CHY08} enforces sequentiality only on a per channel
basis. Concretely, this means that for every projection the
interactions on $h$ in the first and fifth or sixth lines of the
protocol in (\ref{popl}) must happen in the same relative order as
they appear in the global types, and the same must hold for
interactions on $k'$ in the third and sixth lines. A rough way to
ensure this property would be to prune all actions that are not on a
given channel and then impose a well-formedness condition akin to the
one we introduced in Definition~\ref{def:wf}. In~\cite{CHY08} much a
finer-grained technique is used: it performs a global analysis of the
dependency relation of a global type and ensures sequentiality on a
given channel by exploiting synchronization information on
interactions occurring also on different channels.
In~\cite{carbone.honda.yoshida:esop07} a stricter condition (dubbed
``well-threadedness'') is described for dyadic sessions, and it
enforces a sequentiality condition similar to our well-formedness.

Finally, we already saw that messages in the global types
of~\cite{CHY08} can be either types (to describe value of the
communication) or labels (to perform branching), but they can also be
channels such as in
$$\cdots{}\,.\act{buyer1}{$l\langle k\rangle$}{buyer2}.\,\cdots,$$
which allows global types to describe \emph{delegation}. Delegation
was introduced in~\cite{CHY08} for multi-party sessions and is
directly inherited from the homonym feature of dyadic
sessions~\cite{honda.vasconcelos.kubo:language-primitives}. A
participant can delegate another agent to play his role in a session
in a way that is transparent for all the remaining participants of the
session. In the example above \ba\ delegates to \bb\ the task to
continue the conversation with \sel\ on $k$.  By allowing higher-order
channels, the concrete topology of communications may dynamically
evolve. To ensure projectability in the presence of such a feature,
further restrictions are required~\cite{CHY08}.

If we focus on semantically justified restrictions, the presence of
channels requires types to be ``well-threaded'' (to avoid that the use
of different channels disrupts the sequentiality constraints of the
specification) and message structures to be used ``coherently'' in
different threads (to assure that a fixed server offers the same
services to different clients), as discussed
in~\cite{carbone.honda.yoshida:esop07}.  We did not include such
features in our treatment since we wanted to study the problems of
sequentiality (which yielded Definition~\ref{def:wf} of well-formed
global type) and of coherence (which is embodied by the subsession
relation whose algorithmic counterpart is the merge operator) in the
simplest possible setting (a single multi-party session) without
further complexity induced by extra features.  As a consequence of
this choice, our merge between session types is a generalization of
the merge in~\cite{YDBH10,DY11} since we allow inputs from different
senders (this is the reason why our compatibility is more demanding
than the corresponding notion in~\cite{YDBH10}). Since our framework
does not include channels, we naturally disregarded any issue arising
from delegation.

Our crusade for simplification did not restrict itself to exclude
features that seemed inessential or too syntax dependent, but it also
used simpler forms of existing constructs. In particular an important
design choice was to use Kleene star instead of more expressive
recursive global types used
in~\cite{%carbone.honda.yoshida:esop07,
CHY08,DY11}.
As an example, the global type describing an arbitrary long
interaction between participants $\roleP$ and $\roleQ$ that $\roleP$
may terminate at any time can be described as
\[
  (\gaction\roleP\roleQ\valA)\gstar \gseq \gaction\roleP\roleQ\valB
\]
in our calculus and as
\[
  \mu X.(\gaction\roleP\roleQ{k\langle\valA\rangle}.X
         \gor
         \gaction\roleP\roleQ{k\langle\valB\rangle}.\End)
\]
in~\cite{CHY08}.
The main advantage of the star over recursion is that it gives us a
fair implementation of the projected specification almost for
free. Fairness seems to us an important ---though mostly neglected by current
literature--- requirement for (multi-party) sessions.
In particular, it allows us to develop a theory where multi-party
sessions preserve a stronger liveness property, namely the potential
to successfully terminate (termination under fairness assumption). A
direct consequence of our choice is that we are capable of projecting
global types where the progress of some participants crucially relies
on the eventual termination of arbitrarily long interactions involving
other participants. For example, the global type
\[
  (\gaction\roleP\roleQ\valA)\gstar \gseq \gaction\roleP\roleQ\valB
  \gseq
  \gaction\roleQ\roleR\valC
\]
is projectable in our theory but its correspondent %equivalent
\[
  \mu X.(\gaction\roleP\roleQ\valA.X \gor
  \gaction\roleP\roleQ\valB.\gaction\roleQ\roleR\valC.\End)
\]
is not in~\cite{CHY08}. The point is that participant $\roleR$ is
waiting for a $\valC$ message that will be sent only if $\roleP$ stops
sending $\valA$ messages to $\roleQ$. This is guaranteed in our theory
but not in~\cite{CHY08} where, in principle, $\roleP$ may send $\valA$
messages to $\roleQ$ forever.

% Without it a session in which a part continuously interacts leaving a
% second one to starve is perfectly acceptable. This is what happens in
% all the papers referred in this subsection.
% Without Kleene star, fairness
% would be more difficult to enforce.  

In general recursion is more expressive than iteration.  For example,
we cannot express non-terminating interactions such as $\mu
X.\gaction\roleP\roleQ\valA.X$. In the present work we regard this
global type as wrong and take the point of view that a session
eventually terminates, although there can be no upper bound to its
duration. Recursion is more flexible when it comes to specifying
iterations with multiple exit paths. For example, the global type
\[
  \mu X.(\gaction\role\varrole\handoverT.
        (\gaction\varrole\role\handoverT.X
        \gor
        \gaction\varrole\role\bailoutT.\End
        )
        \gor
        \gaction\role\varrole\bailoutT.\End)
\]
is a straightforward modeling of the global type that requires
$2$-exit iteration to be projected in our framework (Section
\ref{sec:multistar}).

The exploration of a whole palette of different paradigms for global
and local types and of variations thereof is another element that
distinguishes the research done in the Web service communities from
that in other communities. In particular, the Web service community
does not hesitate to borrow features from other communities and, in
this respect, a remarkable work is the one on dynamic multirole
session types by Deni\'elou and Yoshida~\cite{DY11}. Consider again the %MD0212 aggiunto the
very first example (\ref{spec1}) of the introduction. It consists of
just a single seller and a single buyer. While it seems reasonable to
describe the protocol for a particular seller, it is restrictive to
think that it will handle just one buyer at the time. The idea is that
the seller will interact with a variable number of buyers, all
implementing the same protocol, that will dynamically join and leave
the session. Mutatis mutandis, Deni\'elou and Yoshida propose to
describe the protocol as follows:
\begin{equation}\label{spec11}
\begin{array}{@{}ll}
\forall x:\buy.&(\act{\sell}{\textit{descr}}{$x$}\gand\act{\sell}{\textit{price}}{$x$})\gseq\\
&(\act{$x$}{\textit{accept}}{\sell}\gor\act{$x$}{\textit{quit}}{\sell})
\end{array}\qquad\qquad
\end{equation}
Here \buy\ no longer denotes a single \emph{participant} but rather a
\emph{role} that can be played by different participants (or
processes) ranged over by $x$. The notion of role is extensively used
in the research on the verification of cryptographic protocols,
especially at a meta-linguistic level. Remarkably, Deni\'elou and
Yoshida have internalized it, making it possible to precisely express
the multi-role aspects of an interaction protocol both in global and
in local types. Indeed, the possible projections of the global type
above are:
\[
\begin{array}{rcl}
\sell &\mapsto & \forall x:\buy.\Out{x}{\textit{descr}}.\Out{x}{\textit{price}}.(\In{x}{\textit{accept}}\extsum\In{x}{\textit{quit}})\\
\buy  &\mapsto & \In{\sell}{\textit{descr}}.\In{\sell}{\textit{price}}.(\Out{\sell}{\textit{accept}}\intsum\Out{\sell}{\textit{quit}}) 
\end{array}
\]
and
\[
\begin{array}{rcl}
\sell &\mapsto & \forall x:\buy.\Out{x}{\textit{price}}.\Out{x}{\textit{descr}}.(\In{x}{\textit{accept}}\extsum\In{x}{\textit{quit}})\\
\buy  &\mapsto & \In{\sell}{\textit{price}}.\In{\sell}{\textit{descr}}.(\Out{\sell}{\textit{accept}}\intsum\Out{\sell}{\textit{quit}}) 
\end{array}
\]
Note that session types use participants instead of channels (global
types such as (\ref{spec11}) no longer specify channels). This yields
projections that, apart from the quantifications in \sel, are the same
as those we gave in the introduction for example (\ref{spec1}).
Deni\'elou and Yoshida develop a theory that ensures communication
safety (received messages are of the expected type) and progress
(communications do not get stuck) of sessions in the presence of
dynamically joining and leaving participants.

Finally, although we aimed at simplifying as much as possible, we
still imposed a few restrictions that seemed unavoidable. Foremost, the
sequentiality condition of Section~\ref{sec:projection}: %, that is, that
any two actions that are bound by a semicolon must always appear in
the same order in all traces of (sound and complete) implementations.
Surprisingly, in all current literature of multi-party session types
we are aware of, just one work~\cite{carbone.honda.yoshida:esop07}
enforces the sequential semantics of
\quotesymbol\gseq. In~\cite{carbone.honda.yoshida:esop07} the
sequentiality condition, called \emph{connectedness}, is introduced
(albeit in a simplified setting since---as
in~\cite{honda.vasconcelos.kubo:language-primitives,CHY08}--- instead
of sequential composition the authors consider the simpler case of
prefixed actions) and identified as one of three basic principles for
global descriptions under which a sound and complete implementation
can be defined.  All other (even later) works admit to project, say,
$\gaction{\varrole}{\role}{\val}\gseq\gaction{\roleR}{\role}{\val}$ in
implementations in which $\role$ receives from $\roleR$ before having
received from $\varrole$.  While the technical interest of relaxing
the sequentiality constraint in the interpretation of the
\quotesymbol\gseq{} operator is clear ---it greatly simplifies
projectability--- we really cannot see any semantically plausible
reason to do it.

% Of course all this effort of simplification is worth only if it brings
% clear advantages. First and foremost,
Our simpler setting allows us to give a semantic justification of the
formalism and of the restrictions and the operators we introduced in
it. For these reasons many restrictions that are present in other
formalisms are pointless in our framework. For instance, two global
types whose actions can be interleaved in an arbitrary way (\ie,
composed by \quotesymbol\gand{} in our calculus) can share common
participants in our global types, while in \cite{CHY08} (which use the
parallel operator for \quotesymbol\gand) this is forbidden. So these
works fail to project (actually, they reject) protocols as simple as
the first line of the example given in the specification (\ref{spec1})
in the introduction.  Likewise we can have different receiver
participants in a choice like, for example, the case in which two
cooperating buyers wait for a price from a given seller:
 \[\act{seller}{price}{buyer1}\gseq\act{buyer1}{price}{buyer2}\gor\act{seller}{price}{buyer2}\gseq\act{buyer2}{price}{buyer1}\]
while such a situation is forbidden in~\cite{CHY08}.

Another situation possible in our setting but forbidden
in~\cite{CHY08,DY11}
is to have different sets of participants for alternatives, such as in
the following case where a buyer is notified about a price by the
broker or directly by the seller, but in both cases gives an answer to
the broker:
 \begin{equation}\label{spec12}\begin{array}{l}
(~~\act{seller\!}{agency}{\!broker}\gseq\act{broker\!}{price}{\!buyer}%\\
\;\gor\act{seller\!}{price}{\!buyer})\gseq\\
\act{buyer\!}{answer}{broker}
\end{array}
\end{equation}
A similar situation may arise when choosing between repeating or
exiting a loop:
\begin{equation}\label{spec13}\begin{array}{l}
\act{seller}{agency}{broker};(\act{broker}{offer}{buyer};\act{buyer}{counteroffer}{broker})\gstar;\\
(\act{broker}{result}{seller}\gand\act{broker}{result}{buyer})
\end{array}
\end{equation}
which is again forbidden
in~\cite{CHY08,DY11}.
Note that the interaction following {\quotesymbol\gseq}
in~\eqref{spec12} can be distributed on the two branches, yielding the %MD0212 sostituito the a a
global type
\[
\begin{aligned}
& \act{seller}{agency}{broker}\gseq\act{broker}{price}{buyer}
\gseq\act{buyer}{answer}{broker}
\\
\gor~
& \act{seller}{price}{buyer}\gseq\act{buyer}{answer}{broker}
\end{aligned}
\]
where the two branches involve exactly the same set of participants.
This form is compatible with respect to the notion of projection
in~\cite{CHY08,DY11}. However, the same transformation is not possible
for~\eqref{spec13} because in this case projectability relies on the
fairness assumption. Indeed while we can consider a Kleene star as an
infinite union of finite branches and thus, semantically, add the
continuation to each of these branches, the finiteness of each branch
is guaranteed in our framework but not in~\cite{CHY08,DY11}.

\subsubsection{Choreographies.}

Global types can be seen as choreographies~\cite{WSCDL} describing the
interaction of some distributed processes connected through a private
multi-party session.  Therefore, there is a close relationship between
our work and those by Zavattaro and his
colleagues~\cite{BZ07,LGMZ08,BZ08,BLZ08}, which concern the projection
of choreographies into the contracts of their participants.
The choreography language in these works coincides with our language
of global types (including the use of iteration instead of
recursion). Basically, the only difference at syntactic level is that
interactions have the form $\valA_{\roleP\to\roleQ}$ instead of
$\gaction\roleP\roleQ\valA$.
Just like in our case, a choreography is correct if it preserves the
possibility to reach a state where all of the involved Web services
have successfully terminated.
There are some relevant differences though, starting from
choreographic interactions that invariably involve exactly one sender
and one receiver, while in the present work we allow for multiple
senders.
% and we show that this strictly improves the expressiveness of the
% formalism, which is thus capable of specifying the join of
% independent activities.
%
Other differences concern the communication model and the projection
procedure. In particular, the communication model is synchronous in
\cite{BZ07}, based on FIFO buffers associated with each participant of
a choreography in~\cite{BZ08}, and partially asynchronous
in~\cite{BLZ08} (output actions can fire, and thus drive the choice of
an internal choice, also in the absence of a dual active receiving
action, but their continuation is blocked until the message is
consumed by the receiver).
Our model (Section~\ref{sec:sessions}) closely follows the ones
adopted for multi-party sessions, where there is a single buffer and
we consider the possibility for a receiver to specify the participant
from which a message is expected.
In~\cite{BZ07,LGMZ08,BZ08,BLZ08} the projection procedure is basically
an homomorphism from choreographies to the behavior of their
participants, which is described by a contract language equipped with
parallel composition, while our session types are purely sequential.
\cite{BZ07,BZ08} give no conditions to establish which choreographies
produce correct projections. In contrast, \cite{BLZ08,LGMZ08} define
three \emph{connectedness conditions} that guarantee correctness of
the projection for various (synchronous and asynchronous)
semantics. The interesting aspect is that these conditions are solely
stated on the syntax of the choreography, while we need the
combination of projectability (Table~\ref{tab:semantic_projection})
and well-formedness (Definition~\ref{def:wf}).
Depending on the communication semantics, which can be synchronous or
asynchronous in~\cite{BLZ08,LGMZ08}, the connectedness conditions may
impose different constraints if compared to our well-formedness.  For
example, the choreography
\[
\gaction\roleP\roleQ\valA;\gaction\roleR\roleP\valB
\]
is connected for sequence according to~\cite{BLZ08} but is not well
formed according to Definition~\ref{def:wf}.  This is a consequence of
the different communication models adopted in~\cite{BLZ08} and in the
present work. In~\cite{BLZ08} it is not possible for $\roleP$ to
receive the $\valB$ message from $\roleR$ before $\roleQ$ has received
the $\valA$ message from $\roleP$ because $\roleP$ will block on the
output of $\valA$ until $\roleQ$ receives the message. In our model,
output messages are inserted within the buffer associated with the
session, so the sender can immediately proceed. This corresponds to
the \emph{receiver semantics} in~\cite{LGMZ08}.

The connectedness conditions for alternative choreographies
in~\cite{BLZ08,LGMZ08} impose stricter constraints since they require
that the roles in both branches be the same. Therefore, the two global
types involving the \texttt{broker} participant described by
examples~(\ref{spec12}) and~(\ref{spec13}) are not connected.
Additionally, the fact that these conditions are stated by looking at
the syntax of choreographies may discriminate between equivalent
choreographies.
For example, the choreographies
\[
(\gaction\roleP\roleQ\valA \gand \gaction\roleR\roleS\valA)
\gor
(\gaction\roleP\roleQ\valA \gand \gaction\roleR\roleS\valB)
\text{\quad and\quad}
\gaction\roleP\roleQ\valA \gand (\gaction\roleR\roleS\valA \gor \gaction\roleR\roleS\valB)
\]
are equivalent (they generate the same set of traces), but only the
second one is connected. In the first one, the fact that both branches
emit actions where the sender can be either $\roleP$ or $\roleR$ seems
to suggest the absence of a decision maker, while in fact there is one
($\roleR$). Our definition of well-formedness, being
based on the set of traces generated by a global type rather than its
syntax, does not distinguish between the two choreographies.
As we have shown, a careful projection procedure does not need these
requirements for the projection to respect the choreography.

In~\cite{BZ07} the projection of choreographies with iteration is
taken into account and in~\cite{LGMZ08} it is argued that the
connected conditions scale without problems to this more general
scenario. The authors do not address the limited expressiveness of
single-exit iterations. For example, the first global type at the
beginning of Section~\ref{sec:multistar} yields a deadlocking
projection also for~\cite{BZ07}. Given the similarities between
choreographies and global types it is reasonable to expect that the
adoption of $k$-exit iterations might resolve the issue in their
setting as well.

While discussing MSGs we argued that requiring the specification and
its projection produce the same set of traces (called \emph{standard
  implementation} in~\cite{GMP03}) seemed overly constraining and
advocated a more flexible solution such as the definitions of
soundness and completeness introduced in the present work.
Interestingly, Bravetti, Lanese and Zavattaro~\cite{BLZ08} take the
opposite viewpoint, and make this relation even stricter by
requiring the relation between a choreography and its projection to be
a strong bisimulation.

The problem of analyzing choreographies and characterizing their
properties has been addressed also by the community studying
multiagent systems. In particular, Baldoni \emph{et
  al}.~\cite{BBCDPS09} propose a notion of interoperable choreography
which basically coincides with our notion of liveness: the interaction
between the parties must preserve the ability to reach a state in
which every party has successfully completed its
task. Interoperability induces a notion of conformance between parties
that is similar to our implementation pre-order and to other
refinement relations. The main difference with respect to our work and
those cited above is that in~\cite{BBCDPS09} a choreography is
directly represented as the composition of its participants and their
behavior is described by means of finite-state automata rather than
terms of a process algebra. It appears that the techniques of
choreography projection described in the present paper can be easily
adapted to the context of~\cite{BBCDPS09} and that multiagent systems
might provide an additional playground to further explore and validate
the whole approach.

\subsubsection{Other calculi.}

In this brief overview we focused on works that study the relation
between global specifications and local machine-oriented
implementations. However in the literature there is an important
effort to devise new description paradigms for either global
descriptions or local descriptions. In the latter category we wish to
cite \cite{honda.vasconcelos.kubo:language-primitives,BBDL08},
while~\cite{CP09} seems a natural candidate in which to project an
eventual higher order extension of our global types. For what concerns
global descriptions, the Conversation Calculus \cite{CV09} stands out
for the originality of its approach.

%%% Local Variables: 
%%% mode: latex
%%% TeX-master: "main"
%%% End: 

\section{Conclusion}
\label{sec:conclusion}

We think that the design-by-contract approach advocated
in~\cite{carbone.honda.yoshida:esop07,CHY08} and expanded in later works is a very reasonable way
to implement distributed systems that are correct by construction.
In this work we have presented a theory of global types in an
attempt of better understanding their properties and their
relationship with multi-party session types. We summarize the results
of our investigations in the remaining few lines.
First of all, we have defined a proper algebra of global types whose
operators have a clear meaning. In particular, we distinguish between
sequential composition, which models a strictly sequential execution
of interactions, and unconstrained composition, which allows the
designer to underspecify the order of possibly dependent interactions. 
The semantics of global types is expressed in terms of regular
languages.  Aside from providing an accessible intuition on the
behavior of the system being specified, the most significant
consequence is to induce a \emph{fair} theory of multi-party session
types where correct sessions preserve the ability to reach a state in
which all the participants have successfully terminated. This property
is stronger than the usual progress property within the same session
that is guaranteed in other works. We claim that eventual termination
is both desirable in practice and also technically convenient, because
it allows us to easily express the fact that every participant of a
session makes progress (this is non-trivial, especially in an
asynchronous setting).
We have defined two projection methods from global to session types, a
semantic and an algorithmic one. The former allows us to reason about \emph{which}
are the global types that can be projected, the latter about \emph{how} these types are projected.  This
allowed us to define three classes of flawed global types and to
suggest if and how they can be amended. Most notably, we have
characterized the absence of sequentiality solely in terms of the
traces of global types, while we have not been able to provide similar
trace-based characterizations for the other flaws.
Finally, we have defined a notion of completeness relating a global
type and its implementation which is original to the best of our
knowledge. In other theories we are aware of, this property is either completely
neglected or it is stricter, by requiring the equivalence between the
traces of the global type and those of the corresponding
implementation.

\iffalse
(VECCHIO CUT AND PASTE)

When a choreography designer writes a global type we want \emph{at least} detect two kinds of errors: 
\begin{iteMize}{$\bullet$}
\item if the designer wrote an interaction, then this interaction may happen.
      In other terms we want to detect dead code in global types, that is part
      of the types that cannot be executed by their interaction.

\item if the designer wrote that an interaction must happen after another interaction,
      then it may not be possible that the latter is executed before the first is
      finished. In other terms the communications used to implement the global types
      must enforce the order described by the designer (this is decidable if and only
      if inclusion of $\omega$-regular languages (or of their B\"uchi automata) is. 
\end{iteMize}
Finally we want to warn the designer that even we ensure that all the interaction used may happen, the actual implementation may be required not to use the whole palette of the messages specified in the type but just a subtype of it. (eg,?? un esempio in cui questo sia necessario?)

As a secondary goal we may also envisage a light rewriting of the global type so as to better stress branching points (\eg, $(\gaction{a}{b}\type\gseq\gaction{a}{b}{\true}\gor\gaction{a}{b}\type\gseq\gaction{a}{b}{\false}$ into $\gaction{a}{b}\type\gseq(\gaction{a}{b}{\true}\gor\gaction{a}{b}{\false})$).
\fi

%%% Local Variables: 
%%% mode: latex
%%% TeX-master: "main"
%%% End: 

\section*{Acknowledgments.}
We are indebted to several members of the LIAFA laboratory: Ahmed
Bouajjani introduced us to Parikh's equivalence, Olivier Carton
explained us subtle aspects of the shuffle operator, Mihaela
Sighireanu pointed us several references to global specification
formalisms, while Wies{\l}aw Zielonka helped us with references on trace
semantics. Anca Muscholl helped us on surveying MSCs and Mart{\'\i}n
Abadi and Roberto Amadio with the literature on security
protocols. Finally, Nobuko Yoshida, Roberto Bruni, Ivan Lanese and the
anonymous referees gave us several useful suggestions to improve the
final version of this work. This work was partially supported by the
ANR TYPEX project n.\ ANR-11-BS02-007, by the MIUR Project IPODS, by a
visiting researcher grant of the ``Fondation Sciences Math\'ematiques
de Paris'', and by a visiting professor position of the Universit\'e
Paris Diderot.

%%% Local Variables: 
%%% mode: plain-tex
%%% TeX-master: "mainte"
%%% End: 

\bibliographystyle{alpha}
\bibliography{session,survey}
\clearpage
\appendix
\section{Proof of Theorem  \ref{scT}}\label{ptscT}

For the sake of readability we recall some definitions which will be
largely used.

\begin{definition}$\;$\label{os}
  \begin{enumerate}[(1)]
  \item\label{os1} $
    \perm\lang \eqdef \{
  \gact_1\cdots\gact_n
  \mid
  \text{there exists a permutation $\permutation$ such that $\gact_{\permutation(1)}\cdots\gact_{\permutation(n)}\in\lang$}\}.
$
\item\label{os2} $\closure\lang$ is the smallest
well-formed set such that $\lang\subseteq\closure\lang$.
\end{enumerate}
\end{definition}

The properties stated in the following lemma are easily shown from
Definitions \ref{os}, \ref{pog} and \ref{st}.

\begin{lemma}\label{aux}
The following properties hold:
\begin{enumerate}[\em(1)]
\item\label{aux0} $\closure{(\lang_1\cup{\lang_2})}=\closure{\lang_1}\cup\closure{\lang_2}$.
\item\label{aux1} $\closure{(\lang_1\closure{\lang_2})}=\closure{(\lang_1\lang_2)}$.
\item\label{aux1B} $\closure{(\lang_1\perm{\lang_2})}\subseteq\perm{(\lang_1\lang_2)}$.
\item\label{aux2} $\closure{\lang_1}\subseteq \perm{\lang_2}$ implies $\perm{\lang_1}\subseteq \perm{\lang_2}$.
\item\label{aux3}If $\gless{\lang_2}{\lang_1}$ then
\begin{enumerate}[(a)]
\item\label{aux31}$\gless{\lang_3}{\lang_2}$ implies $\gless{\lang_3}{\lang_1}$;
\item\label{aux32}$\gless{\closure{(\lang_4\lang_1)}}{\lang_3}$ implies $\gless{\closure{(\lang_4\lang_2)}}{\lang_3}$;
\item\label{aux33}$\gless{\closure{(\lang_1\lang_4)}}{\lang_3}$ implies $\gless{\closure{(\lang_2\lang_4)}}{\lang_3}$;
\item\label{aux34}$\gless{\lang_4}{\lang_3}$ implies $\gless{\lang_2\cup\lang_4}{\lang_1\cup\lang_3}$.
\end{enumerate}
\item\label{aux4}$\traces(\{ \role : \sesst_1 \intsum \sesst_2 \}  \scup \cont)
=\traces(\{ \role : \sesst_1  \}  \scup \cont) \cup \traces(\{ \role :  \sesst_2 \}  \scup \cont).$
\end{enumerate}
\end{lemma}

\begin{proof}[Proof of Theorem \ref{scT}] We show:
\begin{center}
{\em If $\project{\cont}{\gtype}{\cont'}$, then $\gless{\closure{(\traces(\gtype)\traces(\cont))}}{\traces(\cont')}$.}
\end{center}
The theorem follows immediately, since by definition if $\gtype$ is well formed, then $\traces(\gtype)=\closure{\traces(\gtype)}$.

The proof is
by induction on the deduction of $\project{\cont}{\gtype}{\cont'}$ and by cases on the last applied rule.\\

Rule \rulename{SP-Skip}:
\inferrule[]{}{
    \project{\cont}{\gend}{\cont}
  } Immediate.\\
  
Rule \rulename{SP-Action}:
\[
 \inferrule[]{}{
    \project{
      \{ \role_i : \sesst_i \}_{i\in I}
      \scup
      \{ \role : \sesst \}
      \scup
      \cont
    }{
      \gaction{\roles}{\role}{\val}
    }{
      \{ \role_i : \Out\role\val.\sesst_i \}_{i\in I}
      \scup
      \{ \role : \In{\roles}\val.\sesst \}
      \scup
      \cont
    }
  }
\]
where $\roles = \{ \role_i \mid i \in I \}$.  We get
$\traces(\gaction{\roles}{\role}{\val})=\set{\gaction{\roles}{\role}{\val}}$
by definition, and
\[
  \traces(\{ \role_i : \Out\role\val.\sesst_i \}_{i\in I}
      \scup
      \{ \role : \In\roles\val.\sesst \}
      \scup
      \cont)
      \subseteq
      \closure{(\set{\gaction{\roles}{\role}{\val}}\traces(\{ \role_i : \sesst_i \}_{i\in I}
      \scup
      \{ \role : \sesst \}
      \scup
      \cont))}
\]
since all actions not involving $\role$ commute with
$\gaction{\roles}{\role}{\val}$, and
\[
\closure{(\set{\gaction{\roles}{\role}{\val}}\traces(\{ \role_i : \sesst_i \}_{i\in I}
      \scup
      \{ \role : \sesst \}
      \scup
      \cont))}\subseteq\perm{\traces(\{ \role_i : \Out\role\val.\sesst_i \}_{i\in I}
      \scup
      \{ \role : \In\roles\val.\sesst \}
      \scup
      \cont)}
\]
by Definition \ref{os}.
      
Rule \rulename{SP-Sequence}:  \inferrule[]{
    \project{\cont}{\gtype_2}{\cont'}
    \\
    \project{\cont'}{\gtype_1}{\cont''}
  }{
    \project{\cont}{\gtype_1\gseq\gtype_2}{\cont''}
  }\\
  By induction
  $\gless{\closure{(\traces(\gtype_1)\traces(\cont'))}}{\traces(\cont'')}$
  and
  $\gless{\closure{(\traces(\gtype_2)\traces(\cont))}}{\traces(\cont')}$,
  which imply
  \begin{iteMize}{$\bullet$}
  \item $\gless{\closure{(\traces(\gtype_1)\closure{(\traces(\gtype_2)\traces(\cont))})}}{\traces(\cont'')}$ by Lemma \ref{aux}(\ref{aux32});
  \item $\gless{\closure{(\traces(\gtype_1)\traces(\gtype_2)\traces(\cont))}}{\traces(\cont'')}$ by Lemma \ref{aux}(\ref{aux1});
  \item $\gless{\closure{(\traces(\gtype_1\gseq\gtype_2)\traces(\cont))}}{\traces(\cont'')}$ by Definition \ref{def:traces}.
   \end{iteMize} 
  
Rule \rulename{SP-Alternative}: \inferrule[]{
    \project{\cont}{\gtype_1}{
      \{ \role : \sesst_1 \} \scup \cont'
    }
    \\
    \project{\cont}{\gtype_2}{
      \{ \role : \sesst_2 \} \scup \cont'
    }
  }{
    \project{\cont}{\gtype_1 \gor \gtype_2}{
      \{ \role : \sesst_1 \intsum \sesst_2 \}
      \scup
      \cont'
    }
  }\\
  By induction
  $\gless{\closure{(\traces(\gtype_1)\traces(\cont))}}{\traces(\{
    \role : \sesst_1 \} \scup\cont')}$ and
  $\gless{\closure{(\traces(\gtype_2)\traces(\cont))}}{\traces(\{
    \role : \sesst_2 \} \scup\cont')}$, which imply
  \begin{iteMize}{$\bullet$}
  \item $\gless{\closure{(\traces(\gtype_1)\traces(\cont))}\cup\closure{(\traces(\gtype_2)\traces(\cont))}}{\traces(\{ \role : \sesst_1 \}\scup\cont')\cup\traces(\{ \role : \sesst_2 \}\scup\cont')}$ by Lemma \ref{aux}(\ref{aux34});
  \item $\gless{\closure{(\traces(\gtype_1)\traces(\cont)\cup\traces(\gtype_2)\traces(\cont))}}{\traces(\{ \role : \sesst_1 \}\scup\cont')\cup\traces(\{ \role : \sesst_2 \}\scup\cont')}$ by Lemma \ref{aux}(\ref{aux0});
   \item $\gless{\closure{(\traces(\gtype_1)\traces(\cont)\cup\traces(\gtype_2)\traces(\cont))}}{\traces(\{ \role : \sesst_1 \intsum\sesst_2 \} \scup\cont')}$ by Lemma \ref{aux}(\ref{aux4});
  \item $\gless{\closure{(\traces(\gtype_1\gor\gtype_2)\traces(\cont))}}{\traces(\{ \role : \sesst_1 \intsum\sesst_2 \} \scup\cont')}$ by Definition \ref{def:traces}.
   \end{iteMize}
   
Rule \rulename{SP-Iteration}:  \inferrule[]{
    \project{
      \{ \role : \sesst_1 \intsum \sesst_2 \} \scup \cont
    }{\gtype}{
      \{ \role : \sesst_1 \} \scup \cont
    }
  }{
    \project{
      \{ \role : \sesst_2 \} \scup \cont
    }{\gtype\gstar}{
      \{ \role : \sesst_1 \intsum \sesst_2 \} \scup \cont
    }
  }
\\
By induction $\gless{\closure{(\traces(\gtype)\traces(\{ \role :
    \sesst_1 \intsum \sesst_2 \} \scup \cont))}}{\traces(\{ \role :
  \sesst_1 \} \scup\cont)}$, \ie:
\begin{enumerate}[1.]
\item $\traces(\{ \role : \sesst_1 \} \scup\cont)\subseteq\closure{(\traces(\gtype)\traces(\{ \role : \sesst_1 \intsum \sesst_2 \}  \scup \cont))}$
\item   
  $\closure{(\traces(\gtype)\traces(\{ \role : \sesst_1 \intsum \sesst_2 \}  \scup \cont))}\subseteq\perm{\traces(\{ \role : \sesst_1 \} \scup\cont)}$.
\end{enumerate}
Notice that by Definition \ref{def:traces} and Lemma \ref{aux}(\ref{aux0}):
\[\closure{(\traces(\gtype^*)\traces(\{ \role : \sesst_2 \} \scup
  \cont))}=\bigcup_{m\geq 0}\closure{(\traces(\gtype^m)\traces(\{
  \role : \sesst_2 \} \scup \cont))}\]
We get: 
\[
\begin{array}{l@{\qquad}l}
\multicolumn{2}{l}{\traces(\{ \role : \sesst_1 \intsum \sesst_2 \}  \scup \cont)}\\
&=\traces(\{ \role : \sesst_1  \}  \scup \cont)\cup \traces(\{ \role :  \sesst_2 \}  \scup \cont)\hfill\text{by Lemma \ref{aux}(\ref{aux4})}\\
&\subseteq\closure{(\traces(\gtype)\traces(\{ \role : \sesst_1 \intsum \sesst_2 \}  \scup \cont))}\cup \traces(\{ \role :  \sesst_2 \}  \scup \cont)\hfill\text{by 1.}\\
&=\closure{(\traces(\gtype)(\traces(\{ \role : \sesst_1 \}  \scup \cont)\cup\traces(\{ \role :  \sesst_2 \}  \scup \cont))}
\cup \traces(\{ \role :  \sesst_2 \}  \scup \cont)\hfill\text{by Lemma \ref{aux}(\ref{aux4})}\\
&=\closure{(\traces(\gtype)\traces(\{ \role : \sesst_1 \}  \scup \cont))}\cup\closure{(\traces(\gtype)\traces(\{ \role :  \sesst_2 \}  \scup \cont))}
\cup \traces(\{ \role :  \sesst_2 \}  \scup \cont)\\
&\hfill\text{by Lemma \ref{aux}(\ref{aux0})}\\
&\subseteq\closure{(\traces(\gtype)\closure{(\traces(\gtype)\traces(\{ \role : \sesst_1 \intsum \sesst_2 \}  \scup \cont))})}
\cup\closure{(\traces(\gtype)\traces(\{ \role :  \sesst_2 \}  \scup \cont))}
\cup \traces(\{ \role :  \sesst_2 \}  \scup \cont)\\
&\hfill\text{by 1.}\\
&=\closure{(\traces(\gtype)\traces(\gtype)\traces(\{ \role : \sesst_1 \intsum \sesst_2 \}  \scup \cont))}
\cup\closure{(\traces(\gtype)\traces(\{ \role :  \sesst_2 \}  \scup \cont))}
\cup \traces(\{ \role :  \sesst_2 \}  \scup \cont)\\
&\hfill\text{by Lemma \ref{aux}(\ref{aux1})}\\
&=\closure{(\traces(\gtype^2)\traces(\{ \role : \sesst_1 \intsum \sesst_2 \}  \scup \cont))}
\cup\closure{(\traces(\gtype)\traces(\{ \role :  \sesst_2 \}  \scup \cont))}
\cup \traces(\{ \role :  \sesst_2 \}  \scup \cont)\\
&\hfill\text{by Definition \ref{def:traces}}
\end{array}
\]
and then by iterating:
\[
\begin{array}{l@{\qquad}l}
\multicolumn{2}{l}{\traces(\{ \role : \sesst_1 \intsum \sesst_2 \}  \scup \cont)}\\
&\subseteq\closure{(\traces(\gtype^{m+1})\traces(\{ \role : \sesst_1 \intsum \sesst_2 \}  \scup \cont))}\cup
\closure{(\traces(\gtype^{m})\traces(\{ \role : \sesst_2 \}  \scup \cont))}\cup\\
&\ldots\cup\closure{(\traces(\gtype)\traces(\{ \role :  \sesst_2 \}  \scup \cont))}\cup \traces(\{ \role :  \sesst_2 \}  \scup \cont)\\
&\subseteq\closure{(\traces(\gtype^*)\traces(\{ \role : \sesst_2 \}  \scup \cont))}.
\end{array}
\]

We show by induction on $m$ that
$\closure{(\traces(\gtype^{m+1})\traces(\{ \role : \sesst_2 \} \scup
  \cont))}\subseteq\perm{(\traces(\{ \role : \sesst_1 \intsum \sesst_2
  \} \scup \cont))}$.
For $m=0$: \[
\begin{array}{llll}\closure{(\traces(\{ \role : \sesst_2 \}  \scup \cont))}&\subseteq&\closure{(\traces(\{ \role : \sesst_1 \intsum \sesst_2 \}  \scup \cont))}&\text{by Lemma \ref{aux}(\ref{aux4}) and Definition \ref{os}(\ref{os2})}\\
&\subseteq&\perm{(\traces(\{ \role : \sesst_1 \intsum \sesst_2 \}  \scup \cont))}&\text{by Definition \ref{os}.}
\end{array}
\]
For $m+1$:
\[
\begin{array}{llll}
  \closure{(\traces(\gtype^{m+1})\traces(\{ \role : \sesst_2 \}  \scup \cont))}&=&\closure{(\traces(\gtype)\traces(\gtype^{m})\traces(\{ \role : \sesst_2 \}  \scup \cont))}&\text{by Definition \ref{def:traces}}\\
  &\subseteq&\closure{(\traces(\gtype)\perm{(\traces(\{ \role : \sesst_1 \intsum
    \sesst_2 \} \scup \cont))})}
  &\text{by induction}\\
  &\subseteq&\perm{(\traces(\gtype)\traces(\{ \role : \sesst_1 \intsum
    \sesst_2 \} \scup \cont))}
  &\text{by Lemma \ref{aux}(\ref{aux1B})}\\
  &\subseteq&\perm{(\traces(\{ \role : \sesst_1 \} \scup \cont))}
  &\text{by 2. and Lemma \ref{aux}(\ref{aux2})}\\
  &\subseteq&\perm{(\traces(\{ \role : \sesst_1 \intsum \sesst_2 \}
    \scup \cont))}
  &\text{by Lemma \ref{aux}(\ref{aux4}).}\\ \\
\end{array}
\]

Rule \rulename{SP-Subsumption}: \inferrule[]{
  \project{\cont}{\gtype'}{\cont'}
  \\
  \gless {\gtype}{\gtype'}
  \\
  \gless {\cont' }{ \cont''} }{ \project{\cont}{\gtype}{\cont''}
}\\
By induction
$\gless{\closure{(\traces(\gtype')\traces(\cont))}}{\traces(\cont')}$,
so by Lemma \ref{aux}(\ref{aux31})
$\gless{\closure{(\traces(\gtype')\traces(\cont))}}{\traces(\cont'')}$. From
$\gless\gtype {\gtype'}$ we conclude
$\gless{\closure{(\traces(\gtype)\traces(\cont))}}{\traces(\cont'')}$
by Lemma \ref{aux}(\ref{aux33}) and (\ref{aux31}).
\end{proof}

\begin{corollary}\label{sc}
If $\cont$ is live and $\project{\cont}{\gtype}{\cont'}$,  then  $\cont'$ is live.
\end{corollary}

\section{More on merge and compatibility}\label{mmc}

We start with an example showing the utility of the compatibility
condition. Let
$\cont_1=\set{\varrole:\In{\role}{\val}. \Out{\roleR}{\varval}.\End,\roleR:\In{\varrole}{\varval}.\End}$
and $\cont_2=\set{
  \varrole:\In{\role}{\valC}. \Out{\role}{\valD}.\Out{\roleR}{\varval}.\End,
  \roleR:\In{\role}{\valE}.\In{\varrole}{\varval}.\End}$. The merge of
$\cont_1$ and $\cont_2$ is undefined, since the session types of
$\roleR$ in $\cont_1$ and $\cont_2$ are not compatible: the problem is
that the input $\In{\varrole}{\varval}$ is not compatible with the
session type $\In{\role}{\valE}.\In{\varrole}{\varval}.\End$.  Let
$\cont$ be the session obtained by adding role $\role$ with the
expected session type to the merge of $\cont_1$ and $\cont_2$
(ignoring the compatibility condition), that is, $\cont=\set{\role:
  \Out{\varrole}{\val}.\End\intsum\Out{\varrole}{\valC}.\In{\varrole}{\valD}.\Out{\roleR}{\valE}.\End,
  \varrole:\In{\role}{\val}. \Out{\roleR}{\varval}.\End\extsum\In{\role}{\valC}. \Out{\role}{\valD}.\Out{\roleR}{\varval}.\End,
  \roleR:\In{\varrole}{\varval}.\End\extsum\In{\role}{\valE}.\In{\varrole}{\varval}.\End}$. Starting
from the empty buffer and $\cont$ we can reach the stuck configuration
in which the buffer contains the action
$\gaction{\role}{\roleR}{\valE}$ and all roles in the session are
typed by $\End$. More precisely if
$\str=\gaction{\role}{\varrole}{\valC}\gseq\gaction{\varrole}{\role}{\valD}$:
\[
\begin{array}{lll}
   \emptybuffer\bsep
 \cont &\wlred{\str}& \gaction{\varrole}{\roleR}{\varval}::\gaction{\role}{\roleR}{\valE}\bsep \set{\role:\End,\varrole:\End,\roleR:\In{\varrole}{\varval}.\End\extsum\In{\role}{\valE}.\In{\varrole}{\varval}.\End}\\
 &\wlred{\gaction{\varrole}{\roleR}{\varval}}& \gaction{\role}{\roleR}{\valE}\bsep \set{\role:\End,\varrole:\End,\roleR:\End}
\end{array}
\]
\ie, participant $\roleR$ chooses the wrong session type, since he is
not aware in which branch he is. Notice that $\cont_1\scup\set{\role:
  \Out{\varrole}{\val}.\End}$ and $\cont_2\scup\set{\role:
  \Out{\varrole}{\valC}.\In{\varrole}{\valD}.\Out{\roleR}{\valE}.\End}$
can be obtained as algorithmic projections of the well-formed global
types
$\gtype_1=\gaction{\role}{\varrole}{\val}\gseq\gaction{\varrole}{\roleR}{\varval}$
and
$\gtype_2=\gaction{\role}{\varrole}{\valC}\gseq(\gaction{\varrole}{\role}{\valD}\gseq\gaction{\role}{\roleR}{\valE}\gand\gaction{\varrole}{\roleR}{\varval})$,
when to project $\gtype_2$ we use the ill-formed global type
$\gaction{\role}{\varrole}{\valC}\gseq\gaction{\varrole}{\role}{\valD}\gseq\gaction{\role}{\roleR}{\valE}\gseq\gaction{\varrole}{\roleR}{\varval}$
(see Subsection \ref{gts}). Using
$\gaction{\role}{\varrole}{\valC}\gseq\gaction{\varrole}{\roleR}{\varval}\gseq\gaction{\varrole}{\role}{\valD}\gseq\gaction{\role}{\roleR}{\valE}$
to project $\gtype_2$ and reasoning as before we get
$\cont'=\set{\role:
  \Out{\varrole}{\val}.\End\intsum\Out{\varrole}{\valC}.\In{\varrole}{\valD}.\Out{\roleR}{\valE}.\End,
  \varrole:\In{\role}{\val}. \Out{\roleR}{\varval}.\End\extsum\In{\role}{\valC}. \Out{\roleR}{\varval}.\Out{\role}{\valD}.\End,
  \roleR:\In{\varrole}{\varval}.\End\extsum\In{\role}{\valE}.\In{\varrole}{\varval}.\End}$. Also
$\cont'$ is not a live session, and since we eliminated $\gand$ from
$\gtype_2$ in all possible ways we see no way to semantically
project $\gtype_1\gor\gtype_2$.

We can semantically but not algorithmically project a slight variation
of the previous example. Let $\cont_3=\set{
  \varrole:\In{\role}{\valC}. \Out{\role}{\valD}.\In{\roleR}{\valF}.\Out{\roleR}{\varval}.\End,
  \roleR:\In{\role}{\valE}.\Out{\varrole}{\valF}.\In{\varrole}{\varval}.\End}$. Notice
that the session types of $\roleR$ in $\cont_1$ and $\cont_3$ are not
compatible. It is easy to verify that choosing $\cont''=\set{
  \varrole:\In{\role}{\val}. \Out{\roleR}{\varval}.\End\extsum\In{\role}{\valC}. \Out{\role}{\valD}.\In{\roleR}{\valF}.\Out{\roleR}{\varval}.\End,
  \roleR:\In{\varrole}{\varval}.\End\extsum\In{\role}{\valE}.\Out{\varrole}{\valF}.\In{\varrole}{\varval}.\End}$
we get
\[
\begin{array}{c}
\gless{\set{\role: \Out{\varrole}{\val}.\End}\scup\cont_1}{\set{\role: \Out{\varrole}{\val}.\End}\scup\cont''}\text{ and }\\
   \gless{\set{\role: \Out{\varrole}{\valC}.\In{\varrole}{\valD}.\Out{\roleR}{\valE}.\End}\scup\cont_3}{\set{\role: \Out{\varrole}{\valC}.\In{\varrole}{\valD}.\Out{\roleR}{\valE}.\End}\scup\cont''}
\end{array}
\]
Notice that $\cont_3\scup\set{\role:
  \Out{\varrole}{\valC}.\In{\varrole}{\valD}.\Out{\roleR}{\valE}.\End}$
can be obtained as the algorithmic projection of the well-formed
global type
$\gtype_3=\gaction{\role}{\varrole}{\valC}\gseq\gaction{\varrole}{\role}{\valD}\gseq\gaction{\role}{\roleR}{\valE}\gseq\gaction{\roleR}{\varrole}{\valF}\gseq\gaction{\varrole}{\roleR}{\varval}$.
Then the global type $\gtype_1\gor\gtype_3$ can be semantically but
not algorithmically projected. It is interesting to observe that in
one branch participant $\roleR$ receives the message $\valB$ from
$\varrole$, in the other branch participant $\roleR$ receives first
the messages $\valE$ from $\role$ and then the message $\valB$ from
$\varrole$. This assures that $\roleR$ always chooses the right
session type.  Comparing $\gtype_2$ and $\gtype_3$ of previous
examples one can see how the addition of the action
$\gaction{\roleR}{\varrole}{\valF}$ introduces a sequentialization
which is the key of projectability.

\section{More on the elimination of $\gand$}\label{moreee}

We conjecture that the following rewriting rules (together with the
symmetric ones) are necessary and sufficient in order to eliminate
\quotesymbol\gand{} from global types:
\[\begin{array}{lll@{\qquad\qquad\qquad\qquad}lll}
\gtype\gand\gtype'&\mapsto&\gtype\gseq\gtype'&(\gtype_1\gor\gtype_2)\gand\gtype&\mapsto&(\gtype_1\gand\gtype)\gor(\gtype_2\gand\gtype)\\
(\gtype_1\gseq\gtype_2)\gand\gtype&\mapsto&(\gtype_1\gand\gtype)\gseq\gtype_2&\gtype\gstar\gand\gtype'&\mapsto&(\gtype\gand\gtype')\gseq\gtype\gstar\gor\gtype'\\
(\gtype_1\gseq\gtype_2)\gand\gtype&\mapsto&\gtype_1\gseq(\gtype_2\gand\gtype)&
\gtype\gstar\gand\gtype'&\mapsto&\gtype\gstar\gseq(\gtype\gand\gtype')\gor\gtype'
\end{array}\]
Sometimes \quotesymbol\gand{} with stars can be dealt with using the first rule in the right way. 
 The global type
 $(\gaction{\role}{\varrole}{\val})^*\gand\gaction{\role}{\varrole}{\varval}$
 sequentialized as
 $(\gaction{\role}{\varrole}{\val})^*\gseq\gaction{\role}{\varrole}{\varval}$
 is algorithmically projected from $\set{\role:\End,\varrole:\End}$,
 while
 $\gaction{\role}{\varrole}{\varval}\gseq(\gaction{\role}{\varrole}{\val})^*$
 is not algorithmically projected from
 $\set{\role:\End,\varrole:\End}$.  Vice versa the global type
 $(\gaction{\role}{\varrole}{\val})^*\gand\gaction{\roleR}{\roleS}{\varval}$
 sequentialized as
 $(\gaction{\role}{\varrole}{\val})^*\gseq\gaction{\roleR}{\roleS}{\varval}$
 is not algorithmically projected from $\set{\role:\Out{\varrole}{\valC}.\End,\varrole:\In{\role}{\valC}.\End}$,
 while
 $\gaction{\roleR}{\roleS}{\varval}\gseq(\gaction{\role}{\varrole}{\val})^*$
 is algorithmically projected from
 $\set{\role:\Out{\varrole}{\valC}.\End,\varrole:\In{\role}{\valC}.\End}$.
 
 The following example shows the utility of the last two rewriting rules to project stars. Let $\cont=\set{\roleS: \Out{\varrole_1}{\valD}.\Out{\roleR_1}{\valD}.\Out{\varrole_2}{\valD}.\Out{\roleR_2}{\valD}.\End, \varrole_1:\In{\roleS}{\valD}, \roleR_1:\In{\roleS}{\valD}, \varrole_2:\In{\roleS}{\valD}, \roleR_2:\In{\roleS}{\valD}}$, ${\mathcal A}_i=\gaction{\role_i}{\varrole_i}{\val}\gor\gaction{\role_i}{\roleR_i}{\varval}$ for $i=1,2$ and $\gtype_1=\gaction{\set{\varrole_1,\roleR_1}}{\roleS}{\valC}\gseq{\mathcal A}_1$, $\gtype_2={\mathcal A}_2\gseq\gaction{\roleS}{\varrole_1}{\valE}\gseq\gaction {\roleS}{\roleR_1}{\valE}$, $\gtype=\gaction{\set{\varrole_2,\roleR_2}}{\roleS}{\valF}\gseq{\mathcal A}_1\gseq{\mathcal A}_2$. The only way to eliminate \quotesymbol\gand{} from $(\gtype_1\gseq\gtype_2)\gstar\gand\gtype$ and obtain a global type projectable  with the continuation $\cont$ is $[(\gtype_1\gseq\gtype\gseq\gtype_2)\gseq(\gtype_1\gseq\gtype_2)\gstar]\gor\gtype$.
 
\section{Proof of Theorem \ref{thm:ap}}\label{ptthm:ap}

\begin{lemma}\label{dl}
The following properties hold:
\begin{enumerate}[\em(1)]
\item\label{dl1} If $\set{\role:\sesst} \scup\cont$ is live, then $\traces(\set{\role:\sesst} \scup\cont)= \traces(\set{\role:\sesst} \scup(\cont\asup\cont'))$ for all $\cont'$ such that $\cont\asup\cont'$ is defined.
\item\label{dl2} If $\set{\role:\sesst_1} \scup\cont_1$ and $\set{\role:\sesst_2} \scup\cont_2$ are live, then $\{ \role : \sesst_1 \intsum \sesst_2 \}  \scup ({\cont_1}\asup{\cont_2})$ is live if defined.
\end{enumerate}
\end{lemma}
\begin{proof}(\ref{dl1}) If $\set{\role:\sesst}\scup\cont$ is live,
  then each output in a session type of $\set{\role:\sesst}\scup\cont$
  has a dual input and therefore the addition of compatible inputs
  cannot change the set of traces.

  (\ref{dl2}) If $\cont_1\asup\cont_2$ is defined, then the types in
  $\cont_1$ and $\cont_2$ for the same participant can only differ on
  inputs, so no new trace can arise in
  $\set{\role:\sesst}\scup(\cont_1\asup\cont_2)$ which was not already
  in $\set{\role:\sesst_1}\scup\cont_1$.
\end{proof}

We use $\substitu$ to range over substitutions of session type variables with closed session types. We extend $\substitu$ to session types and environments in the expected way. %MD0212

\begin{lemma}\label{sf}
If $\substitution\cont$ is live and $ \projecta{\cont}{\gtype}{\cont'}$, then  $\substitution{\cont'}$ is live.
\end{lemma}
\begin{proof}
By induction on the derivation of $ \projecta{\cont}{\gtype}{\cont'}$.  We only consider interesting cases.\\

\noindent
For rule \rulename{AP-Alternative}: \inferrule[]{
    \projecta{\cont}{\gtype_1}{
      \{ \role : \sesst_1 \} \scup \cont_1
    }
    \\
    \projecta{\cont}{\gtype_2}{
      \{ \role : \sesst_2 \} \scup \cont_2
    }
  }{
    \projecta{\cont}{\gtype_1 \gor \gtype_2}{
      \{ \role : \sesst_1 \intsum \sesst_2 \}
      \scup
      (\cont_1\asup\cont_2)
    }
  } we use Lemma \ref{dl}(\ref{dl2}).\\

\noindent
For rule \rulename{AP-Iteration}:\\ \inferrule[]{
    \projecta{
      \{ \role : \recvar \}
      \scup \{ \role_i : \recvar_i \}_{i\in I}
     \scup \cont}{\gtype}{
      \{ \role : \varsesst \}
      \scup \{ \role_i : \varsesst_i \}_{i\in I}
     \scup \cont}
  }{
    \projecta{
      \{ \role : \sesst \}
      \scup \{ \role_i : \sesst_i \}_{i\in I}
      \scup \cont
    }{\gtype\gstar}{
      \{ \role : \rec{\recvar}{(\sesst \intsum \varsesst)} \}
      \scup \{ \role_i : \rec{\recvar_i}{(\sesst_i \asup \varsesst_i)} \}_{i\in I}
      \scup \cont
    }
  }\\
 we define 
\[\begin{array}{lll@{\qquad}lll}
\substitutionZ\recvar&=&\substitution{\sesst}&\substitutionLP\recvar&=&\substitutionL\varsesst\\
\substitutionZ{\recvar_i}&=&\substitution{\sesst_i}&\substitutionLP{\recvar_i}&=&\substitutionL{\varsesst_i}\\
\substitutionZ{Y}&=&\substitution{\varY}\text{ for } \varY\not\in\set{\recvar,\recvar_i\mid i\in I}&\substitutionLP{Y}&=&\substitution{\varY}\text{ for } \varY\not\in\set{\recvar,\recvar_i\mid i\in I}\\
\end{array}\]
for $i\in I$ and $\ell\geq0$. Since $\substitution{\{ \role : \sesst \}
  \scup \{ \role_i : \sesst_i \}_{i\in I}
  \scup \cont}=\substitutionZ{\{ \role : \recvar \}
  \scup \{ \role_i : \recvar_i \}_{i\in I} \scup\cont}$ is live by hypothesis and $\projecta{
  \{ \role : \recvar \}
  \scup \{ \role_i : \recvar_i \}_{i\in I}
}{\gtype}{
  \{ \role : \varsesst \}
  \scup \{ \role_i : \varsesst_i \}_{i\in I}
}
$, by induction we get that $\substitutionZ{\{ \role : \varsesst \}
  \scup \{ \role_i : \varsesst_i \}_{i\in I}       \scup \cont}=\substitutionU{\{ \role : \recvar \}
  \scup \{ \role_i : \recvar_i \}_{i\in I} \scup\cont}$ is live. By iterating this argument we get the liveness of  $\substitutionLP{\{ \role : \recvar \}
  \scup \{ \role_i : \recvar_i \}_{i\in I} \scup\cont}$ from the liveness of  $\substitutionL{\{ \role : \recvar \}
  \scup \{ \role_i : \recvar_i \}_{i\in I} \scup\cont}$ for all $\ell\geq0$. By Lemma \ref{dl}(\ref{dl2})  $\{ \role : \substitutionZ{\recvar}\intsum\cdots\intsum\substitutionL{\recvar} \}
\scup \{ \role_i : \substitutionZ{\recvar_i}\asup\cdots\asup\substitutionL{\recvar_i} \}_{i\in I}
\scup \substitution{\cont}$ is live for all $\ell\geq0$. By construction every finite subtree of $\rec{\recvar}{(\substitution{\sesst \intsum \varsesst})}$ is a subtree of $\substitutionZ{\recvar}\intsum\cdots\intsum\substitutionL{\recvar}$ for some $\ell\geq0$ and  every finite subtree of $\rec{\recvar}{(\substitution{\rec{\recvar_i}{(\sesst_i \asup \varsesst_i)}})}$ is a subtree of $\substitutionZ{\recvar_i}\asup\cdots\asup\substitutionL{\recvar_i}$ for some $\ell\geq0$. We can conclude that $\substitution{\{ \role : \rec{\recvar}{(\sesst \intsum \varsesst)} \}
  \scup \{ \role_i : \rec{\recvar_i}{(\sesst_i \asup \varsesst_i)} \}_{i\in I}
  \scup \cont}$ is live. 
\end{proof}

\begin{lemma}\label{en}
If $ \project{\cont}{\gtype}{\cont'}$ and $\traces(\cont'')=\traces(\cont)$, then  $ \project{\cont''}{\gtype}{\cont'}$.
\end{lemma}
\begin{proof}
  We can derive $ \project{\cont''}{\gend}{\cont''}$, which implies
  $\project{\cont''}{\gend}{\cont}$. Then
  $\project{\cont''}{\gend\gseq\gtype}{\cont'}$, so we conclude $
  \project{\cont''}{\gtype}{\cont'}$.
\end{proof}

%\begin{lemma}\label{weak}
%If $ \projecta{\cont}{\gtype}{\cont'}$, then  $ \projecta{\cont \scup\cont''}{\gtype}{\cont' \scup\cont''}$ for all $\cont''$ such that $\cont'\scup\cont''$ is defined.
%\end{lemma}
%\begin{proof}
%By induction on the derivation of $ \projecta{\cont}{\gtype}{\cont'}$.
%\end{proof}

\begin{proof}[Proof of Theorem \ref{thm:ap}]
We show
\begin{center}{\em
If $\substitution\cont$ is live and $ \projecta{\cont}{\gtype}{\cont'}$, then  $\project{\substitution\cont}{\gtype}{\substitution{\cont'}}$}
\end{center}
by induction on the derivation of $ \projecta{\cont}{\gtype}{\cont'}$.
  
  If the last applied rule is \rulename{AP-Alternative}:
\[ \inferrule{
    \projecta{\cont}{\gtype_1}{
      \{ \role : \sesst_1 \}  \scup \cont_1
    }
    \\
    \projecta{\cont}{\gtype_2}{
      \{ \role : \sesst_2 \}  \scup \cont_2
    }
  }{
    \projecta{\cont}{\gtype_1 \gor \gtype_2}{
      \{ \role : \sesst_1 \intsum \sesst_2 \}
       \scup
      (\cont_1\asup\cont_2)
    }
  }
\]  by induction $\project{\substitution\cont}{\gtype_1}{\substitution{\{ \role : \sesst_1 \}  \scup\cont_1}}$ and    $\project{\substitution\cont}{\gtype_2}{
     \substitution{ \{ \role : \sesst_2 \}  \scup \cont_2}}$.  By Lemma \ref{sf} 
     $\substitution{\{ \role : \sesst_1 \}  \scup\cont_1}$ and $\substitution{\{ \role : \sesst_2 \}  \scup \cont_2}$ are live. 
     By Lemma \ref{dl}(\ref{dl1}) we get $\gless {\substitution{\{ \role : \sesst_1 \}  \scup\cont_1}}{\substitution{\{ \role : \sesst_1 \}  \scup(\cont_1\asup\cont_2)}}$ and 
     $\gless{\substitution{\{ \role : \sesst_2 \}  \scup\cont_2}}{ \substitution{\{ \role : \sesst_2 \}  \scup(\cont_1\asup\cont_2)}}$.
      We can then derive $\project{\substitution\cont}{\gtype_1}{\substitution{\{ \role : \sesst_1 \}  \scup(\cont_1\asup\cont_2)}}$ and    
      $\project{\substitution\cont}{\gtype_2}{ \substitution{\{ \role : \sesst_2 \}  \scup (\cont_1\asup\cont_2)}}$ by rule \rulename{SP-Subsumption}, so we conclude $\project{\substitution\cont}{\gtype_1 \gor \gtype_2}{
     \substitution{ \{ \role : \sesst_1 \intsum \sesst_2 \}
       \scup
      (\cont_1\asup\cont_2)
    }}$ by rule \rulename{SP-Alternative}.
    
   Let the last applied rule be \rulename{AP-Iteration}:
   \[ \inferrule{
    \projecta{
      \{ \role : \recvar \}
       \scup \{ \role_i : \recvar_i \}_{i\in I}
    \scup \cont}{\gtype}{
      \{ \role : \varsesst \}
       \scup \{ \role_i : \varsesst_i \}_{i\in I}
    \scup \cont}
  }{
    \projecta{
      \{ \role : \sesst \}
       \scup \{ \role_i : \sesst_i \}_{i\in I}
       \scup \cont
    }{\gtype\gstar}{
      \cont'       \scup \cont
    }
  }
\]
 where $\cont'=\{ \role : \rec{\recvar}{(\sesst \intsum \varsesst)} \}
       \scup \{ \role_i : \rec{\recvar_i}{(\sesst_i \asup \varsesst_i)} \}_{i\in I}
$.
    If $\substitution{\{ \role : \sesst \}
       \scup \{ \role_i : \sesst_i \}_{i\in I}
       \scup \cont}$ is live, then $\substitution{\cont'       \scup \cont
}$ is live by Lemma \ref{sf}.   
%By Lemma \ref{weak} we get 
%$\projecta{\{ \role : \recvar \}
%       \scup \{ \role_i : \recvar_i \}_{i\in I} \scup\cont
%    }{\gtype}{
%      \{ \role : \varsesst \}
%       \scup \{ \role_i : \varsesst_i \}_{i\in I} \scup\cont
%    }.$
We define 
\[\begin{array}{lll}
\substitutionZ\recvar&=&\substitution{\rec{\recvar}{(\sesst \intsum \varsesst)}}\\
\substitutionZ{\recvar_i}&=&\substitution{\rec{\recvar_i}{(\sesst_i \asup \varsesst_i)}}\\
\substitutionZ{Y}&=&\substitution{\varY}\text{ for } \varY\not\in\set{\recvar,\recvar_i\mid i\in I}
\end{array}\]
Since $\substitutionZ{ \{ \role : \recvar \}
       \scup \{ \role_i : \recvar_i \}_{i\in I} \scup \cont}=\substitution{\cont'       \scup \cont
}$ we get by induction $\project{
      \substitution{\cont'       \scup \cont
}
    }{\gtype}{\substitutionZ{
      \{ \role : \varsesst \}
       \scup \{ \role_i : \varsesst_i \}_{i\in I}  \scup \cont
    }}$. This implies that $\substitutionZ{
      \{ \role : \varsesst \}
       \scup \{ \role_i : \varsesst_i \}_{i\in I}  \scup \cont
    }$ is live by Corollary \ref{sc}.
    We define:
    \[\begin{array}{llllllll}
    \sesst'&=&\substitution\sesst&\qquad&\sesst'_i&=&\substitution{\sesst_i}\\
    \varsesst'&=&\substitutionZ\varsesst&\qquad&\varsesst'_i&=&\substitutionZ{\varsesst_i}
    \end{array}\]
      \[\begin{array}{lll}
      \cont_0&=&\{ \role_i : \sesst'_i \asup \varsesst'_i \}_{i\in I} \scup \substitution{\cont}
      \end{array}\]
      Since $\substitution{\cont'       \scup \cont}=\{ \role :\sesst' \intsum \varsesst' \} \scup\cont_0$ and by Lemma \ref{dl}(\ref{dl1}) 
      $\gless{\{ \role :\varsesst' \}  \scup \{ \role_i : \varsesst'_i \}_{i\in I}  \scup \substitution\cont}{\{ \role :\varsesst' \}  \scup\cont_0}$ we derive
      $\project{
\{ \role :\sesst' \intsum \varsesst' \} \scup\cont_0
    }{\gtype}{\{ \role :\varsesst' \}  \scup\cont_0}$ by rule \rulename{SP-Subsumption}, which implies $\project{
\{ \role :\sesst'  \} \scup\cont_0
    }{\gtype\gstar}{\{ \role :\sesst' \intsum \varsesst'\}  \scup\cont_0}$ by rule \rulename{SP-Iteration}.
    By Lemma \ref{dl}(\ref{dl1}) $\traces(\{ \role :\sesst'  \} \scup\cont_0)=\traces(\{ \role :\sesst' \}  \scup \{ \role_i : \sesst'_i \}_{i\in I}  \scup \substitution\cont)$, so we conclude by Lemma \ref{en} 
    $\project{
\{ \role :\sesst' \}  \scup \{ \role_i : \sesst'_i \}_{i\in I}  \scup \substitution\cont
    }{\gtype\gstar}{\substitution{\cont'       \scup \cont}}$.
\end{proof}
  
%%% Local Variables: 
%%% mode: latex
%%% TeX-master: "main"
%%% End: 

\end{document}